\documentclass{article}
\usepackage[utf8]{inputenc}
\usepackage{secdot, amsfonts, amsmath, amssymb, bm, breqn, mathtools, graphicx, derivative, caption, bbm, dsfont, float, subcaption, tikz, relsize, amsthm, fullpage, array, authblk, algpseudocode, algorithm, algcompatible}
\usepackage{natbib}

\usepackage[export]{adjustbox}
\usepackage{multirow}
\usepackage{color}
\usetikzlibrary{positioning,fit,arrows.meta,calc}
\usetikzlibrary{automata, fit, shapes.geometric, backgrounds} 
\usetikzlibrary{bayesnet}
\usetikzlibrary{arrows}
\tikzset{
  latentnode/.style={draw, minimum width=5mm, shape=circle, scale = 0.6, ultra thick, black, transform shape, minimum size = 1cm},
  dagconn/.style={arrows=->, black, thick},
  plate/.style={draw, shape=rectangle, rounded corners=0.5ex, thick,
    minimum width=2.1cm, text width=3.1cm, align=right, inner sep=10pt, inner ysep=10pt,label={[xshift=-14pt,yshift=14pt]south east:#1}}
}

\usepackage[title]{appendix}

\newtheorem{prp}{Proposition}

\makeatletter
\def\Vdots{\vbox{\baselineskip12\p@ \lineskiplimit\z@
  \kern1\p@\hbox{.}\hbox{.}\hbox{.}\hbox{.}\hbox{.}\hbox{.}\hbox{.}\hbox{.}\hbox{.}}}
\makeatother

\allowdisplaybreaks 
\usepackage[colorlinks=false, linktocpage=true]{hyperref}
\hypersetup{
  colorlinks,
  citecolor=blue,
  linkcolor=blue,
  urlcolor=blue}

\usepackage[capitalise]{cleveref} 
\usetikzlibrary{automata, fit, shapes.geometric}
\crefrangelabelformat{equation}{(#3#1#4--#5#2#6)} 
\usepackage[export]{adjustbox} 
\usepackage[margin=1in]{geometry}
\newcommand{\beginsupplement}{%
        \setcounter{section}{0}
        \renewcommand{\thesection}{\Alph{section}} 
        \setcounter{lemma}{0}
        \setcounter{equation}{0}
        \setcounter{table}{0}
        \renewcommand{\thetable}{S\arabic{table}}%
        \setcounter{figure}{0}
        \renewcommand{\thefigure}{S\arabic{figure}}%
     }

 \makeatletter
\newenvironment{breakablealgorithm}
  {
   \begin{center}
     \refstepcounter{algorithm}
     \hrule height.8pt depth0pt \kern2pt
     \renewcommand{\caption}[2][\relax]{
       {\raggedright\textbf{\fname@algorithm~\thealgorithm} ##2\par}%
       \ifx\relax##1\relax 
         \addcontentsline{loa}{algorithm}{\protect\numberline{\thealgorithm}##2}%
       \else 
         \addcontentsline{loa}{algorithm}{\protect\numberline{\thealgorithm}##1}%
       \fi
       \kern2pt\hrule\kern2pt
     }
  }{
     \kern2pt\hrule\relax
   \end{center}
  }
\makeatother

\DeclareMathOperator*{\argmax}{arg\,max}
\usepackage{setspace}
\begin{document}

\title{\bf Nested Atoms Model with Application to Clustering Big Population-Scale Single-Cell Data}
\date{}

\author[1]{\small Arhit Chakrabarti}
\author[1, 2]{\small Yang Ni}
\author[1, 3, 4]{\small Yuchao Jiang}
\author[1]{\small Bani K. Mallick}

\affil[1]{\footnotesize \textit{Department of Statistics, Texas A\&M University}}
\affil[2]{\footnotesize \textit{Department of Statistics and Data Sciences, The University of Texas at Austin}}
\affil[3]{\footnotesize \textit{Department of Biology, Texas A\&M University}}
\affil[4]{\footnotesize \textit{Department of Biomedical Engineering, Texas A\&M University}}

\maketitle

\begin{abstract}
\sloppy We consider the problem of clustering \emph{nested} or \emph{hierarchical} data, where observations are grouped and there are both group-level and observation-level variables. In our motivating OneK1K dataset, observations consist of single-cell RNA-sequencing (scRNA-seq) data from 982 individuals (groups), totaling 1.27 million cells (observations), along with individual-specific genotype data. This type of data would enable the identification of cell types and the investigation of how genetic variations among individuals influence differences in cell-type profiles. Our goal, therefore, is to jointly cluster cells and individuals to capture the heterogeneity across both levels using cell-specific gene expressions as well as individual-specific genotypes. However, existing grouped clustering methods do not incorporate group-level variables, thereby limiting their ability to capture the heterogeneity of genotypes in our motivating application. To address this, we propose the Nested Atoms Model (NAM), a new Bayesian nonparametric approach that enables the desired two-layered clustering, accounting for both group-level and observation-level variables. To scale NAM for high-dimensional data, we develop a fast variational Bayesian inference algorithm. Simulations show that NAM outperforms existing methods that ignore group-level variables. Applied to the OneK1K dataset, NAM identifies clusters of genetically similar individuals with homogeneous cell-type profiles. The resulting cell clusters align with known immune cell types based on differential gene expression, underscoring the ability of NAM to capture nested heterogeneity and provide biologically meaningful insights.
\end{abstract}
\sloppy \noindent%

\vfill

\newpage
\section{Introduction}
\label{sec:intro}
\subsection{Motivating Application: OneK1K study}
The rapidly decreasing cost and increasing throughput of single-cell RNA-sequencing (scRNA-seq) has enabled researchers to sequence a large number of cells over a large cohort of individuals. Our motivating application is the OneK1K study \citep{OneK1K_Yazar}, which sequenced 1.27 million peripheral blood mononuclear cells from 982 individuals of Northern European ancestry. The scRNA-seq data is also combined with genotype data comprising approximately 760,000 single-nucleotide polymorphisms (SNPs). Such population-scale single-cell data is inherently grouped and nested -- cells are grouped by individuals, and SNPs are individual-specific, and RNA gene expression is cell-specific while further clustered by cellular subtypes. Such data are also known as hierarchical or multi-level data, where heterogeneity exhibits across both cells and individuals. It is worth noting that while we focus our study on the OneK1K dataset, more datasets of this setup are becoming readily available in single-cell population genomics, such as the more recent ROSMAP \citep{ROSMAP} and TenK10K \citep{TenK10K} studies.

Figure~\ref{fig:Biometrics_Motivation} presents an illustrative overview of the motivating OneK1K application. Specifically, individuals exhibiting similar patterns of genetic variations (i.e., SNPs) may also share comparable cell-type-specific profiles. Furthermore, among individuals with distinct genetic variations, partial similarity in cell-type composition and/or cell-type-specific expression may also be expected. While the original OneK1K study \citep{OneK1K_Yazar} first clustered cells into distinct cell types and states at the cell level and then conducted association testing between gene expression and genetic variants at the individual level, our overarching goal is to simultaneously cluster both cells and individuals to achieve a better understanding of heterogeneity at both levels. Importantly, the cell clusters reflect genetically governed molecular cell types, whereas the individual clusters are characterized by both shared SNP profiles and similar cell-type-specific profiles. For generality, we refer to individuals as groups and the cells within each individual as group-specific observations. Using these terminologies, we aim to identify both observation-level and group-level clusters. 

\begin{figure}
\centerline{\includegraphics[width=1\linewidth]{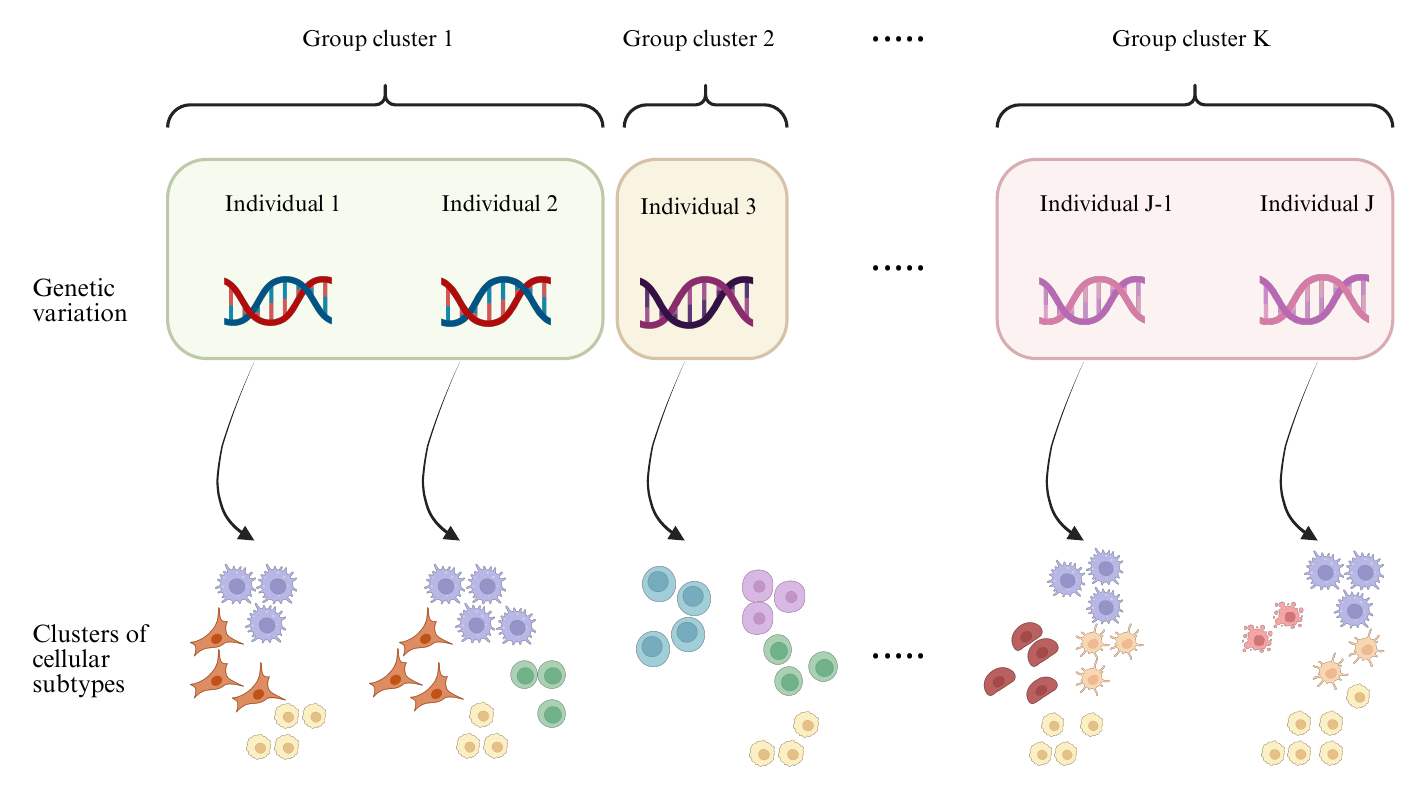}}
\caption{Illustrative figure showing the inherent groups of individuals based on shared genetic variation along with clusters of cell-type distribution within and across such groups in the motivational OneK1K data.
}
\label{fig:Biometrics_Motivation}
\end{figure}
 
\subsection{Literature Review on Grouped Clustering}
Bayesian nonparametric methods have been one of the most popular approaches for achieving such two-layered clustering due to the flexible specification of generative/sampling models for complex data structures, the sharing of observation-level clustering information across groups, and the automatic determination of the number of clusters.    

The hierarchical Dirichlet process (HDP; \citealp{hdp}) assumes group-specific Dirichlet processes (DP; \citealp{ferguson}) with a common DP-distributed base measure. Since the draws from DP are almost surely discrete, all group-specific distributions are based on a shared set of atoms, which leads to the sharing of clusters of observations across groups. However, while the HDP facilitates the clustering of group-specific observations across groups, it does not cluster the groups themselves. The nested DP (nDP; \citealp{nestedDP}) considers a DP-distributed random measure whose base measure is also a DP. The discreteness of the two DPs respectively induces the clustering of groups and group-specific observations. Although nDP has been widely employed in various contexts \citep{Functionalclusteringinnesteddesigns2014, graziani2015bayesian, zuanetti2018clustering}, recently, \cite{latent_nDP} highlighted the inherent degeneracy property of the nDP -- two groups sharing even one observation-level cluster are automatically assigned to the same group-level cluster. To mitigate this degeneracy issue, several recent works have been proposed to combine the cluster-sharing feature of the HDP and the group-clustering feature of the nDP \citep{beraha2021semi, balocchi2022clusteringarealunitsmultiple_arXiv, SCSC_arXiv, hidden_HDP}. 

Notably, the common atoms model (CAM; \citealp{denti2023common}) closely resembles the nDP, but with a crucial difference that avoids the degeneracy issue of the nDP. Particularly, the DP-distributed random measure in nDP is replaced by a discrete measure whose atoms are assumed to be distributions with a common set of atoms. Several variants of CAM have also been proposed, including the finite version of CAM \citep{BNPnestedDP} and the finite-infinite shared atoms nested model (fiSAN; \citealp{fiSAN}). More recently, \citealp{DENTI2025}  extends the CAM's nested framework for a general class of nested random measures with common atoms. The corresponding mixture models thereby achieve a simultaneous clustering of groups, referred to as group clusters (GCs), and clustering of observations across groups, referred to as observational clusters (OCs). 
\subsection{Our Major Contributions}
Although the aforementioned methods have been proven useful for clustering grouped data, they are not suitable for nested data, where there are not only grouped observation-level variables as in grouped data but also group-level variables. For example, in our motivating application, the scRNA-seq gene expressions are observation-level variables grouped by individuals, and SNPs are group-level variables. Therefore, to handle nested data, we propose a nested atoms model (NAM), which allows for the desired two levels of clustering and accommodates both observation- and group-level variables. Importantly, the group-level variables can influence both levels of clustering. We study the properties of the NAM prior and provide a fast variational Bayesian inference algorithm to scale the proposed model to the size of the motivating application. Through extensive simulations, we demonstrate the superior clustering performance of the proposed model in comparison to the existing methods that do not account for the group-level variables. Applying the proposed model to the motivating OneK1K data reveals that individuals with shared genetic variation show homogeneous cell type distribution compared to individuals with contrasting SNP patterns. Furthermore, differential gene expression analysis within each identified observational cluster enabled the identification of marker genes consistent with known cell-type annotations, thereby facilitating biologically interpretable labeling of clusters. By applying existing methods to the same dataset, we also demonstrate that disregarding the SNP data leads to suboptimal clustering performance at both the group- and observation-levels. These findings underscore the practical utility and the necessity of NAM in disentangling complex structures in nested data.

The remainder of the article is organized as follows. In Section~\ref{sec::preliminaries}, we provide a brief overview of the existing Bayesian nonparametric models for clustering grouped data. In Section~\ref{sec::NAM}, we introduce the proposed model for clustering nested data. Additionally, we study some properties of the proposed model. Section~\ref{sec::posterior_inference} provides the variational inference algorithm. We present simulation studies in Section~\ref{sec::simulations}, showcasing the superior two-layered clustering performance of the proposed method over existing methods. In Section~\ref{sec:real_data_analysis}, we apply the proposed model to the motivating OneK1K data. The paper concludes with a brief discussion in Section~\ref{sec::conclusion}. The code implementing the proposed method is available in the GitHub repository: \url{https://github.com/Arhit-Chakrabarti/NAM}. A glossary of all abbreviations used throughout the paper is provided in Section~\ref{supp::sec::glossary} of the Supplementary Materials.

\section{Backgrounds}\label{sec::preliminaries}
We present a brief overview of four Bayesian nonparametric models -- nDP, CAM, fiSAN, and the hidden hierarchical Dirichlet process (HHDP; \citealp{hidden_HDP}) -- for clustering grouped data. 
For the $j$-th individual ($j=1,\dots,J$), let $\bm{y}_{ji}$ denote the vector of $p$ gene expressions measured in cell $i=1,\dots, n_j$ and let $\bm{x}_j$ denote the vector of $q$ SNPs (in fact, later we will use the principal components). We call $\{\{\bm{y}_{ji}\}_{i=1}^{n_j}\}_{j=1}^J$ grouped data and $\{\bm{x}_j,\{\bm{y}_{ji}\}_{i=1}^{n_j}\}_{j=1}^J$ nested data as $\bm{x}_j$ and $\bm{y}_{ji}$ have different granularity. The models that we review in this section are only applicable to the grouped data $\{\{\bm{y}_{ji}\}_{i=1}^{n_j}\}_{j=1}^J$.

\subsection{Nested Dirichlet Process (nDP)}
Assume $\bm{y}_{ji}$ takes values in a suitable Polish space $\mathbb{Y}$ endowed with the respective Borel $\sigma$-field $\mathcal{Y}$. Let $G_j$ denote the random probability measure or distribution of the $j$-th group,
\begin{equation}
\begin{aligned}
\bm{y}_{ji} \overset{iid}{\sim} G_j\quad \text{for} \quad i=1,\dots,n_j,
\end{aligned}
\label{eq::nDP1}
\end{equation}
where iid denotes that the variables are independent and identically distributed.
Let $\bm{y}_j=\{\bm{y}_{ji}\}_{i=1}^{n_j}$. nDP simultaneously partitions $\bm{y}_1, \ldots, \bm{y}_J$ into group clusters (GCs) and partitions $\bm{y}_{j1},\ldots,\bm{y}_{jn_j}$ into observational clusters (OCs) within each such GC by assuming that the $G_j$'s are a sample from an almost surely discrete distribution over the space of probability distributions on $(\mathbb{Y}, \mathcal{Y})$,
\begin{equation}
\begin{aligned}
G_j\overset{iid}{\sim}Q=\sum_{k=1}^{\infty}\,  \pi_k \, \delta_{G^*_k },
\end{aligned}
\label{eq::nDP2}
\end{equation}
with
\begin{equation}
G^*_k = \sum_{l = 1}^{\infty} \omega_{ lk}\, \delta_{\theta_{lk}},
\label{eq::nDP3}
\end{equation}
where $\delta$ is a point mass, the atoms $\theta_{lk}$'s are iid samples from a non-atomic base probability measure $H$ on $(\mathbb{Y}, \mathcal{Y})$, and the weights $\bm{\pi}=\{\pi_k\}_{k=1}^{\infty}\sim \mathrm{GEM}(\alpha)$ and $\bm{\omega}_k = \{\omega_{lk}\}_{l=1}^{\infty}\sim \mathrm{GEM}(\beta)$ are GEM-distributed, where GEM stand for the Griffiths-Engen-McCloskey distribution \citep{pitman_GEM}. In other words, $\pi_k = v_k\prod_{j=1}^{k-1}(1-v_j)$ and $\omega_{lk} = u_{lk}\prod_{j=1}^{l-1}(1-u_{jk})$ with $v_k \overset{iid}{\sim} \mathrm{Beta}(1, \alpha)$ and $u_{lk} \overset{iid}{\sim} \mathrm{Beta}(1, \beta)$ for all $k\geq 1$,  $l\geq 1$ \citep{Sethuraman_arXiv}.  Consequently, $Q,G_1^*,G_2^*,\dots$ are all DP-distributed. It has been shown \citep{latent_nDP} that two distributions $G_j$ and $G_{j'}$ sharing even one atom are automatically assigned to the same GC as a consequence of \eqref{eq::nDP3}.

\subsection{Common Atoms Model (CAM)}
To overcome the degeneracy property of the nDP, \cite{denti2023common} proposed the CAM. Their specification closely resembles that of the nDP, but with \eqref{eq::nDP3} replaced by,
\begin{equation}
G^*_k = \sum_{l = 1}^{\infty} \omega_{ lk}\, \delta_{\theta_{l}},
\label{eq::CAM}
\end{equation}
where the common atoms $\theta_{l} \overset{iid}{\sim} H$, a non-atomic base measure on $(\mathbb{Y}, \mathcal{Y})$. However, under the CAM, the prior correlation between two random distributions $G_j$ and $G_{j'}$ ($j \neq j'$) lies between 0.5 and 1 by construction. As noted by \citealp{fiSAN}, such strong prior dependence may lead to biased posterior inference when the underlying distributional clusters are well separated. To address this issue, the authors proposed the fiSAN model by modifying the distributional assumptions on the observational weights, which we review next.

\subsection{Finite-Infinite Shared Atoms Nested Model (fiSAN)}
As a finite variation of CAM, fiSAN fixes the number of OCs by modifying \eqref{eq::CAM} as,
\begin{equation}
G^*_k = \sum_{l = 1}^{L} \omega_{lk}\, \delta_{\theta_{l}},
\label{eq::fiSAN}
\end{equation}
where the mixing weights are assigned the Dirichlet prior distribution,
\begin{equation}
    \begin{aligned}
    \bm{\omega}_k = (\omega_{1k},\ldots,\omega_{Lk}) &\sim \mathrm{Dirichlet}_L(\beta, \ldots, \beta).
    \label{eq::fiSAN_weights}
    \end{aligned}
\end{equation}
The authors show that adopting a symmetric Dirichlet distribution implies that, a priori, all atoms are equally likely to be resampled across the random measures. Consequently, the induced prior correlation structure under the fiSAN lies in $(0,1)$, thereby allowing a broader range of dependence and leading to improved prior flexibility and posterior inference. Next, we provide a brief review of an alternative yet related model, the HHDP. In particular, the HHDP explicitly integrates the cluster-sharing mechanism of the HDP with the group-clustering structure of the nDP. Despite some initially apparent structural differences, the HHDP, similar to the CAM and fiSAN, ultimately relies on a common set of atoms.

\subsection{Hidden hierarchical Dirichlet process (HHDP)}
The HHDP can be hierarchically expressed as \eqref{eq::nDP2}, where the random measures $G_k^*$ are themselves discrete and admit the representation
\begin{equation}
G_k^* = \sum_{l=1}^{\infty} \omega_{lk}\,\delta_{Z_{lk}}, 
\qquad 
Z_{lk} \mid G_0 \overset{iid}{\sim} G_0,
\qquad 
G_0 = \sum_{\ell=1}^{\infty} \omega_{\ell0}\,\delta_{\theta_\ell}.
\label{eq::HHDP2}
\end{equation}
Here, the atoms $\theta_\ell$ are independently drawn from a non-atomic base measure $H$ defined on $(\mathbb{Y}, \mathcal{Y})$. Consequently, the atoms $Z_{lk}$ are sampled from the discrete random measure $G_0$, thereby inducing sharing of atoms across the measures $\{G_k^*\}_{k \geq 1}$. The observational weights are further assumed to follow
\[
\{\omega_{lk}\}_{l=1}^{\infty} \overset{iid}{\sim} \mathrm{GEM}(\beta), \; k\ge 1,
\qquad 
\{\omega_{\ell0}\}_{\ell=1}^{\infty} \sim \mathrm{GEM}(\beta_0).
\]
By the closure property of the DP, the representation in \eqref{eq::HHDP2} can be equivalently expressed using a common set of atoms,
\begin{equation}
G_k^* = \sum_{l=1}^{\infty} \omega_{lk}^*\,\delta_{\theta_l}, 
\qquad 
G_0 = \sum_{l=1}^{\infty} \omega_{l0}\,\delta_{\theta_l},
\label{eq::HHDP3}
\end{equation}
where
\[
\{\omega_{lk}^*\}_{l=1}^{\infty} \mid \bm{\omega}_0 
\overset{iid}{\sim} \mathrm{DP}(\beta, \bm{\omega}_0), 
\qquad 
\{\omega_{l0}\}_{l=1}^{\infty} \sim \mathrm{GEM}(\beta_0),
\qquad 
\theta_l \overset{iid}{\sim} H.
\]
The key difference between the distributional assumptions underlying the HHDP and CAM lies in the dependence structure of the component weights. Under the CAM, the weight sequences $\{\omega_{lk}\}_{l\ge1}$ are independent across groups, implying that each group independently determines the relative importance of the common atoms $\{\theta_l\}_{l\ge1}$. Consequently, there is no mechanism encouraging similarity in cluster proportions across groups, and the same clusters may appear with markedly different weights in different groups. In contrast, under the HHDP the group-specific weights arise as Dirichlet process perturbations around a base weight distribution. As a result, clusters exhibit a form of global popularity governed by the base measure $\bm{\omega}_0$, and individual groups inherit this structure while allowing for group-specific variability.

\section{Nested Atoms Model}\label{sec::NAM} Since existing methods are only applicable to grouped data,
clustering nested data  $\{\bm{x}_j, \{\bm{y}_{ji}\}_{i=1}^{n_j}\}_{j=1}^J$ calls for a new Bayesian nonparametric model. Assume that each $\bm{x}_j$ takes values in a suitable Polish space $\mathbb{X}$ endowed with the respective Borel $\sigma$-field $\mathcal{X}$. 
We propose the following NAM model,
\begin{equation}
\begin{aligned}
\bm{x}_j\overset{ind}{\sim} G_j^x\quad\text{and}\quad \bm{y}_{ji}\overset{iid}{\sim} G_j^y,
\end{aligned}
\label{eq::NAM1}
\end{equation}
where ind denotes that the variables are independently distributed,
\begin{equation}
\begin{aligned}
G_j:=(G_j^x, G_j^y)&\overset{iid}{\sim}Q=\sum_{k=1}^\infty\,  \pi_k \, \delta_{(\delta_{\theta_k^x}, G^{y*}_k)},
\end{aligned}
\label{eq::NAM2}
\end{equation}
and
\begin{equation}
G^{y*}_k = \sum_{l = 1}^{\infty} \omega_{ lk}\, \delta_{\theta_{l}^y}.
\label{eq::NAM3}
\end{equation}
Here the atoms $\theta_k^x$ associated with $\bm{x}_j$  are drawn from a base measure $H^x$ on $(\mathbb{X}, \mathcal{X})$, while the common atoms $\theta_l^y$ associated with $\bm{y}_{ji}$ are independently sampled from $H^y$ on $(\mathbb{Y}, \mathcal{Y})$. Both base measures $H^x$ and $H^y$ are assumed to be non-atomic. The fact that $\delta_{\theta_k^x}$ is coupled with $G_k^{y*}$ in \eqref{eq::NAM2} suggests that $\bm{x}_j$ and $\bm{y}_{ji}$ 
 would jointly determine both OCs and GCs. Additionally, the inclusion of common atoms $\theta_l^y$ facilitates the sharing of OCs across groups, thereby enabling the \emph{comparable clusters} feature of the NAM. Furthermore, we assume $\bm{\pi} =\{\pi_k\}_{k=1}^\infty\sim \mathrm{GEM}(\alpha)$ and $\bm{\omega}_k =\{\omega_{lk}\}_{l=1}^\infty\sim \mathrm{GEM}(\beta)$ for all $k$. We refer to the model given by \eqref{eq::NAM2} and \eqref{eq::NAM3} as the NAM. Because, marginally,
\begin{equation}
\begin{aligned}
G_j^y&\overset{iid}{\sim} Q^y= \sum_{k= 1}^\infty \pi_k \delta_{G_k^{y*}},
\end{aligned}
\label{eq::NAM_Marginal_for_x2}
\end{equation}
the proposed NAM reduces to CAM in the absence of $\bm{x}_j$. Furthermore, as marginally, 
\begin{equation}
\begin{aligned}
G_j^x &\overset{iid}{\sim} Q^x= \sum_{k= 1}^\infty \pi_k \delta_{\delta_{\theta_k^x}}
\end{aligned}
\label{eq::NAM_Marginal_for_x}
\end{equation}
implies $P(G_j^x = \delta_{\theta_k^x}|Q^x) = \pi_k$, it is straightforward to see for all $k\geq 1$, 
\begin{align*}
    P(\bm{x}_j = \theta_k^x |Q^x) & = \sum_{l=1}^{\infty} P(\bm{x}_j = \theta_k^x, G_j^x = \delta_{\theta_l^x} |Q^x)\\
    & = P(\bm{x}_j = \theta_k^x, G_j^x = \delta_{\theta_k^x} |Q^x)\\
    & = P(G_j^x = \delta_{\theta_k^x} |Q^x)\\ 
    & = \pi_k.
\end{align*}
In other words, marginally $\bm{x}_j \overset{iid}{\sim}Q^x$, where $Q^x \sim \mbox{DP}(\alpha, H^x)$.

We note that our proposed construction of the NAM adopts a weight specification similar to that of the CAM, in that the sequences $\{\omega_{lk}\}_{l\ge1}$ are assumed to follow independent $\mathrm{GEM}(\beta)$ distributions across all groups $k\geq 1$. This modeling choice is particularly appealing for our motivating OneK1K cohort, where it is desirable to allow cell-type clusters to be shared across individuals while permitting their relative abundances to vary across groups. At the same time, even among individuals with dissimilar SNP profiles, the proportions of the same cell-types may remain comparable due to underlying biological mechanisms. Consequently, modeling the group-specific weight sequences independently provides sufficient flexibility to capture heterogeneous cluster prevalence across groups, while the SNP-derived data possibly inform the similarity among individuals through the group-level variables.

\subsection{Properties of the NAM}
In this section, we study some properties of the NAM. In particular, we investigate the prior mean, variance, co-clustering probabilities, and correlation structure. In the following propositions, we assume that the concentration parameters, $\alpha$ and $\beta$ are fixed. All proofs are presented in Section~\ref{sec::suppl_proofs_properties} of the Supplementary Materials. The first proposition provides the expressions for the prior mean and variance for a random measure under the NAM prior.
\begin{prp}\label{prop1}
Consider the random distribution $G_j$ defined on $(\mathbb{X}\times \mathbb{Y}, \mathcal{X}\otimes\mathcal{Y})$, with $G_j | Q \sim  Q$, where $Q$ is defined in \eqref{eq::NAM2}. Then, for any Borel set $A \in \mathcal{X}\otimes\mathcal{Y}$,
\begin{align*}
    \mathbb{E}\left[ G_j(A)\right] &= H^x(A)H^y(A),\\
    Var\left[ G_j(A)\right] & = H^x(A)H^y(A)\left[q_2 + H^y(A) - q_2H^y(A) - H^x(A) H^y(A) \right],
\end{align*}
where $q_2 = 1/{(1+ \beta)}$.
\end{prp}
Both the prior mean and variance of the NAM depend on the base measures $H^x(\cdot)$ and $H^y(\cdot)$. The next proposition states the prior co-clustering probabilities under the NAM prior.
\begin{prp}\label{prop2}
Consider two random distributions $G_j$ and $G_{j'}, \ j' \neq j$, defined on $(\mathbb{X}\times \mathbb{Y}, \mathcal{X}\otimes\mathcal{Y})$, with $G_j,\,G_{j'} \mid Q \overset{iid}{\sim} Q$ and $Q$ is defined by \eqref{eq::NAM2}. Additionally, consider two observations $\bm{z}_{ji} \mid G_j \sim G_j$ and $\bm{z}_{j'i'} \mid G_{j'} \sim G_{j'}$, where $i' \neq i$, $\bm{z}_{ji}= (\bm{x}_{j}, \bm{y}_{ji})$, and $\bm{z}_{j'i'}= (\bm{x}_{j'}, \bm{y}_{j'i'})$. Then, the prior co-clustering probabilities under the NAM are given by,
\begin{align*}
    P\left[G_j = G_{j'}\right] & = \frac{1}{1+ \alpha},\\
    P\left[\bm{z}_{ji} = \bm{z}_{j'i'}\right] & = \frac{1}{1+\alpha}\left[\frac{1}{1+\beta} + \frac{\alpha}{2\beta + 1}\right].
\end{align*}
\end{prp}
Proposition~\ref{prop2} highlights that both the group-level and observation-level co-clustering probabilities of the NAM are equivalent to those of the CAM. A larger $\alpha$ leads to smaller group-level co-clustering probabilities and hence more GCs whereas a larger $\beta$ leads to smaller observation-level co-clustering probabilities and hence more OCs. Although the nested data structure includes additional group-level variables, the prior co-clustering probabilities are determined by the stick-breaking weights rather than the atoms, and are therefore identical to those in CAM. If one wishes to impose a stronger dependence of the group-level variables on the group and/or observation levels, such dependence could instead be incorporated into the stick-breaking weights. This modification would possibly yield a different structure for the prior co-clustering probabilities, which we do not pursue in the present work. Regardless, the prior correlation structure of the NAM depends on group-level variables, which we study here. 
As noted in Section~\ref{sec::preliminaries}, under the CAM, the prior correlation lies between $0.5$ and $1$ by construction, which may be restrictive \citep{fiSAN}. In contrast, the prior correlation induced by the fiSAN lies in $(0,1)$, thereby allowing a more flexible alternative. The following proposition  provides the expression for the correlation structure between two random measures under the NAM prior.
\begin{prp}
Consider two random distributions $G_j$ and $G_{j'}, \ j' \neq j$, defined on $(\mathbb{X}\times \mathbb{Y}, \mathcal{X}\otimes\mathcal{Y})$, with $G_j,\,G_{j'} \mid Q \overset{iid}{\sim} Q$ and $Q$ is defined by \eqref{eq::NAM2}. Then, the prior correlation between the two random measures evaluated on the same Borel set $A \in \mathcal{X}\otimes\mathcal{Y}$ is given by,
\begin{align*}
    \rho_{jj'}^{NAM}(A) = q_1 + \frac{q_3(1-q_1) H^x(A)}{\left[q_2 + \frac{H^y(A)}{(1-H^y(A))} (1- H^x(A))\right]},
\end{align*}
where $q_1 = 1/{(1+\alpha)}, \: q_2 = 1/(1+\beta),$ and $q_3 = 1/(1+ 2\beta)$.
\end{prp}
It is worthwhile to note that for the NAM, unlike many other Bayesian nonparametric priors, the prior correlation structure depends on the set $A$, indicating \emph{non-stationarity}. Rearranging the prior correlation, we obtain
$$\rho_{jj'}^{NAM}(A) = q_1 + \frac{q_3(1-q_1)}{\left[\frac{q_2}{H^x(A)} + {\Big\{}\frac{H^y(A)}{1-H^y(A)}{\Big/}\frac{H^x(A)}{1-H^x(A)}{\Big\}}\right]},$$
which depends on the ``odds-ratio'' of the set A under the base measure corresponding to the atoms of $G_j^y$ compared to the base measure corresponding to the atoms of $G_j^x$, i.e., $OR^{y,x}(A) = \{H^y(A)/(1-H^y(A))\}/\{H^x(A)/(1-H^x(A))\}$. It is easy to check that the term, $OR^{y,x}(A)$ 
is increasing in $H^x(A)$ for fixed $H^y(A)$.
Intuitively, for a set $A$, if the prior base measure associated with $G_j^x$ assigns greater mass, $H^x(A)$, then the prior correlation between the two random measures is higher over that set; conversely, when $H^x(A)$ is smaller, the prior correlation is lower. In other words, in regions where $\bm{x}_j$ and $\bm{x}_{j'}$ are ``similar'', the prior induces stronger correlation between the random measures $G_j$ and $G_{j'}$.
This property is particularly relevant in our motivating OneK1K application, where similarity in the SNP profiles is a possible factor in defining the GCs. Furthermore, in realistic biological settings, individuals with similar SNP profiles may nonetheless exhibit substantially different cell-type compositions due to environmental, developmental, disease-related, or epigenetic factors, while genetically distinct individuals may display similar cell-type profiles owing to shared biological processes. Accordingly, in our modeling framework we assume a priori that individuals with similar SNP profiles are more likely to belong to the same GC. However, a posteriori, GC assignments are determined by both the SNP data and the gene expression data, allowing individuals with similar SNP profiles to belong to different GCs and genetically distinct individuals to belong to the same GC.
Additionally, if $H^x(A) = 1$ for any non-null set $A$, signifying that $\bm{x}_j$ is deterministic, then the NAM prior is equivalent to the CAM, and subsequently, $\rho_{jj'}^{NAM}(A) = 1 - \alpha\beta/\{(1 + 2\beta)(1+\alpha)\} $ is exactly the prior correlation structure corresponding to the CAM.

Furthermore, for any Borel set $A$ such that $H^x(A) \neq 0, 1$ and $H^y(A) \neq 0, 1$, it is straightforward to see that, for fixed $\beta > 0$,
\begin{equation}
    \rho_{jj'}^{NAM}(A) \rightarrow \begin{cases} 1 &  \text{as $\alpha \rightarrow 0$,}\\
    \frac{q_3H^x(A)}{q_2 + \frac{H^y(A)}{1- H^y(A)}\{1-H^x(A)\}} &\text{as $\alpha \rightarrow \infty$,}
    \end{cases}
\end{equation}
and for fixed $\alpha > 0$,

\begin{equation}
    \rho_{jj'}^{NAM}(A) \rightarrow \begin{cases} q_1 + (1-q_1) \left[\frac{H^x(A) (1- H^y(A))}{1 - H^x(A)H^y(A)} \right] &  \text{as $\beta \rightarrow 0$,}\\
    q_1 &\text{as $\beta \rightarrow \infty$.}
    \end{cases}
\end{equation}
Additionally, as $\alpha \rightarrow \infty, \beta\rightarrow\infty$, then $\rho^{NAM}_{jj'}(A) \rightarrow 0$ and as $\alpha \rightarrow 0, \beta\rightarrow 0$, then the prior correlation, $\rho^{NAM}_{jj'}(A) \rightarrow 1$. In other words, unlike the CAM, the prior correlation structure under the NAM varies between 0 and 1.
\subsection{NAM Mixture Model}\label{subsec::NAM_MixtureModel}
The proposed NAM defined through equations \eqref{eq::NAM1} - \eqref{eq::NAM3} is almost surely discrete and is hence not suitable for modeling continuous distributions. We will use it as a prior distribution or a mixing measure of a mixture model. Specifically, we assume that
\begin{align*}
&\bm{x}_j\sim\int p^x(\bm{x}_j|\theta^x)G_j^x(d\theta^x),\\
&\bm{y}_{ji}\sim\int p^y(\bm{y}_{ji}|\theta^y)G_j^y(d\theta^y),
\end{align*}
where $G_j^x$ and $G_j^y$ are specified in \eqref{eq::NAM2} - \eqref{eq::NAM3} and the mixture kernels $p^x(\cdot|\theta^x)$ and $p^y(\cdot|\theta^x)$ will be specified later. For simplicity, we have assumed $\bm{x}_j$ and $\bm{y}_{ji}$'s to be conditionally independent given $G_j$. But if desired, we could easily generalize it to have a stronger dependence by using a conditional kernel $p^{y|x}(\cdot|\theta^{y|x},\bm{x}_j)$, which could be an expression quantitative trait loci model \citep{cheung2002genetics,stranger2007population,coulter2025distqtl} in our application. 

Borrowing the same notations as in \citealp{fiSAN}, we introduce an alternative representation of the proposed NAM mixture model using two sequences of latent variables, $\bm{S}=\{S_j\}_{j\geq 1}$ and $\bm{M}=\{M_{ji}\}_{i \geq 1,j\geq 1}$, describing respectively the clustering at the group-level and the observation-level. Specifically, $S_j=k$ means that group $j$ belongs to the $k$-th GC, and $M_{ji}=l$ means that observation $i$ in group $j$ belongs to the $l$-th OC. 
Letting $\bm{\omega} = \{\bm{\omega}_k\}_{k \geq 1}$, $\boldsymbol{\theta}^y=\{\theta_l^y\}_{l\geq 1}$, and $\boldsymbol{\theta}^x=\{\theta_k^x\}_{k\geq 1}$, we have the following equivalent representation of the proposed NAM mixture model:
\begin{equation}
\begin{aligned}
\bm{y}_{ji}|\bm{M}, \bm{\theta}^y & \sim  p^{y}\left(\cdot | \theta_{M_{ji}}^y\right), 
&M_{ji}|\bm{S,\omega}  &\sim \sum_{l=1}^{\infty} \omega_{lS_j} \delta_l(\cdot),\\
\bm{\omega}_k\, |\, \alpha &\sim GEM(\alpha), 
&S_{j}|\bm{\pi}  &\sim \sum_{k=1}^{\infty} \pi_{k} \delta_k(\cdot),\\
\bm{\pi}\, | \, \beta &\sim GEM(\beta), 
& \bm{x}_j |S_j, \bm{\theta}^x&\sim p^x(\cdot|\theta_{S_j}^x), \\
\theta^y_{l} & \sim h^y(\cdot),\:\: l\geq1, & \theta^x_{k} & \sim h^x(\cdot),\:\: k\geq1,
\end{aligned}
\label{Membership}
\end{equation}
 where $h^x(\cdot)$ and $h^y(\cdot)$ denote the densities corresponding to the base measures $H^x$ and $H^y$, respectively. The Figure~\ref{fig:NAM_Graphical_Rep} in the Supplementary Material shows the graphical model representation of the NAM mixture model, which further highlights the desired simultaneous clustering of groups and observations within groups. Furthermore, following \cite{ClusteringconsistencywithDPM}, we assume non-informative gamma priors on the concentration parameters $\alpha$ and $\beta$ to possibly mitigate issues related to inconsistencies in the estimation of the number of clusters in Bayesian nonparametric mixture models \citep{miller_harrisonDP, miller_harrisonPY_arXiv, yang2020posterior_arXiv}. In particular, under the NAM mixture model, the group-level latent parameters $\bm{\theta}^x_j$ admit the marginal representation $\bm{\theta}^x_j \overset{iid}{\sim} Q^x$, where $Q^x \sim \mathrm{DP}(\alpha, H^x)$. This representation motivates the use of a non-informative Gamma prior for the concentration parameter $\alpha$ following one possible recommendations by \citealp{ClusteringconsistencywithDPM}. For the parameter $\beta$, although an analogous marginal DP interpretation is not available, we adopt a non-informative Gamma prior for coherence and with the aim of similarly alleviating potential inconsistency concerns.

\section{Scalable Posterior Inference}\label{sec::posterior_inference}
The posterior distribution for NAM is not available in closed form. The standard posterior inference approach is Markov chain Monte Carlo, which has limited scalability for large datasets such as our motivating data, which have a total of 1.27 million single cells from 982 individuals. We develop a variational posterior inference (VI) algorithm, by extending the grouped-data VI framework of \citealp{fiSAN} to our nested data setting. VI reformulates Bayesian inference as an optimization problem and aims to find, among a set of simple distributions, called variational distributions, the one that minimizes the Kullback-Leibler divergence from the posterior distribution. This is equivalent to maximizing the evidence lower bound (ELBO). We adopt a mean-field approximation, under which the variational distributions are assumed to factorize across latent variables. Furthermore, we assume a multivariate Gaussian likelihood, although it can be replaced by any exponential family distribution, as discussed in \cite{BleiJordan2006}. \par

We assume $p^x(\cdot\mid \theta^x) = \mathcal{N}_q(\cdot\mid \bm{\mu}^x,(\bm{\Lambda}^{x})^{-1})$, with $\bm{\mu}^x$ a $q$-dimensional mean vector and $\bm{\Lambda}^x$ a $q\times q$ precision matrix. We assume a conjugate normal-Wishart prior distribution on the parameters, $(\bm{\mu}^x,\bm{\Lambda}^x) \sim \mathrm{NW}(\bm{\mu}^x_0,\lambda^x_0,\nu_0^x,\boldsymbol{\Psi}^x_0)$. Likewise, we assume $p^{y}(\cdot\mid \theta^y) = \mathcal{N}_p(\cdot\mid \bm{\mu}^y,(\bm{\Lambda}^{y})^{-1})$, where $\bm{\mu}^y$ and $\bm{\Lambda}^y$ are the $p$-dimensional mean vector and the $p\times p$ precision matrix, respectively. 
As mentioned before, if a more direct dependence between $\bm{y}_{ji}$ and $\bm{x}_j$ is desired, we could consider, e.g., $p^{y|x}(\cdot\mid \theta^y,\bm{x}_j) = \mathcal{N}_p(\cdot\mid \bm{B}^T\bm{x}_j,(\bm{\Lambda}^{y})^{-1})$, where $\bm{B}$ is a $q \times p$ matrix of regression coefficients.
We assume a conjugate normal-Wishart prior distribution on the parameters, $(\bm{\mu}^y,\bm{\Lambda}^y) \sim \mathrm{NW}(\bm{\mu}^y_0,\lambda^y_0,\nu_0^y, \boldsymbol{\Psi}^y_0)$. Furthermore, we assume non-informative gamma priors on the concentration parameters, i.e., $\beta \sim \mathrm{Gamma}(a_{\beta}, b_{\beta})$ and $\alpha \sim \mathrm{Gamma}(a_{\alpha}, b_{\alpha})$. To derive the proposed algorithm, we exploit the model formulation based on the data augmentation scheme introduced in Section~\ref{sec::NAM}, which makes use of the cluster allocation variables $S_j$, $j=1,\dots,J$, and $M_{ji}$, $i=1,\dots, n_j$, $j=1,\dots, J$.
Thus, we can write the model as 
\begin{equation*}
\begin{aligned}
    p(\bm{S}\mid \bm{\pi}) = \prod_{j=1}^J\prod_{k=1}^{\infty} {\pi_{k}}^{\mathbbm{1}({S_{j}=k)} }, & \quad \:\:\: p(\bm{M}\mid \bm{S},\bm{\omega}) = \prod_{j=1}^J\prod_{i=1}^{n_j}\prod_{k=1}^{\infty}\prod_{l=1}^{\infty} {\omega_{lk}}^{\mathbbm{1}({M_{ji}=l \: \cap \: S_{j}=k})},\\ 
    \qquad \qquad p^x(\bm{x} \mid \bm{S}, \{\bm{\mu}^x_k,\bm{\Lambda}^x_k\}_{k=1}^\infty) & = \prod_{j=1}^J\prod_{k=1}^{\infty} \mathcal{N}_q(\bm{x}_j\mid \bm{\mu}^x_k, (\bm{\Lambda}^x_k)^{-1})^{\mathbbm{1}({S_j=k})}, \\
      \qquad \qquad  p^{y}(\bm{y} \mid \bm{M}, \{\bm{\mu}^y_l,\bm{\Lambda}^y_l\}_{l=1}^\infty) & = \prod_{j=1}^J\prod_{i=1}^{n_j}\prod_{l=1}^{\infty} \mathcal{N}_p(\bm{y}_{ji}\mid \bm{\mu}^y_l, (\bm{\Lambda}^y_l)^{-1})^{\mathbbm{1}({M_{ji}=l})}. \\
    \label{eq::memb_model}
\end{aligned}
\end{equation*}
Following~\cite{BleiJordan2006}, we use truncated variational families to deal with the nonparametric mixtures at the group and observation level, where the truncation levels are denoted by $K$ and $L$, respectively, which are set to 30 in simulations and 50 in the real data analysis. Note that $K$ and $L$ are only the upper bounds for the numbers of GCs and OCs. We can always increase them if the upper bounds are ever reached. Specifically, the variational distribution is given by,
\begin{flalign*}
\begin{aligned}
     & q(\bm{S}, \bm{M}, \bm{v}, \bm{u}, \{\bm{\mu}^x_k, \bm{\Lambda}^x_k\}_{k=1}^{K}, \{\bm{\mu}^y_l, \bm{\Lambda}^y_l\}_{l=1}^{L}, \alpha, \beta ; \boldsymbol{\Lambda}) \\
     & \quad = \prod_{j=1}^J q(S_j; \{\rho_{jk}\}_{k=1}^K ) \: \prod_{j=1}^J \prod_{i=1}^{n_j} q(M_{ji}; \{\xi_{jil}\}_{l=1}^{L})   \prod_{k=1}^{K-1} q(v_{k} ; \bar{a}_{k},\bar{b}_{k})\: \prod_{k=1}^K \prod_{l=1}^{L-1} q(u_{lk}; \bar{a}_{lk}, \bar{b}_{lk})\times \\ 
     & \qquad \times \prod_{k=1}^K q(\bm{\mu}^x_k,\bm{\Lambda}^x_k ; \bm{m}^x_k, t^x_k, c^x_k, \boldsymbol{D}^x_k)\: \prod_{l=1}^L q(\bm{\mu}^y_l,\bm{\Lambda}^y_l ; \bm{m}^y_l, t^y_l, c^y_l, \boldsymbol{D}^y_l) \times q(\alpha ; s_1,s_2)\: q(\beta ; r_1,r_2),
\end{aligned}
\end{flalign*}
\sloppy where $q(S_j; \{\rho_{jk}\}_{k=1}^K )$ and $q(M_{ji}; \{ \xi_{jil} \}_{l=1}^{L} )$ are multinomial distributions; $q(v_{k} ; \bar{a}_{k},\bar{b}_{k})$ are beta distributions, and they are such that $q(v_{K}=1)=1$ and $q(v_{g}=0)=1$ for $g>K$; $q(u_{lk}; \bar{a}_{lk}, \bar{b}_{lk})$ are beta distributions such that, for all $k=1,\dots,K$,  $q(u_{Lk}=1)=1$ and $q(u_{hk}=0)=1$ for $h>L$; $q(\bm{\mu}^x_k,\bm{\Lambda}^x_k ; \bm{m}^x_k, t^x_k, c^x_k, \boldsymbol{D}^x_k)$ and $q(\bm{\mu}^y_l,\bm{\Lambda}^y_l ; \bm{m}^y_l, t^y_l, c^y_l, \boldsymbol{D}^y_l)$ are normal-Wishart distributions; $q(\alpha ; s_1,s_2)$ and $q(\beta ; r_1,r_2)$ are gamma distributions. The set of latent variables is $\Theta = (\bm{M}, \bm{S}, \bm{u}, \bm{v}, \{\bm{\mu}^x_k, \bm{\Lambda}^x_k\}_{k=1}^{K}, \{\bm{\mu}^y_l, \bm{\Lambda}^y_l\}_{l=1}^{L}, \alpha, \beta)$ and the set of variational parameters is $\bm{\Lambda} = (\bm{\rho}, \bm{\xi}, \{\bar{a}_k, \bar{b}_k\}_{k=1}^{K-1}, \{\{\bar{a}_{lk}, \bar{b}_{lk}\}_{l=1}^{L-1}\}_{k=1}^{K},  \bm{m}^x, \bm{t}^x, \bm{c}^x, \bm{D}^x, \bm{m}^y, \bm{t}^y, \bm{c}^y, \bm{D}^y, s_1, s_2, r_1, r_2)$. Optimization is performed by searching for the set of variational parameters $\boldsymbol{\Lambda}^*$ that maximizes the ELBO. To this end, we extend the  coordinate-ascent variational inference (CAVI) algorithm of \citealp{fiSAN} for our motivational nested data, which is detailed in Algorithm~\ref{alg:NAM_CAVI} in Section~\ref{subsec:CAVI} of the Supplementary Materials. Additional details on the ELBO and its derivation are presented in Section~\ref{subsec::elbo} of the Supplementary Materials. To monitor the convergence of the VI algorithm, we track the difference in ELBO in successive iterations and declare convergence if it is less than $10^{-5}$. Although the ELBO is guaranteed to increase at each iteration, there is no guarantee that the CAVI algorithm will converge to a global optimum. Hence, we execute 50 distinct runs of the algorithm with different starting points, keeping the one with the highest ELBO to draw the inference. 
We obtain the cluster assignment probabilities, $\hat{\rho}_{jk}=q^*(S_{j}=k)$ and $\hat{\xi}_{jil}=q^*(M_{ji}=l)$, and the cluster assignment:
\begin{equation*}
\hat{S}_{j}= \argmax_{k=1,\ldots,K}\hat{\rho}_{jk} \quad \text{and} \quad
\hat{M}_{ji}=\argmax_{l=1,\ldots,L}\hat{\xi}_{jil}
\end{equation*}
for $j=1,\dots,J$ and $i=1,\dots,n_j$.
\subsection{Truncation Approximation Bounds}\label{sec:theory}
The proposed variational inference algorithm is based on the truncated variational families. 
Therefore, it is important to evaluate the error arising from the truncated NAM prior, the corresponding NAM mixture model, and the resulting posterior distribution.

We recall that $G_j | Q \stackrel{iid}{\sim} Q$, where 
$Q$ is defined in \eqref{eq::NAM2}.
We define the truncated versions of $G_j$ as follows, $$G_j|Q^{K,L} \stackrel{iid}{\sim} Q^{K,L},$$
where
\begin{equation} \label{eq:truncation_Q}
Q^{K,L}  = \sum_{k=1}^{K} \pi_k^* \delta_{\left(\delta_{\theta_k^x}, G_k^{y*,L}\right)}, 
\end{equation}
with $\theta_k^x\ \stackrel{iid}{\sim} H^x, \ k = 1,\ldots, K$, $v_l^* \stackrel{iid}{\sim} \mathrm{Beta}(1, \alpha)$, and
\begin{equation} \label{eq:truncation_Q_pi}
\pi_k^* =
\begin{cases}
  v_k^* \prod_{l=1}^{k-1}(1 - v_l^*), & \text{if } k \leq K-1, \\
  \prod_{l=1}^{K-1}(1 - v_l^*),       & \text{if } k = K .
\end{cases}
\end{equation}
Additionally, define,
\begin{equation} \label{eq:truncation_G_k^y}
G_k^{y*,L}  = \sum_{l=1}^{L} \omega_{lk}^*\delta_{\theta_l^y}, 
\end{equation}
with $\theta_l^y\ \stackrel{iid}{\sim} H^y, \ l = 1,\ldots,L$, and for $k=1,\dots, K$, 
\begin{equation} \label{eq:truncation_G_k^y_omega}
\omega_{lk}^*\ =\begin{cases}
 u_{lk}^*\prod_{s = 1}^{l-1}(1-u_{sk}^*), & \text{if  } l \leq L-1,\\
\prod_{s = 1}^{L-1}(1-u_{sk}^*), &\text{if  } l = L,
\end{cases}  
\end{equation}
where $u_{sk}^* \stackrel{iid}{\sim} \mathrm{Beta}(1, \beta)$, for $k = 1,\dots, K$.
Here $L, K>0$ define the truncation levels for the different random probability measures. Consider $J$ groups, each of them containing $n_j$ observations, $j=1,\ldots,J$. Denote by $\bm{z}_j = (\bm{x}_j, (\bm{y}_{ji})_{i=1}^{n_j})$ the collection of all observations from the $j$-th group and by $\bm{z}_{ji} = (\bm{x}_j, \bm{y}_{ji})$ the observation $i$ from the group $j$. Here $\bm{x}_j$ arises from the mixture model $\bm{x}_{j}|\bm{\theta}_{j}^x\sim F^x(\cdot|\bm{\theta}_{j}^x)$ with $\bm{\theta}_{j}^x|G_j^x \sim G_j^x$, where the $G_j^x$'s are generated according to \eqref{eq::NAM_Marginal_for_x} and $\bm{y}_{ji}|\bm{\theta}_{ji}^y\sim F^y(\cdot|\bm{\theta}_{ji}^y)$ with $\bm{\theta}_{ji}^y|G_j^y \sim G_j^y$, where the $G_j^y$'s are generated according to \eqref{eq::NAM_Marginal_for_x2}. Denote by $\bm{\theta}_{ji} = (\bm{\theta}_j^x, \bm{\theta}_{ji}^y)$ and $F(\cdot |\bm{\theta}_{ji}) = F^x(\cdot|\bm{\theta}_j^x) F^y(\cdot|\bm{\theta}_{ji}^y)$. Additionally, denote by $\bm{\theta}_j = (\bm{\theta}_{ji})_{i=1}^{n_j}$ and $\bm{\theta} = (\bm{\theta}_1, \dots, \bm{\theta}_J)$. Let $f(\cdot|\bm{\theta}_{ji})$ be the density of $F(\cdot|\bm{\theta}_{ji})$ with respect to some dominating measure. We assume that $\bm{\theta}_{ji} \in \left(\Theta^x \times \Theta^y\right)$, where $\left(\Theta^x \times \Theta^y\right)$ is a Polish space equipped with its corresponding Borel $\sigma$--field $\mathcal{B}^x \otimes \mathcal{B}^y$.  Finally, we let $\bm{z} = (\bm{z}_1, \ldots , \bm{z}_J)$. Define
\begin{flalign*}
   P^{\infty, \infty}(\bm{\theta}) &= \int \int \left[\prod_{j=1}^{J} \left\{\prod_{i=1}^{n_j}P(\bm{\theta}_{ji}| G_j)\right\}P^{\infty}(dG_j| Q)\right]P^{\infty}(dQ),\\
  P^{K, L}(\bm{\theta}) &= \int \int \left[\prod_{j=1}^{J} \left\{\prod_{i=1}^{n_j}P(\bm{\theta}_{ji}| G_j)\right\}P^{L}(dG_j| Q)\right]P^{K}(dQ),
\end{flalign*}
as the prior distribution of the parameters $\bm{\theta}$ under the NAM and its corresponding truncated version after integrating out the random distributions. Here $P(\bm{\theta}_{ji}\mid G_j)$ denotes the prior distribution of the parameters $\bm{\theta}_{ji}$ given the random measure $G_j$. Furthermore, let $m^{\infty,\infty}(\bm{z})$ and $m^{K,L}(\bm{z})$ denote the marginal distribution of the data $\bm{z}$ derived from these priors. Then we have the following result.
\begin{prp} \label{prp:MarginaAndPriorConvergence}
Let $P^{\infty, \infty}(\bm{\theta})$ and $P^{K, L}(\bm{\theta})$ denote the prior distribution of the parameters $\bm{\theta}$ under the NAM prior and its corresponding truncated version with the random measures integrated out. Furthermore,  let $m^{\infty,\infty}(\bm{z})$ and $m^{K, L}(\bm{z})$ denote the marginal distribution of the data $\bm{z}$, derived from these priors. Then,
\begin{flalign*}
\begin{aligned}
& \int_{\mathcal{X}^N} \left|m^{K, L}(\bm{z}) - m^{\infty,\infty}(\bm{z})\right|d\bm{z} \leq  \int_{\Xi^N} \left|P^{K, L}(\bm{\theta}) - P^{\infty, \infty}(\bm{\theta})\right|d\bm{\theta} \leq \epsilon^{K, L}(\alpha, \beta),
\end{aligned}
\end{flalign*}
where
\begin{flalign*}
    \epsilon^{T,K}(\alpha, \beta) =  4 \left[ 1 - \left\{ 1- \left(\frac{\alpha}{1+ \alpha}\right)^{K-1}\right\}^J \left\{ 1- \left(\frac{\beta}{1+ \beta}\right)^{L-1}\right\}^N \right], 
\end{flalign*}
$N= n_1 +\cdots + n_J$, $\Xi^N = \prod_{j=1}^{J}\left(\Theta^x \times (\Theta)^{n_j}\right)$, and $\mathcal{X}^N$ denotes the sample space of observations $\bm{z}$.
\end{prp}
The proof follows directly by adapting the arguments in Theorem B.1 of \citealp{nestedDP} and is therefore omitted. We further note that the error bounds under the NAM prior coincide with those established for the nDP. From Proposition~\ref{prp:MarginaAndPriorConvergence}, it follows that the bounds vanish asymptotically, implying that both the truncated prior and the marginal data distribution converge in total variation, and hence in distribution, to the NAM. Moreover, the approximation errors decay exponentially in both $K$ and $L$. Consequently, we obtain the following result concerning the posterior distribution of the parameters $\bm{\theta}$ under the truncated NAM prior.

\begin{prp} \label{prp:PosteriorConvergence}
The posterior distribution of the parameters $\bm{\theta}$ under the NAM prior and its truncated version,
\begin{flalign*}
\pi^{\infty,\infty}(\bm{\theta}|\bm{z}) & = \frac{f(\bm{z}|\bm{\theta})P^{\infty, \infty}(\bm{\theta})}{m^{\infty,\infty}(\bm{z})}, \\
\pi^{K, L}(\bm{\theta}|\bm{z}) & = \frac{f(\bm{z}|\bm{\theta})P^{K, L}(\bm{\theta})}{m^{K, L}(\bm{z})},
\end{flalign*}
satisfies
\begin{flalign*}
& \int_{\mathcal{X}^N}\int_{\Xi^N}\left|\pi^{K, L}(\bm{\theta}|\bm{z}) -  \pi^{\infty,\infty}(\bm{\theta}|\bm{z})\right| m^{\infty,\infty}(\bm{z})\:d\bm{\theta}\:d\bm{z}  = \mathcal{O}\left(\epsilon^{K, L}(\alpha, \beta)\right).
\end{flalign*}
\end{prp}
The proof follows as a consequence of Proposition 3.2 of \citealp{chakrabarti2025globallocal}. Thus, Proposition~\ref{prp:PosteriorConvergence} establishes that the posterior distribution of $\bm{\theta}$ under the truncated NAM prior is exponentially accurate when integrated with respect to the marginal data density $m^{\infty,\infty}(\bm{z})$ induced by the original NAM prior.

\section{Simulation Studies}\label{sec::simulations}
We assess the performance of NAM by comparing it with CAM and fiSAN in simulations. For CAM with multivariate data, we implemented posterior inference using a modification of our proposed VI algorithm, while for fiSAN we adapted the \texttt{sanba} R package \citep{fiSAN} to accommodate multivariate data. We did not include comparisons with the HHDP, as no publicly available software for posterior inference under this model currently exists. Throughout the simulations, we assume that there are $J = 100$ groups and $n_j = 100$ observations within each group $j=1,\dots,J$. 

In the first scenario, both the group-level variables ($\bm{x}_j$) and the observation-level variables ($\bm{y}_{ji}$) are two-dimensional (i.e., $p=q=2$). The true number of GCs is $K = 4$ and the true number of OCs is $L=3$. The true cluster-specific parameters and the true mixture weights corresponding to the group-level variables are drawn from,
 \begin{align*}
     &\left(\bm{\mu}^x_{k}, \bm{\Lambda}^x_{k}\right)  \sim \text{NW}(\bm{0}, 0.05, q + 5, \mathds{I}_{q}),\\
     &\boldsymbol{ {\pi}} \sim \mbox{Dirichlet}(\alpha/K, \dots, \alpha/K),\quad \alpha  \sim \mathrm{Gamma}(25, 1).
 \end{align*}
 The true group-level cluster indicator $S_{j}$ is drawn from a multinomial distribution with class probabilities $\boldsymbol{ {\pi}}$. We then generate $\bm{x}_j$ from
 \begin{align*}
 \label{eq:covariates_sim}
     \bm{x}_j | S_j=k\sim  \mathcal{N}_{q}\left(\bm{x}_{j}\mid \bm{\mu}^x_{k}, \Lambda_{k}^{x^{-1}}\right).
 \end{align*} 
  Similarly, the true cluster-specific parameters and mixture weights corresponding to the observation-level variables are drawn from,
 \begin{align*}
     &\left(\bm{\mu}^y_{l}, \bm{\Lambda}^y_{l}\right)  \sim \text{NW}(\bm{0}, 0.05, 5 + p, \mathds{I}_{p}),\\
     &\boldsymbol{\omega}_{k}   \sim \mbox{Dirichlet}(\beta/L, \dots, \beta/L), \quad \beta \sim \mathrm{Gamma}(25, 1).
 \end{align*}
The true observation-level cluster indicators $M_{ji}$, for $i=1,\dots n_j$, are conditionally drawn from a multinomial distribution with class probabilities $\boldsymbol{ {\omega}}_k$, given $S_j = k$. Lastly, we generate $\bm{y}_{ji}$ from
 \begin{align*}
     \bm{y}_{ji} \mid M_{ji} = l \sim \mathcal{N}_{p}\left(\bm{y}_{ji}\mid \bm{\mu}^y_{l}, \Lambda_{l}^{y^{-1}}\right).
 \end{align*} 
All simulations and data analysis were conducted on a machine equipped with a 2.45 GHz 64-core AMD EPYC 7763 processor and 24 GB of RAM. We assessed the accuracy of estimating GCs and OCs using the adjusted Rand index (ARI; \citealp{ARI}). 
We report the ARI for GCs and the average ARI for group-specific OCs along with their standard deviation in Table~\ref{tab:ARI_Comparison}. Furthermore, we report the overall OC accuracy, which measures the accuracy of OC while accounting for the sharing of clusters across groups. It is evident that the GC accuracy of NAM surpasses that of both CAM and fiSAN, highlighting the importance of incorporating informative group-level variables when available. All methods achieve similar performance in identifying OCs. 

\begin{table}
\centering
\begin{tabular}{c|c|c|c}
  \hline
        &   GC   &  Average OC            & Overall OC \\ 
  \hline
  CAM   & 0.4944 & 0.9611 (0.0386)  & 0.9627 \\ 
  fiSAN & 0.4944 & 0.9610 (0.0387)  & 0.9625 \\ 
  NAM  & 1.0000 & 0.9620 (0.0383) & 0.9632 \\ 
   \hline
\end{tabular}
\caption{Comparison of clustering performance at the group- and observation-level for the different methods. The accuracy of clustering was assessed using the adjusted Rand index (ARI) between the estimated and true cluster.}
\label{tab:ARI_Comparison}
\end{table}

In the second scenario, we investigate the effect of varying dimensionality of group-level variables on GC accuracy through a series of simulations. For simplicity, we set $p=q$ and vary it in $\{2,5,10\}$. All other simulation details are the same as before and are replicated 50 times. Figure~\ref{fig:NAM_CAM_fiSAN_DC} shows the ARI of GCs for NAM, CAM, and fiSAN, which demonstrates that NAM consistently outperforms the other two methods across all dimensional settings, especially for $p=q=10$ where NAM achieves near-perfect accuracy in identifying GCs, whereas the performance of the alternative methods is substantially worse. Table~\ref{tab:ARI_OC_Comparison_Summary_by_Dimension} presents the mean and standard deviation of the ARI of OCs aggregated over all $100$ groups and $50$ independent replications. The results indicate that NAM is comparable to CAM and fiSAN across all settings in terms of identifying OCs. Furthermore, Figures \ref{fig:NAM_CAM_fiSAN_OC_2D}-\ref{fig:NAM_CAM_fiSAN_OC_10D} in Section~\ref{sec::suppl_simulations} of the Supplementary Materials show the boxplots of OC accuracy for each group separately, computed across the 50 independent replications for all the methods and dimensions. They further corroborate the comparable performance in identifying OCs.

\begin{figure}
\centerline{\includegraphics[width=0.8\linewidth]{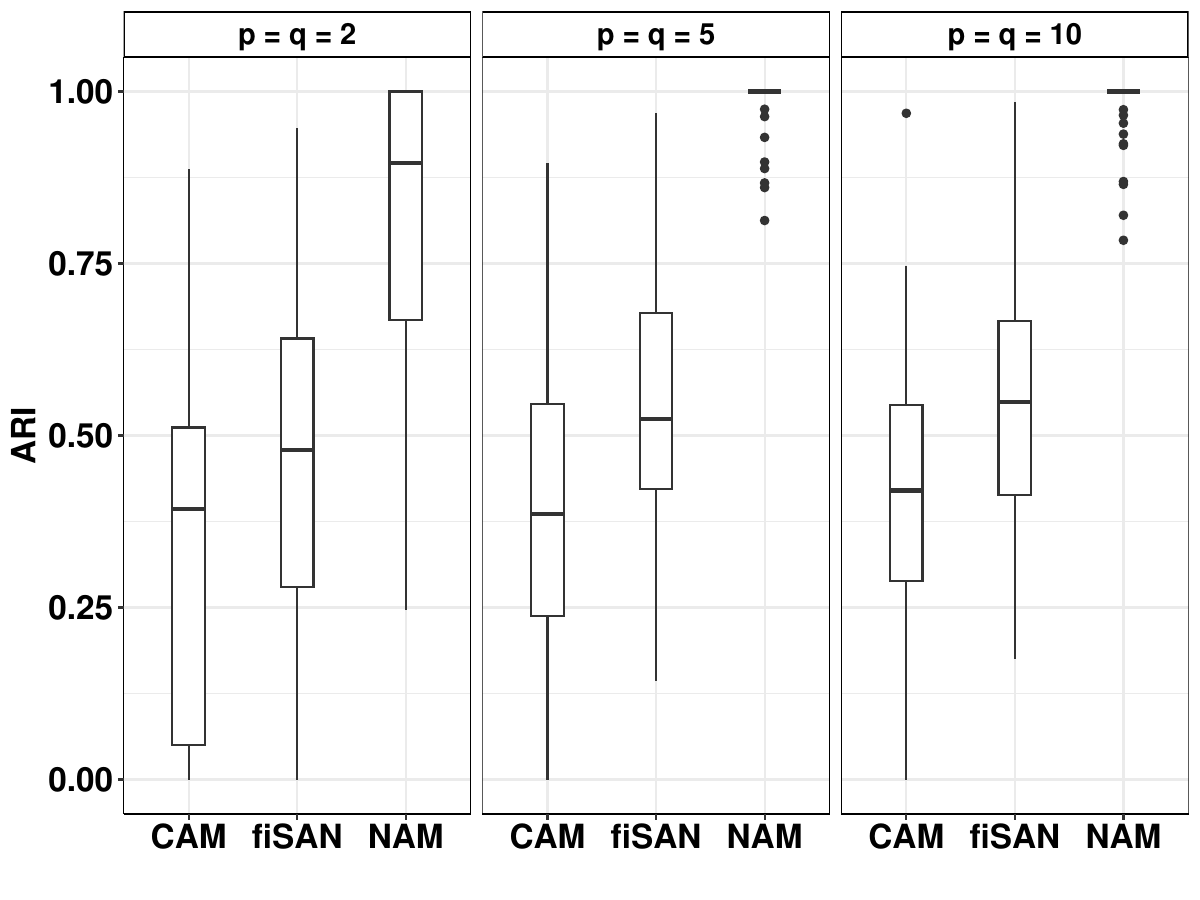}}
\caption{Comparison of group-clustering accuracy of NAM with CAM and fiSAN against varying dimensions of the group-level variables ($q$) and observation-level variables ($p$). The dimensions are reported at the top of the figure panels. Boxplots are reported over 50 independent replications.
}
\label{fig:NAM_CAM_fiSAN_DC}
\end{figure}

\begin{table}
\begin{center}
\begin{tabular}{c|c|c|c}
  \hline
                &   CAM         &   fiSAN        & NAM          \\ 
  \hline
   $p=2$     &  0.896 (0.142) & 0.903 (0.136)  & 0.902 (0.136) \\ 
   $p=5$     &  0.995 (0.021) & 0.997 (0.017)  & 0.997 (0.017) \\ 
   $p=10$    &  1.000 (0.001) & 1.000 (0.001)  & 1.000 (0.001) \\ 
   \hline
\end{tabular}
\end{center}
\caption{Comparison of clustering performance at the observation-level for the different methods for varying dimension of the response variable. The accuracy of clustering was assessed using the adjusted Rand index (ARI) between the estimated and true cluster. The table reports the mean (standard deviation) of the ARI across all individuals over the 50 independent replications.}
\label{tab:ARI_OC_Comparison_Summary_by_Dimension}
\end{table}

Next, we performed simulations to assess the scalability of the proposed VI algorithm. Particularly, we assessed the algorithmic scalability with respect to (i) dimensionality of group- and observation-level variables ($q$ and $p$ respectively), (ii) the number of groups ($J$), and (iii) the number of observations within each group ($n_j, \:j = 1,\dots, J$). First, as before, we assume that $J = 100$ with $n_j = 100$ and vary the dimensions of group- and observation-level variables. Additionally, as before, we set $p=q$ for simplicity and vary it in $\{2, 5, 10\}$. Next, we vary the number of groups, $J$ in $\{100, 200, 400\}$ while fixing the dimensions $p=q=5$ and $n_j = 100$ for $j = 1,\dots, J$. Finally, with the dimensions fixed at $p=q=5$ and the number of groups fixed at $J = 200$, we vary $n_j$ in $\{100, 200, 400\}$. All other simulation details are the same as before and are replicated 50 times. For each replication, we execute 50 distinct runs of the proposed algorithm with different starting points, and we report the median and maximum individual runtime obtained over the 50 runs. The boxplots of runtime in Figure~\ref{fig:ComputationalTime} summarize the results. The runtime scales approximately linearly with $J$ and $n_j$ (Figures~\ref{fig:MedianTime_WithJ}-\ref{fig:MaxTime_Withn}), with the maximum runtime remaining below 18 and 35 minutes, respectively. Interestingly, Figure~\ref{fig:MedianTime_WithDimension} indicates that the median runtime remains below two minutes across all examined values of  $p = q$. Somewhat unexpectedly, the runtime for higher-dimensional settings ($p = q = 5, 10$) is shorter than that observed in the two-dimensional case due to faster convergence of the VI algorithm. The median number of iterations to convergence is 562 for the two-dimensional case ($p = q = 2$), compared to 203.5 and 141.5 for the cases $p = q = 5$ and $p = q = 10$, respectively. Furthermore, the corresponding maximum runtimes in Figure~\ref{fig:MaxTime_WithDimension} are less than 6 minutes, although the variability increases with dimensionality. In summary, the results highlight that the proposed variational inference algorithm for our sophisticated Bayesian nonparametric hierarchical modeling is scalable with respect to both the number of variables and the sample size.

\begin{figure}[!htp]
\centering
\begin{subfigure}{0.48\linewidth}
  \centering
    \includegraphics[width= 1\linewidth]{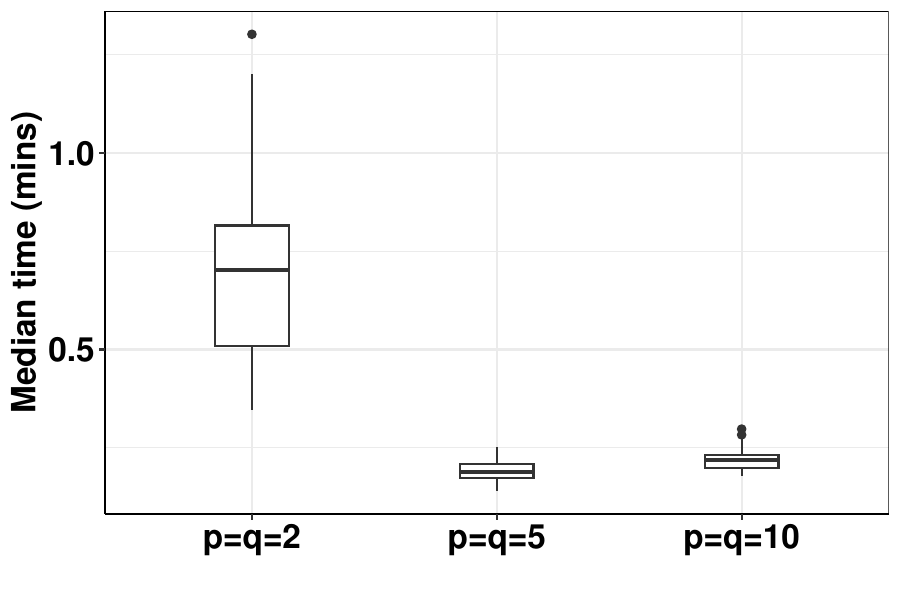}
    \caption{}
    \label{fig:MedianTime_WithDimension}
\end{subfigure}
\hspace{-0.2cm} 
\begin{subfigure}{0.48\linewidth}
    \centering
    \includegraphics[width= 1\linewidth]{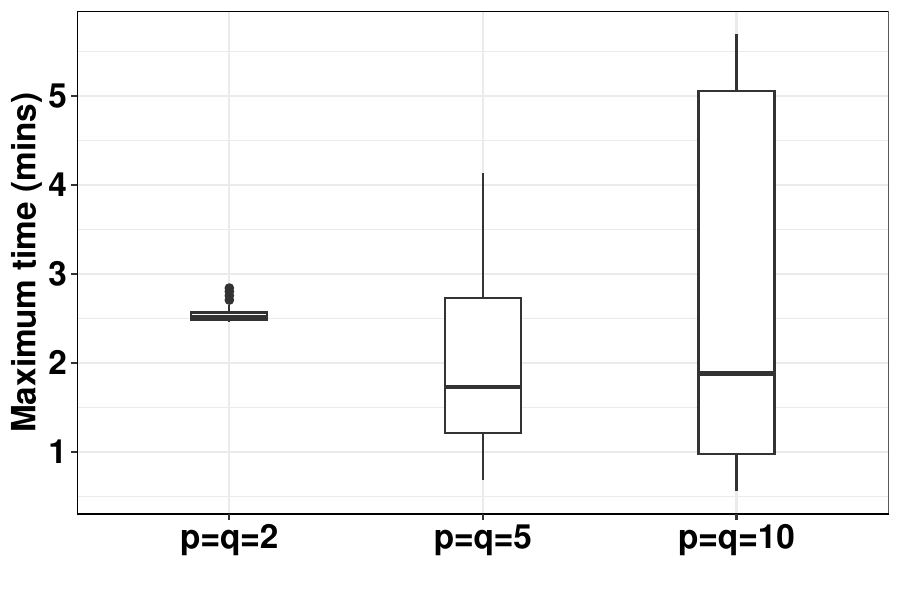}
    \caption{}
    \label{fig:MaxTime_WithDimension}
\end{subfigure}
\hspace{-0.2cm} 
\begin{subfigure}{0.48\linewidth}
  \centering
    \includegraphics[width= 1\linewidth]{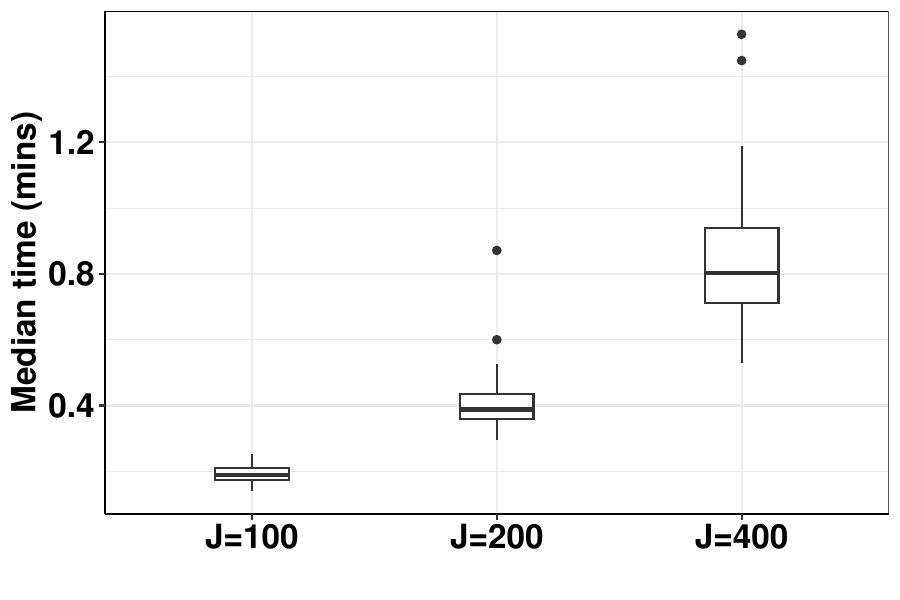}
    \caption{}
    \label{fig:MedianTime_WithJ}
\end{subfigure}
\hspace{-0.2cm} 
\begin{subfigure}{0.48\linewidth}
    \centering
    \includegraphics[width= 1\linewidth]{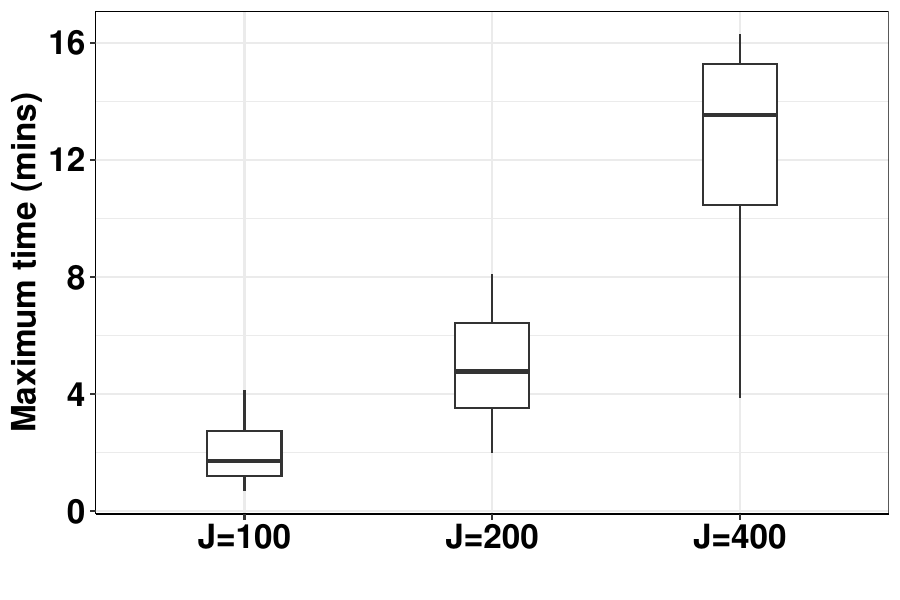}
    \caption{}
    \label{fig:MaxTime_WithJ}
\end{subfigure}
\hspace{-0.2cm} 
\begin{subfigure}{0.48\linewidth}
  \centering
    \includegraphics[width= 1\linewidth]{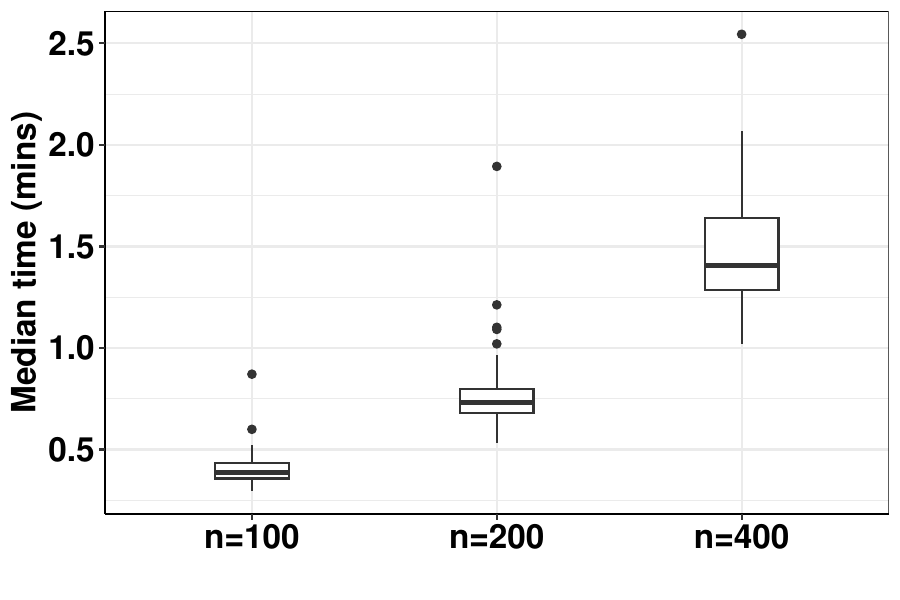}
    \caption{}
    \label{fig:MedianTime_Withn}
\end{subfigure}
\hspace{-0.2cm} 
\begin{subfigure}{0.48\linewidth}
    \centering
    \includegraphics[width= 1\linewidth]{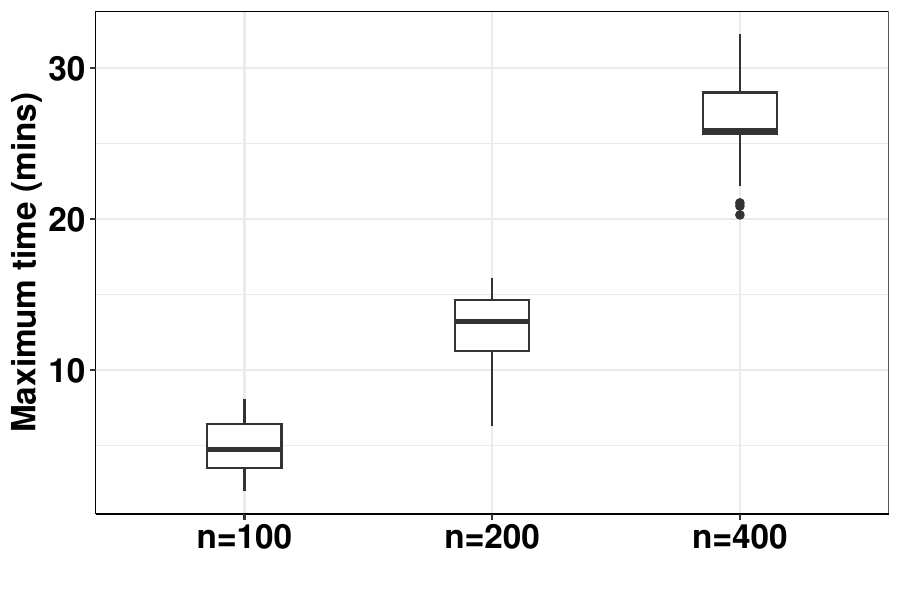}
    \caption{}
    \label{fig:MaxTime_Withn}
\end{subfigure}
\caption{Top row: distributions of the (a) median and (b) maximum computing time with varying dimensions of group- level variables ($q$) and observation-level variables ($p$) for fixed $J=100$ and $n_j=100$. Middle row: distributions of the (c) median and (d) maximum computing time with varying number of groups ($J$) for fixed dimensionality $p=q=5$ and  $n_j = 100$. Bottom row: distributions of the (e) median and (f) maximum computing time with varying number of observations within each group ($n_j$) for fixed $J = 200$ and dimensionality $p=q=5$. The median and maximum are computed over 50 distinct runs of the algorithm. The boxplots are reported over 50 independent replications.}
\label{fig:ComputationalTime}
\end{figure}

We further conducted simulation studies to assess the robustness of our model to deviations from Gaussianity in the group-level variables ($\bm{x}_j$) and/or the observation-level variables ($\bm{y}_{ji}$). Specifically, we generated data from mixture of (i) multivariate $t$ distributions with three degrees of freedom and (ii) multivariate skew-$t$ distributions with three degrees of freedom and skewness parameter set to four (see \citealp{Azzalini_2013} for details), and applied our method and competing methods with mixture of Gaussian kernels. As before, we assumed that the true number of GCs is $K = 4$, the true number of OCs is $L = 3$, and set $p = q = 5$.  All other simulation settings matched those considered before, and results were averaged over 50 replications. We compared the performance of the NAM with the CAM and the fiSAN. Figure~\ref{fig:All_GC_OC} in the Supplementary Materials summarizes the distribution of ARI values for estimating GCs and OCs across different combinations of mixture distributions used to generate the group-level and observation-level variables. Figure~\ref{fig:All_GC} shows that, across all scenarios, NAM is relatively robust to misspecification of the underlying mixture distribution when estimating group clusters. Moreover, when the observation-level variables are generated from Gaussian mixtures, the accuracy of OC estimation is high and comparable across all three methods (Figure~\ref{fig:All_OC}), regardless of the distribution governing the group-level variables. When the observation-level variables arise from non-Gaussian mixtures, clustering accuracy for OCs deteriorates for all methods, however, their performances remain broadly comparable.

We further examine the estimated numbers of GCs and OCs for the three methods; details are provided in Section~\ref{sec::suppl_simulations} of the Supplementary Materials. In summary, NAM demonstrates greater robustness in estimating the number of GCs and OCs across all settings compared to the competing methods.

Since NAM is similar to CAM and fiSAN when there are no group-level variables, we expect the group-level clustering of NAM to deteriorate but still outperform CAM and fiSAN if group-level variables are informative but partially omitted. We conducted simulations to assess the robustness of the proposed NAM framework to partial omission of group-level variables. To investigate the impact of incomplete group-level information, we randomly omitted $r$ group-level variables, with $r \in \{2, 5, 10, 20\}$. The case $r= 0$ corresponds to the full-information setting in which all group-level variables are retained. All other simulation settings matched those considered before. For each value of $r$, we fitted NAM and compared its performance with (i) NAM under the full-information setting ($r= 0$), (ii) CAM, and (iii) fiSAN. Results were averaged over 50 independent replications. Figure~\ref{fig:GC_OC_ARI_Missing} in the Supplementary Materials summarizes the distributions of ARI values for recovering both GCs and OCs. As expected, the accuracy of group-level clustering under NAM gradually declines as the number of omitted group-level variables increases. Nevertheless, even in the most extreme scenario considered ($r=20$, corresponding to omission of 20\% of the group-level variables), NAM continues to demonstrate higher accuracy in estimating GCs relative to both CAM and fiSAN. For observational clustering, NAM exhibits stable performance across all levels of omission, with ARI values comparable to those obtained under CAM and fiSAN.

\section{Real Data Analysis}\label{sec:real_data_analysis}
The motivating OneK1K scRNA-seq data contain gene expressions and genotypes for 1.27 million peripheral blood mononuclear cells from 982 individuals.
For quality control, we remove: (i) cells with cell-type prediction scores $<$ 0.6 or extreme total coverages (number of RNA reads $>$ 3 mean absolute deviation from median); (ii) donors with extreme cell numbers ($>$ 3 mean absolute deviation from median); (iii) genes with low overall expressions ($<0.5\%$ cell expression rate); and (iv) SNPs with low rare allele presence (fewer than five donors with the minor allele). Of the 10,000 genes with highest mean expression, we further filter out those with no SNPs within 200Kb from the genes' transcription start sites.

After quality control, our analysis includes data from $J=961$ individuals. For each individual, we retain 204,246 common genetic variants (SNPs) along with scRNA-seq gene expression measurements on a common set of 10,786 genes obtained from a total of 929,817 blood cells. As in standard practice, we apply principal component analysis (PCA) to the SNP data and retained the top five principal components (PCs). Similarly, we perform PCA on the gene expression data across all 961 individuals and 929,817 single cells, retaining the top five PCs. Given that the PCs derived from SNP data and gene expression data may be on very different scales, we standardized (centered and scaled) the PCs. Thus, the final dataset comprises the $q=5$ standardized PCs of the SNP data and $p=5$ standardized PCs of the gene expression data. The number $n_j$ of cells for each of the 961 individuals ranges from 366 to 1622. Additionally, prior to our analysis, we conducted basic diagnostic checks on these standardized PCs, the details of which are provided in Section~\ref{sec::suppl_rda:diagnostics} of the Supplementary Materials. In summary, these diagnostics suggest that a Gaussian mixture specification adequately captures the marginal distributional characteristics of both the SNP-derived PC data and the gene-expression PC data.

Applying the proposed method to this dataset, we simultaneously cluster the individuals and the cells within each individual while allowing the cell clusters to be comparable across individuals.
This allows us to have a better understanding of inter-individual as well as intra-individual genetic and cellular heterogeneity. The maximum runtime in 30 distinct runs of the VI algorithm was less than 22 hours. Furthermore, we examine whether the chosen truncation boundaries $(K = L = 50)$ for the proposed VI algorithm were active. The results suggest that the selected truncation levels are adequate for posterior inference; see Section~\ref{sec::suppl_rda:sensitivity} of the Supplementary Material for additional details.

NAM finds 6 clusters for individuals and 14 cell clusters. To visualize similarities and differences in the estimated OCs across the six estimated GCs, we randomly select two individuals from each GC and plot the uniform manifold approximation and projection (UMAP; \citealp{McInnes2018}) embeddings of the gene expression colored by the estimated OCs in Figure~\ref{fig:RealData_OC_UMAPPlot}. Clearly, individuals belonging to the same GC show similar cell clusters, while individuals belonging to different GCs show some differences in cell-cluster compositions. For example, the individuals in GC 40 (note that the label is arbitrary) show very similar OC patterns. However, these OCs show a marked difference from those corresponding to the individuals in GC 23. In other words, individuals with similar SNP profiles tend to share more similar cell-type compositions and gene expression profiles than those with dissimilar SNPs. Later, we performed additional analyses to further assess the validity and interpretability of the identified GCs and their corresponding OCs.

The OneK1K dataset provides annotated labels of each single cell from every individual into one of 14 distinct cell types, including B cells, T cells, natural killer cells, monocytes, etc. We merged all B cells subtypes and likewise, the CD4 T cells subtypes, and thereafter, compared our estimated clusters with the annotated cell types. We assess the accuracy of the estimated clusters using the ARI. The accuracy in clustering ranges from 0.049 to 0.918 across 961 individuals, with the median accuracy being 0.57. Furthermore, the overall OC accuracy, while accounting for the sharing of clusters across individuals, is 0.344. 

We compare the clustering performance of the NAM with that of other grouped clustering methods. Since the publicly available implementation of CAM for multivariate observations relies on MCMC, it is computationally infeasible for our high-dimensional dataset and was therefore not considered. For the real data analysis, we instead applied fiSAN with truncation levels fixed at $K = L = 50$, adapting the \texttt{sanba} R package to handle multivariate data. The VI algorithm was executed 30 times with different random initializations in parallel, and the solution corresponding to the highest ELBO was selected for inference. By design, fiSAN disregards the SNP data and, in this setting, produced 37 individual clusters (GCs) and 23 cell clusters (OCs). Both values greatly exceed the numbers of clusters estimated by NAM. It is less plausible that the true number of individual clusters is this large, as the OneK1K dataset consists of a relatively homogeneous cohort of individuals of Northern European ancestry.
Figure~\ref{fig:RealData_OC_UMAPPlot_fiSAN} presents the UMAP embeddings of the gene expression data, colored by the estimated OCs from fiSAN for the same set of individuals as plotted before. The results reveal some inconsistencies. For example, individuals 19 and 666 exhibit highly similar cellular compositions (as reflected by their comparable OCs), and fiSAN correspondingly assigns them to the same GC, labeled 32. In contrast, although individuals 950 and 898 also demonstrate strong similarity in their OCs, fiSAN assigns them to distinct GCs. This discrepancy is likely due in part to fiSAN not accounting for the similarity in the SNP profiles of these individuals. Regarding the group-specific OC accuracy, the ARIs under fiSAN are similar to those under NAM, albeit with slightly larger variation across the 961 individuals, ranging from –0.030 to 0.918, with a median of 0.57. The overall OC accuracy under fiSAN (accounting for the sharing of clusters across groups), however, is lower (0.277) than that achieved by NAM (0.344). This underscores the NAM's critical feature of incorporating genotype-level similarity in guiding clustering not only at the group level but at the observation level as well. 

\begin{figure}
\centering
\includegraphics[width=0.9\columnwidth]{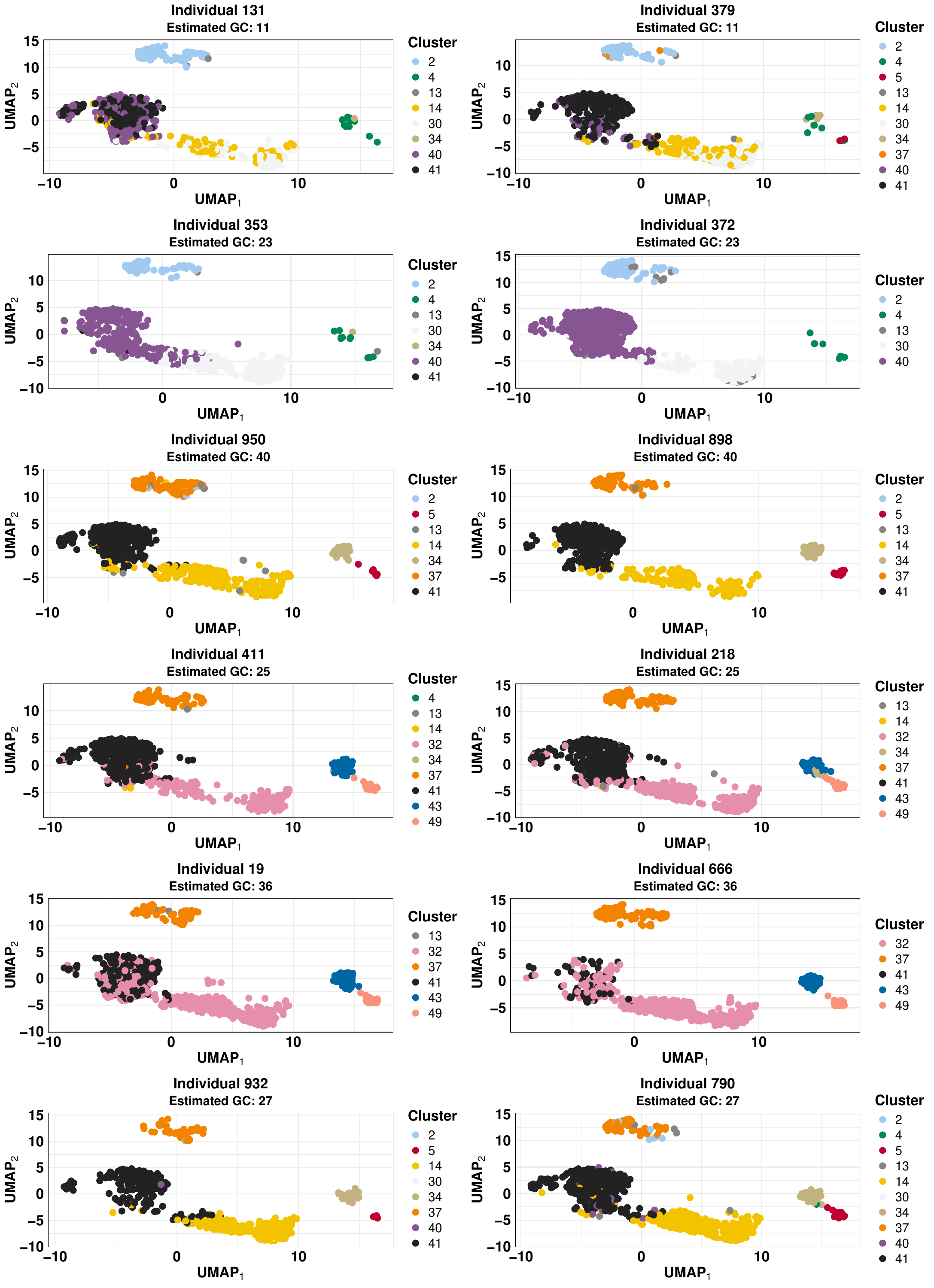}
\caption{Observational clustering for selected individuals by estimated group clusters (GCs). The colors indicate observational clusters (OCs) estimated from NAM.
}
\label{fig:RealData_OC_UMAPPlot}
\end{figure}

\begin{figure}
\centering
\includegraphics[width=0.9\columnwidth]{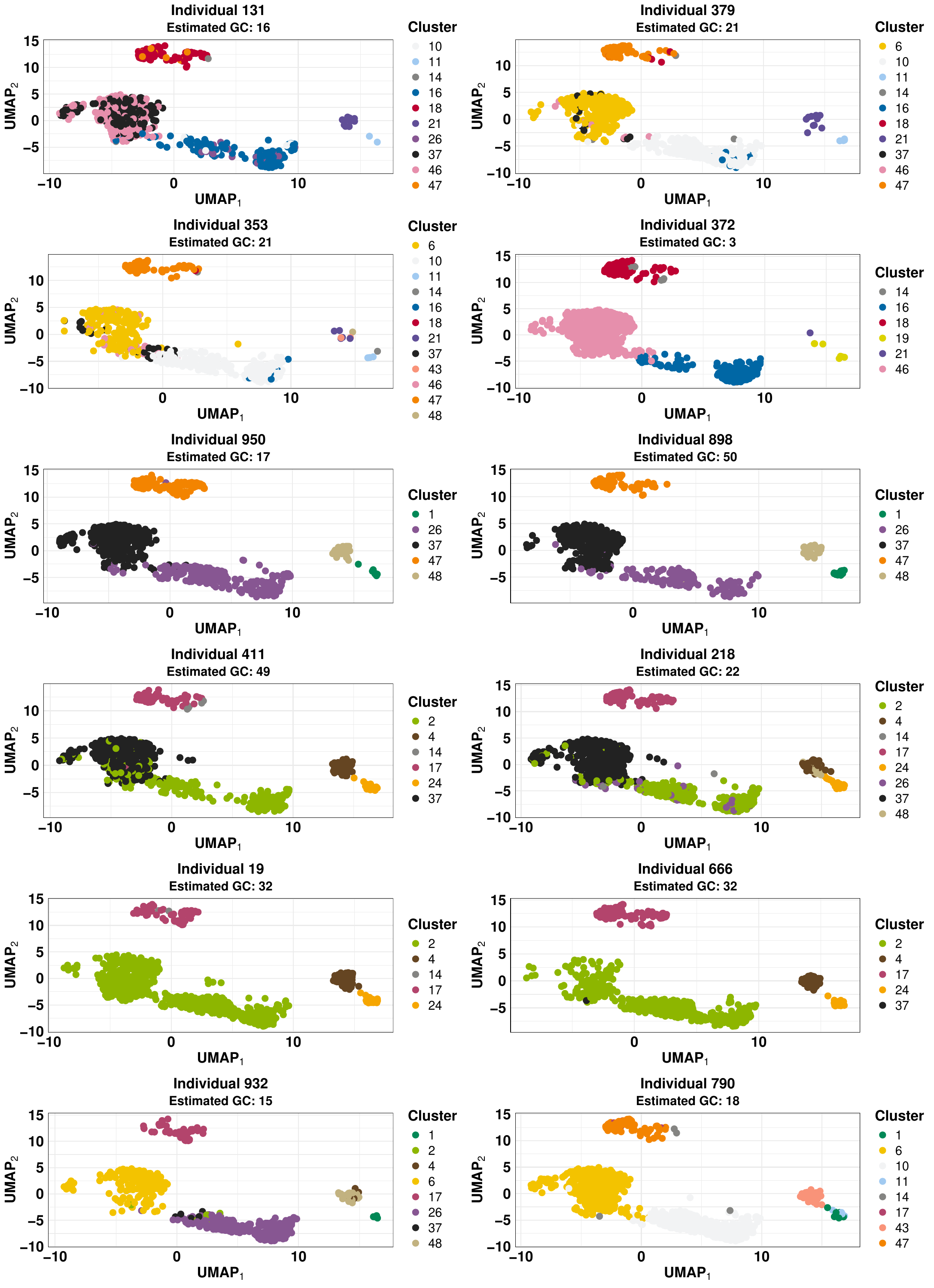}
\caption{Observational clustering for the same individuals by estimated group clusters (GCs). The colors indicate observational clusters (OCs) estimated from fiSAN.
}
\label{fig:RealData_OC_UMAPPlot_fiSAN}
\end{figure}

Next, to support our findings of cell clusters, we identify genes that best characterize the cell clusters via differential gene expression analyses. We use the function \texttt{findMarkers} from the Bioconductor R package \texttt{scran} \citep{scran} and report the 6 most differentially expressed genes for each of the same randomly selected two individuals from each GC as before. Figure \ref{fig:RealData_DEAnalysis} shows the distributions of those differentially expressed genes. The pairs of individuals from GCs 11 and 23 share a common differentially expressed gene \emph{MS4A1}. This gene is a member of the membrane-spanning 4A gene family and encodes a B-lymphocyte surface molecule, which plays a role in the activation, proliferation, and differentiation of B-cells into plasma cells \citep{tedder1985b, petrie2002colocalization}. Additionally, MS4A1 is also differentially expressed in individual 950 from GC 40 and individual 790 from GC 27. Across all individuals, this gene exhibits elevated expression levels within OCs 2 and 37 (recall that the OCs are comparable across individuals because of the sharing-of-information feature of NAM). Cells in these clusters are predominantly B-cells based on the cell markers, further supporting the biological relevance of the observed gene expression pattern.

The individuals from GC 23 share a common differentially expressed gene \emph{CSTA} encoding cystatin A. It is also differentially expressed in individual 131 from GC 11. It has a similar expression pattern across most of the OCs. However, it has a substantially higher expression in OC 4 in individuals 11, 353, and 372. \citet{WEI202212_CSTA_Osteoclast} shows that CSTA is expressed in osteoclasts, which are multinucleated cells originating from monocytes and play a crucial role in bone resorption \citep{Osteoclasts}. Although the blood sample in OneK1K should not contain osteoclasts, it contains monocyte, which derives osteoclasts. In fact, OC 4 appears to be dominated by monocytes according to the cell markers. Although this analysis involves identifying differentially expressed genes and manually comparing them with known markers from the literature, we later propose an automated pipeline for cell-type identification that performs this characterization in a systematic and scalable manner.

\begin{figure}
\centering
\includegraphics[width=0.9\columnwidth]{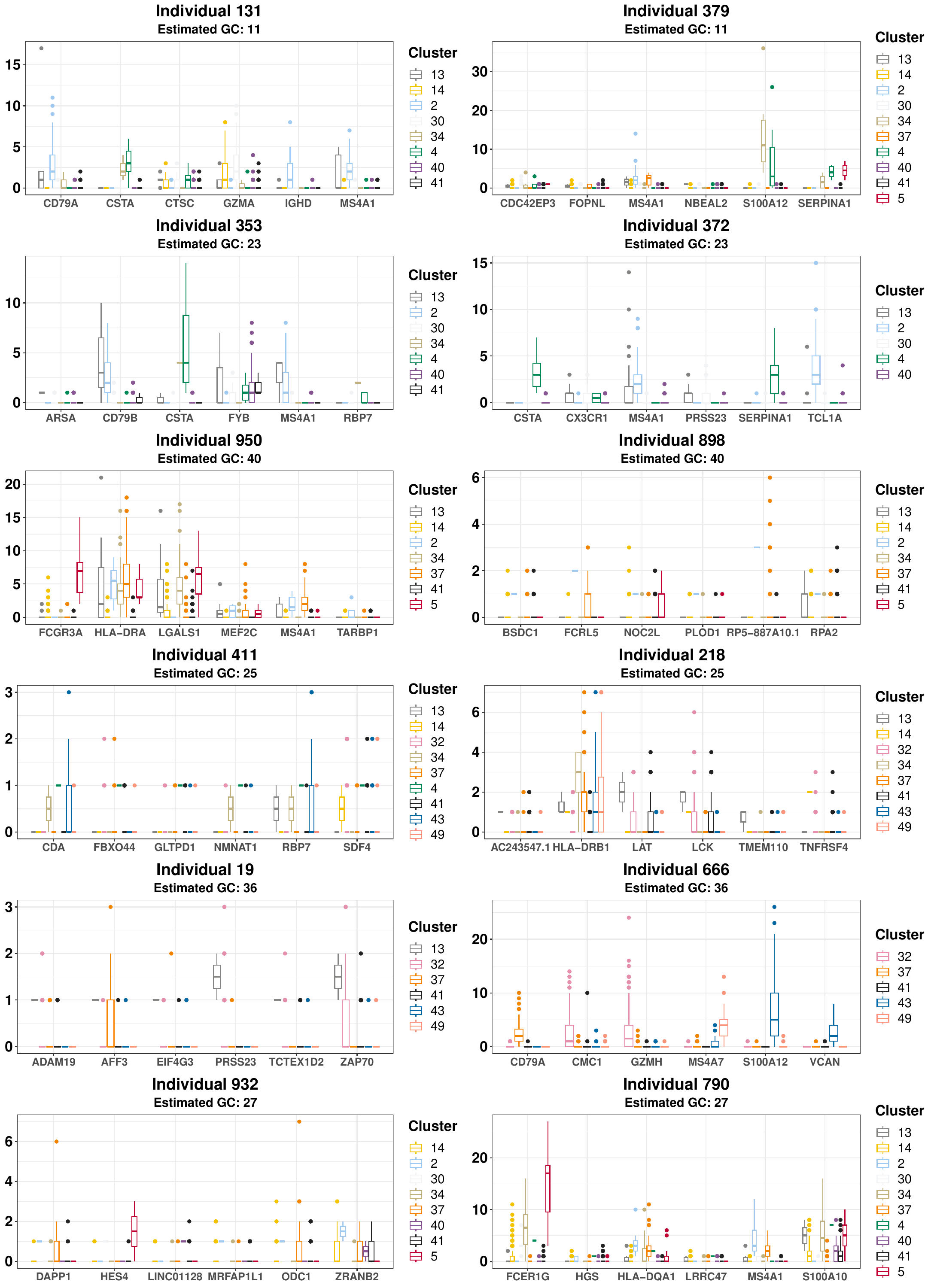}
\caption{Boxplot of gene expressions in the top six differentially expressed genes for selected individuals by estimated group clusters (GCs) in the different observational clusters (OCs).
}

\label{fig:RealData_DEAnalysis}
\end{figure}

Next, we conducted additional analyses to further examine and validate the estimated GCs and OCs. Specifically, we performed differential gene expression analysis on pseudo-bulk data derived from the estimated GCs within the same OC. For each individual and each estimated OC, we aggregated the gene expression counts across all cells belonging to that OC, thereby avoiding treating individual cells as independent replicates. Differential expression analysis was then conducted using \texttt{edgeR} \citep{edgeR}, treating individuals assigned to different estimated GCs as biological replicates. This approach enables testing whether gene expression differs between individuals belonging to two estimated GCs within the same OC (i.e., cell cluster).

Figure~\ref{fig:GC11_vs_GC40_heatmaps} presents heatmaps of the top twenty differentially expressed genes for OCs labeled 2 and 34 when comparing individuals assigned to GC 11 with those assigned to GC 40. The results indicate that different GCs exhibit markedly distinct gene expression patterns, even though they belong to the same OC. Additionally, the two sets of differentially expressed genes from the two different OCs are largely non-overlapping. These findings provide additional support for the biological distinctions between the estimated GCs, which are informed by the SNP-derived group-level data.

Furthermore, from Figure~\ref{fig:RealData_DEAnalysis}, we noted that the gene MS4A1 was differentially expressed between individuals from GCs 11 and 40 based on analyses conducted at the single-cell level. The pseudo-bulk analysis conducted here further corroborates this observation. In particular, MS4A1 exhibits distinct expression patterns across individuals assigned to different GCs within the same estimated OC (see Figure~\ref{fig:MS4A1_OC2_OC34_GC11_vs_GC40_Violin}), providing additional evidence that the GC assignments correspond to meaningful differences in cellular composition and gene expression profiles.

\begin{figure}[!htp]
        \centering
           \begin{subfigure}{0.49\textwidth}
            \centering
            \includegraphics[width= 1\linewidth]{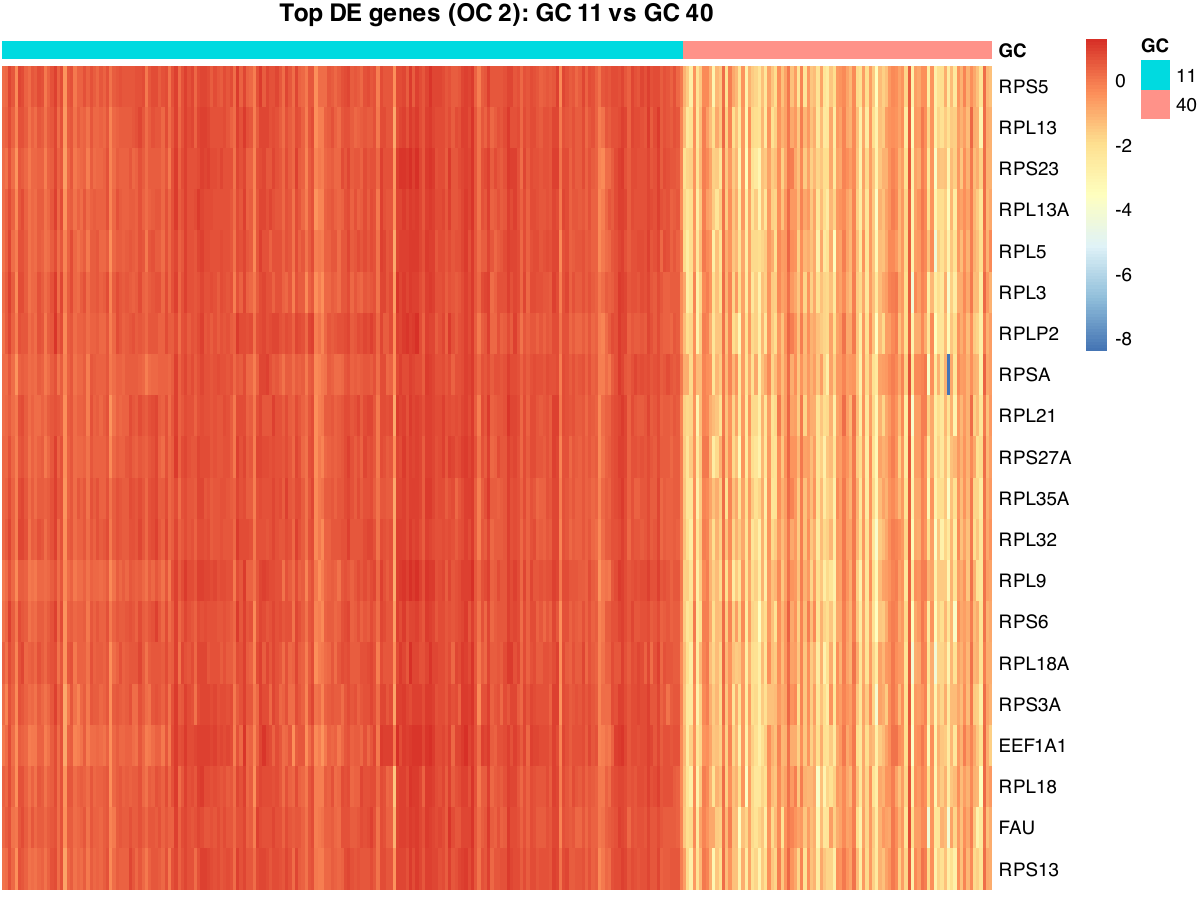}
            \caption{}
            \label{fig:OC2_GC11_vs_GC40_heatmap3}
        \end{subfigure}
            \begin{subfigure}{0.48\textwidth}
              \centering
                \includegraphics[width= 1\linewidth]{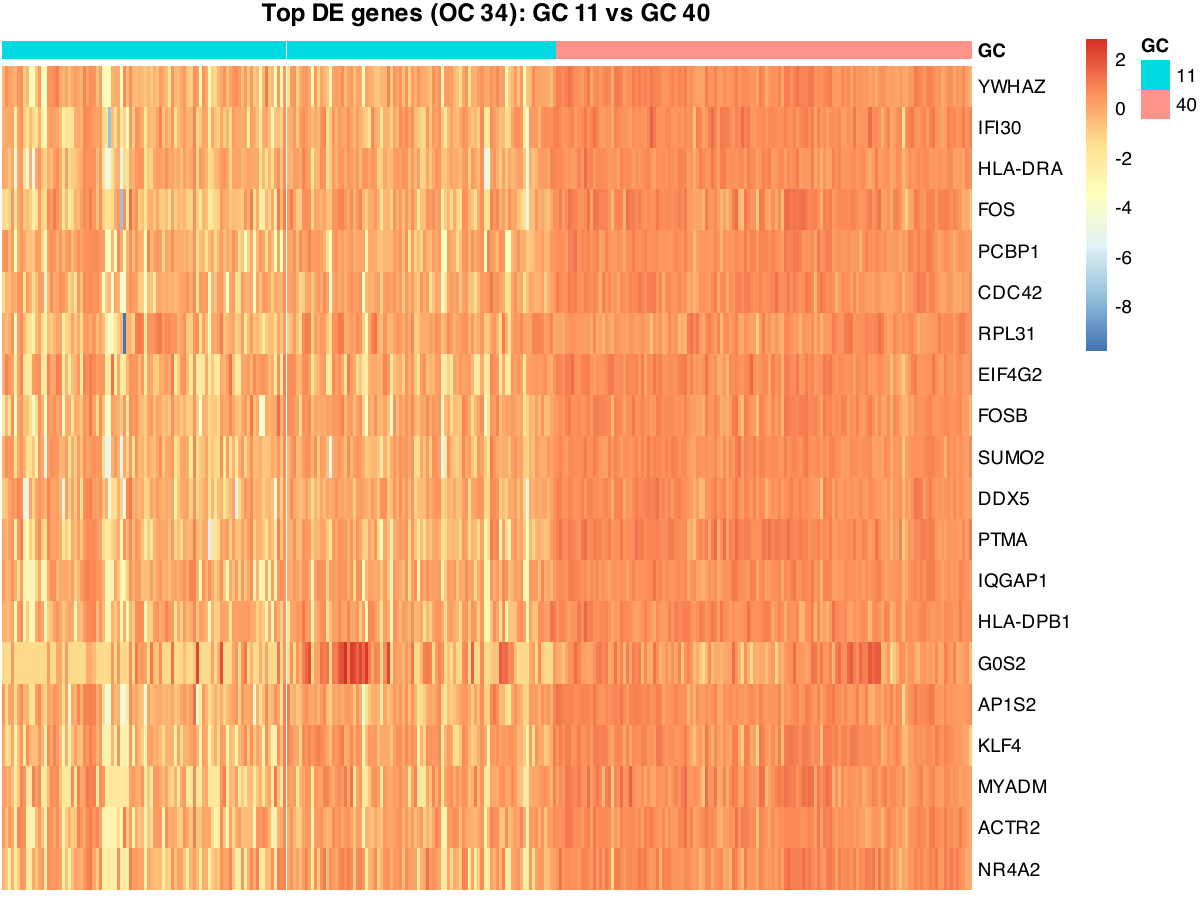}
                \caption{}
                \label{fig:OC34_GC11_vs_GC40_heatmap3}
            \end{subfigure}
        \caption{Heatmaps of the top 20 differentially expressed genes obtained from pseudobulk differential expression analysis comparing individuals assigned to GC 11 and GC 40 within (a) estimated OC 2 and (b) estimated OC 34.}
        \label{fig:GC11_vs_GC40_heatmaps}
        \end{figure}

\begin{figure}[!htp]
                \centerline{\includegraphics[width=0.65\linewidth]{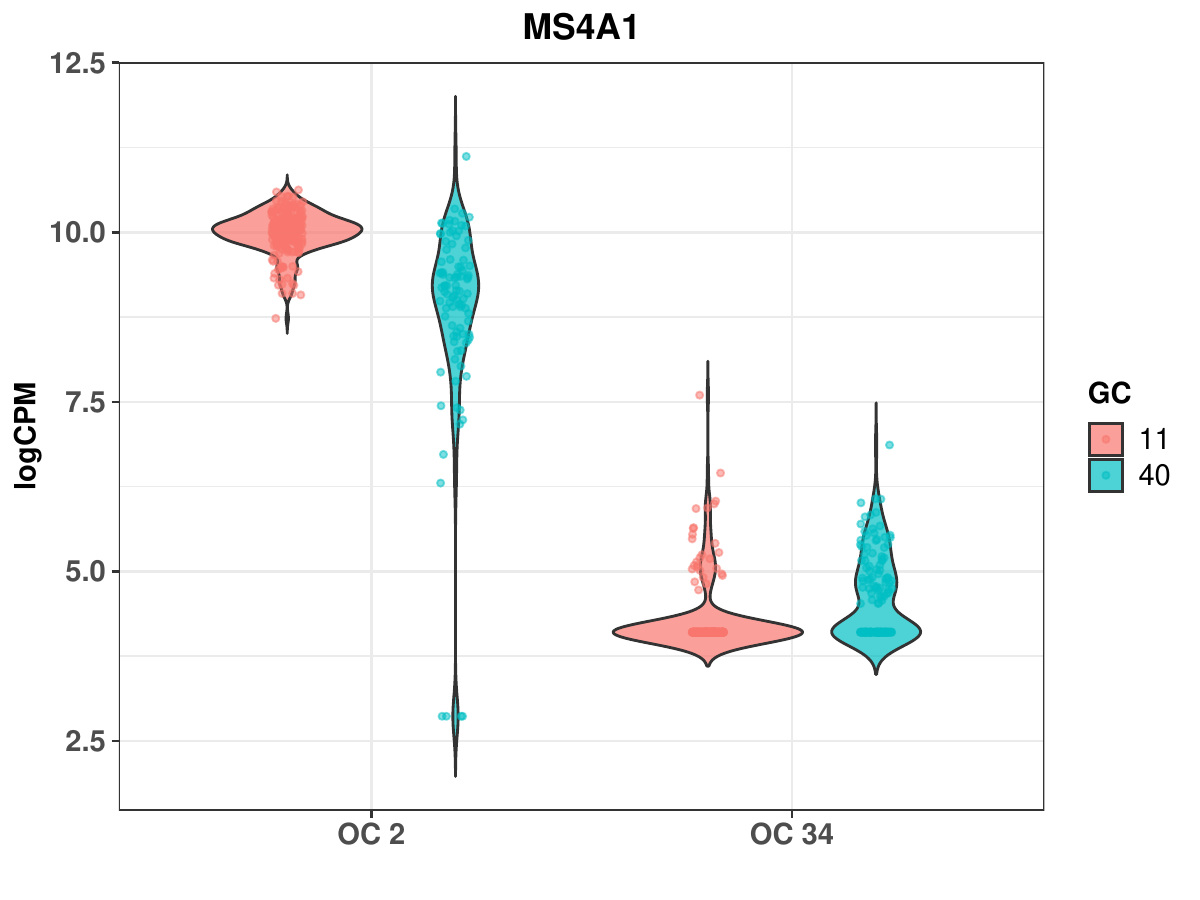
                }}
                \caption{Violin plots of the log counts per million (CPM) for the gene MS4A1, comparing expression levels across two estimated observational clusters (OCs) for individuals assigned to GC 11 and GC 40. 
                }
                \label{fig:MS4A1_OC2_OC34_GC11_vs_GC40_Violin}
                \end{figure}

Finally, we performed Gene Ontology (GO) enrichment analysis to further characterize the biological identity of the estimated OCs. As an illustrative example, Figure~\ref{fig:Indiv131_True_Est_OCs} in the Supplementary Materials presents the UMAP embeddings of the gene expression data for the individual numbered 131, with cells colored according to the estimated OCs alongside the annotated cell-type labels by OneK1K. This visualization allows for a direct qualitative comparison between the clusters inferred by our model and the reference annotations. To investigate the functional characteristics of the specific clusters, we focused on the OCs labeled 2 and 40 for this individual. For each OC, we first identified differentially expressed genes relative to the remaining cells and then selected the top 100 genes after filtering for statistical significance and effect size, retaining only those genes with statistically significant p-values and an average log fold change exceeding 0.25. These genes were subsequently used to perform GO enrichment analysis in the biological process category.

The resulting top five enriched GO terms for OC 2 and OC 40 are shown in Figures~\ref{fig:Individual_131_OC2_DE_Barplot_Dotplot} and~\ref{fig:Individual_131_OC40_DE_Barplot_Dotplot}, respectively. For OC 2, the enriched biological processes include pathways related to B-cell receptor signaling, B-cell activation, and antigen processing and presentation via MHC class II molecules. These processes are well-established functional signatures of B lymphocytes and are consistent with the antigen-recognition and antigen-presenting roles of B cells in adaptive immunity. In contrast, the enriched GO terms for OC 40 predominantly involve T-cell-related biological processes, including T-cell receptor signaling, regulation of T-cell activation, and T-cell differentiation. These pathways are characteristic of T lymphocytes and reflect their role in antigen recognition and immune response coordination through T-cell receptor-mediated signaling. Taken together, the enrichment results provide functional evidence supporting the cellular identities of these clusters. Specifically, OC 2 appears to correspond to a population of B cells, whereas OC 40 corresponds to a population of T cells. Importantly, these functional annotations are consistent with the manually curated cellular subtypes shown in Supplementary Figure~\ref{fig:Indiv131_TrueOCs}, thereby demonstrating that the estimated clusters capture biologically meaningful cellular subtypes present in the data.

In summary, the enrichment analyses demonstrate that GO enrichment analysis applied to the gene signatures of the estimated OCs obtained from the NAM can provide meaningful biological interpretation of the inferred clusters.  In particular, the enriched functional categories offer insight into the underlying cellular identities associated with each OC. Consequently, this approach can serve as a practical strategy for annotating cell types in datasets beyond the OneK1K cohort, particularly in settings where reference cell-type annotations are unavailable. Beyond providing biological validation of the inferred clusters, these analyses also demonstrate how the proposed framework can serve as an automated pipeline for identifying cellular subtypes across individuals, particularly in large-scale datasets where manual annotation is labor-intensive and potentially impractical.


\begin{figure}[!htp]
        \centering
           \begin{subfigure}{0.48\textwidth}
            \centering
            \includegraphics[width= 1\linewidth]{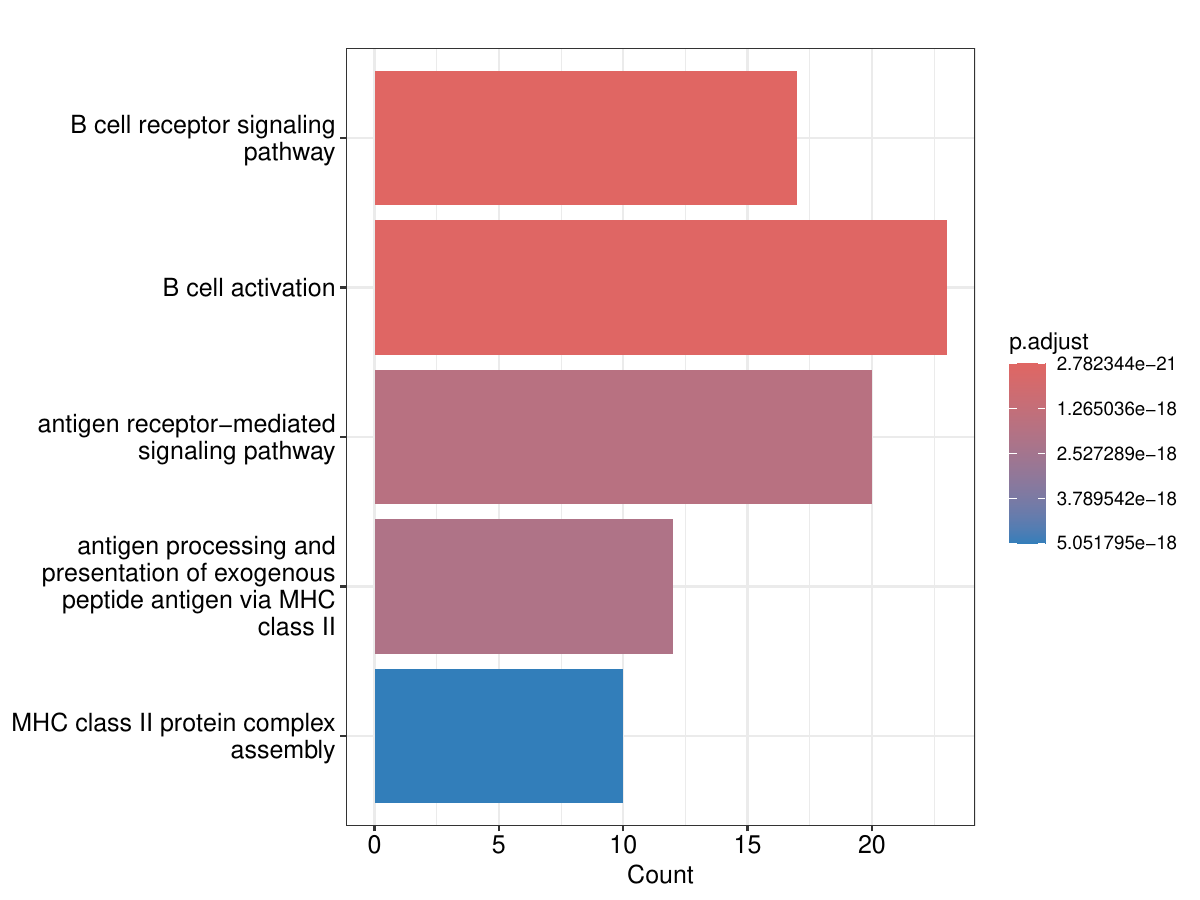}
            \caption{OC 2}
            \label{fig:Individual_131_OC2_DE_Barplot_Dotplot}
        \end{subfigure}
            \begin{subfigure}{0.48\textwidth}
              \centering
            \includegraphics[width= 1\linewidth]{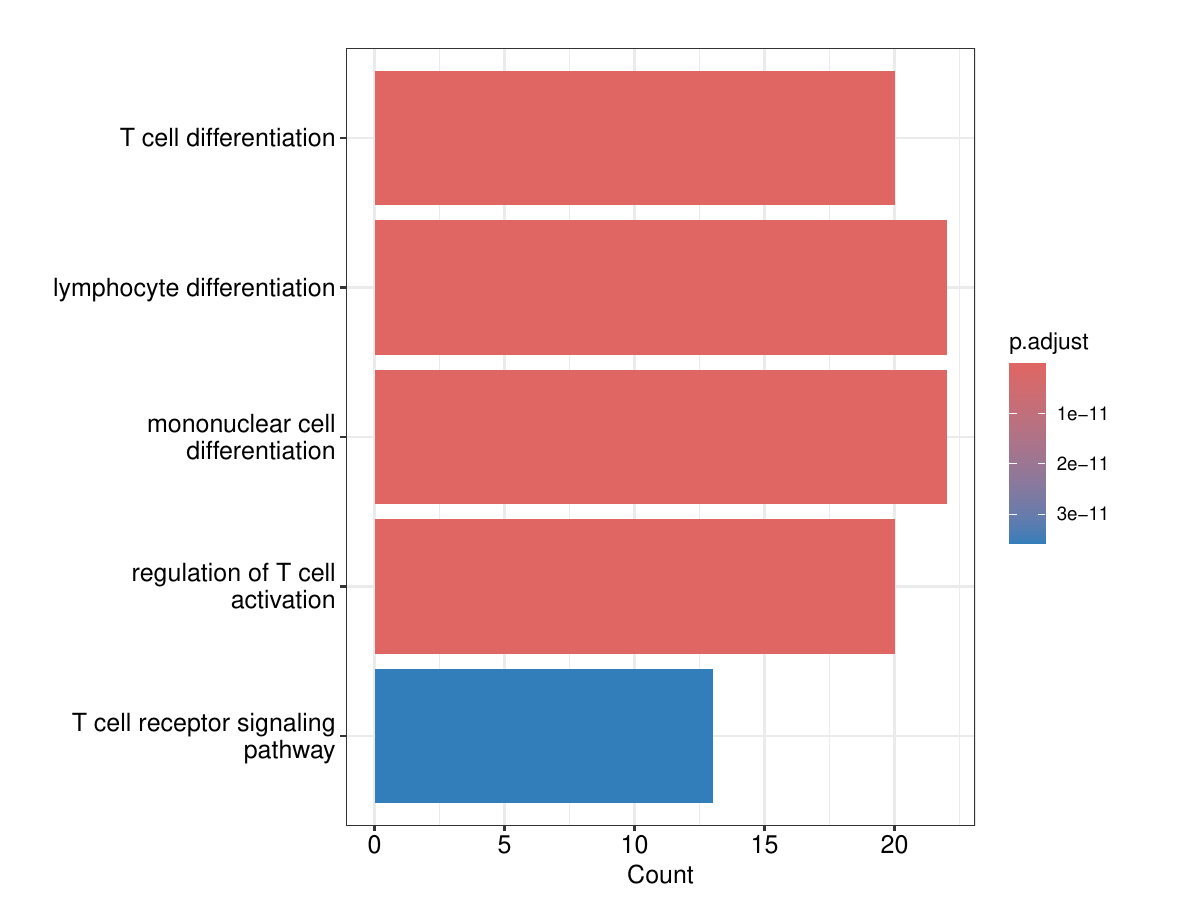}
                \caption{OC 40}
                \label{fig:Individual_131_OC40_DE_Barplot_Dotplot}
            \end{subfigure}
        \caption{Barplot summarizing the Gene Ontology enrichment analysis for the genes differentially expressed in estimated (a) OC 2 and (b) OC 40 corresponding to individual 131. The plot displays the most significantly enriched biological process terms identified from the selected differentially expressed genes associated with this observational cluster.}
        \label{fig:Individual_131_OC2_OC40_DE_Barplot_Dotplot}
        \end{figure}

\section{Conclusion}\label{sec::conclusion}
We have introduced the nested atoms model (NAM) for the simultaneous two-layered clustering of nested datasets. NAM allows for the joint clustering of groups and observations within each group by utilizing both group-level and observation-level data. The observation-level clusters are comparable across groups via the sharing-of-information feature of NAM. 
Posterior inference is carried out through a scalable variational inference algorithm. Simulation studies demonstrate the advantages of the proposed NAM over existing models for clustering nested datasets that do not account for group-level data. The application to the OneK1K data demonstrates the practical utility of the proposed NAM. For example, differential gene expression analyses reveal some biologically meaningful patterns.

For future directions, we may investigate the theoretical properties of the proposed method, such as large-sample theories and its induced partition function. We may also apply the proposed method to other population-scale single-cell data or other nested data to further validate its practical utility. Additionally, for other population-scale single-cell or nested data, if the batch labels of individuals are available, then the batch effects can be modeled within our proposed framework by incorporating them into the group-level variable $\bm{x}$. 

\bibliographystyle{apalike}
\bibliography{references} 

\clearpage
\appendix
\section*{Supplementary Materials for ``Nested Atoms Model with Application to Clustering Big Population-Scale Single-Cell Data''}
\phantomsection
\addcontentsline{toc}{section}{Supplementary Materials}

\beginsupplement
\setcounter{prp}{0}

\section{Proofs of the Properties of the Nested Atoms Model}\label{sec::suppl_proofs_properties}
In this section, we derive the prior mean, variance, co-clustering probabilities, and correlation structure of the nested atoms model (NAM). In the following propositions, we assume that the concentration parameters, $\alpha$ and $\beta$ are fixed.
\begin{prp}\label{suppl:prop1}
Consider the random distribution $G_j$ defined on $(\mathbb{X}\times \mathbb{Y}, \mathcal{X}\otimes\mathcal{Y})$, with $G_j | Q \sim  Q$, where $Q$ is defined in \eqref{eq::NAM2}. Then, for any Borel set $A \in \mathcal{X}\otimes\mathcal{Y}$,
\begin{align*}
    \mathbb{E}\left[ G_j(A)\right] &= H^x(A)H^y(A),\\
    Var\left[ G_j(A)\right] & = H^x(A)H^y(A)\left[q_2 + H^y(A) - q_2H^y(A) - H^x(A) H^y(A) \right],
\end{align*}
where $q_2 = 1/{(1+ \beta)}$.
\end{prp}
\begin{proof}
First, note that from simple calculations using the stick-breaking construction of $\bm{\pi}$ and $\bm{\omega}_k$, for any $k\geq 1$, 
\begin{align*}
    \sum_{k=1}^{\infty}\mathbb{E}\left[\pi_k\right] & = 1, &  \sum_{k=1}^{\infty}\mathbb{E}\left[\pi_k^2\right] & = \frac{1}{1+\alpha}, \\
    \sum_{l=1}^{\infty}\mathbb{E}\left[\omega_{lk}\right] & = 1, &  \sum_{l=1}^{\infty}\mathbb{E}\left[\omega_{lk}^2\right] & = \frac{1}{1+\beta}.
\end{align*}
Then, it follows that \begin{align}
    \nonumber \mathbb{E}\left[G_j(A)\right] & = \mathbb{E}\ \mathbb{E}\left[G_j(A) \mid Q\right]\\
    \nonumber & = \mathbb{E}\left[ \sum_{k=1}^{\infty}\pi_k \delta_{\theta_k^x}(A) G^{y*}_k(A)\right]\\
    \nonumber & = \sum_{k=1}^{\infty} \mathbb{E}(\pi_k) H^x(A)\, \mathbb{E}\left[G^{y*}_k(A)\right]\\
    & = H^x(A)H^y(A),
\end{align}
as $\mathbb{E}\left[G^{y*}_k(A)\right] = \mathbb{E}\left[\sum_{l=1}^{\infty}\omega_{lk} \delta_{\theta_l^y} \right] = H^y(A)$.
Next, we see that 
 \begin{align*}
    \mathbb{E}\left[G^{y*2}_k(A)\right] & = \mathbb{E}\left[\sum_{l=1}^\infty \omega_{lk}^2 \delta^2_{\theta_l^y}(A) + \sum_{l\neq l'}\omega_{lk}\omega_{l'k}\delta_{\theta_l^y}(A)\delta_{{\theta_{l'}^y}}(A)\right]\\
    & = H^y(A) \sum_{l=1}^\infty\mathbb{E}\left[ \omega_{lk}^2\right] + (H^y(A))^2\mathbb{E}\left[\sum_{l\neq l'}\omega_{lk}\omega_{l'k}\right]\\
    & = H^y(A) \sum_{l=1}^\infty\mathbb{E}\left[ \omega_{lk}^2\right] + (H^y(A))^2\left[1 - \sum_{l=1}^\infty\mathbb{E}\left[ \omega_{lk}^2\right]\right]\\
    & = H^y(A)\frac{1}{1+\beta} + \frac{\beta}{1 + \beta} (H^y(A))^2.
\end{align*}
Hence, \begin{align*}
    \mathbb{E}\left[(G_j(A))^2\right] & = \mathbb{E}\ \mathbb{E}\left[(G_j(A))^2\mid Q\right]\\
    & = \mathbb{E}\left[ \sum_{k=1}^{\infty}\pi_k \delta^2_{\theta_k^x}(A) G^{y*2}_k(A)\right]\\
    & = \sum_{k=1}^{\infty} \mathbb{E}(\pi_k) H^x(A)\, \mathbb{E}\left[G^{y*2}_k(A)\right]\\
    & = H^x(A)\left[H^y(A)\frac{1}{1+\beta} + \frac{\beta}{1 + \beta} (H^y(A))^2\right].
\end{align*}

Therefore, 
 \begin{align}
    \nonumber Var\left[G_j(A)\right] & =  \mathbb{E}\left[(G_j(A))^2\right]  - \mathbb{E}^2\left[G_j(A)\right]\\
    \nonumber & = H^x(A)\left[H^y(A)\frac{1}{1+\beta} + \frac{\beta}{1 + \beta} (H^y(A))^2\right] - \left[H^x(A)H^y(A)\right]^2\\
    \nonumber & = H^x(A)H^y(A)\left[\frac{1}{1+\beta} + \frac{\beta H^y(A)}{1+\beta} - H^x(A)H^y(A)\right]\\
    & = H^x(A)H^y(A)\left[\frac{1}{1+\beta} + H^y(A) -\frac{H^y(A)}{1+\beta} - H^x(A)H^y(A)\right],
\end{align}
which concludes the proof.
\end{proof}

\clearpage\newpage
\begin{prp}
Consider two random distributions $G_j$ and $G_{j'}, \ j' \neq j$, defined on $(\mathbb{X}\times \mathbb{Y}, \mathcal{X}\otimes\mathcal{Y})$, with $G_j,\,G_{j'} \mid Q \overset{iid}{\sim} Q$ and $Q$ is defined by \eqref{eq::NAM2}. Additionally, consider two observations $\bm{z}_{ji} \mid G_j \sim G_j$ and $\bm{z}_{j'i'} \mid G_{j'} \sim G_{j'}$, where $i' \neq i$, $\bm{z}_{ji}= (\bm{x}_{j}, \bm{y}_{ji})$, and $\bm{z}_{j'i'}= (\bm{x}_{j'}, \bm{y}_{j'i'})$. Then, the prior co-clustering probabilities under the NAM are given by,
\begin{align*}
    P\left[G_j = G_{j'}\right] & = \frac{1}{1+ \alpha},\\
    P\left[\bm{z}_{ji} = \bm{z}_{j'i'}\right] & = \frac{1}{1+\alpha}\left[\frac{1}{1+\beta} + \frac{\alpha}{2\beta + 1}\right].
\end{align*}
\end{prp}
\begin{proof}
    The proof follows by mimicking the derivation of prior co-clustering probabilities for the common atoms model (CAM; \citealp{denti2023common}).
\end{proof}

\begin{prp}
Consider two random distributions $G_j$ and $G_{j'}, \ j' \neq j$, defined on $(\mathbb{X}\times \mathbb{Y}, \mathcal{X}\otimes\mathcal{Y})$, with $G_j,\,G_{j'} \mid Q \overset{iid}{\sim} Q$ and $Q$ is defined by \eqref{eq::NAM2}. Then, the prior correlation between the two random measures evaluated on a same Borel set $A \in \mathcal{X}\otimes\mathcal{Y}$ is given by,
\begin{align*}
        \rho_{jj'}^{NAM}(A) = q_1 + \frac{q_3(1-q_1) H^x(A)}{\left[q_2 + \frac{H^y(A)}{(1-H^y(A))} (1- H^x(A))\right]},
\end{align*}
where $q_1 = 1/{(1+\alpha)}, \: q_2 = 1/{(1+\beta)},$ and $q_3 = 1/{(1+ 2\beta)}$.
\end{prp}
\begin{proof}
For generality, first, we consider two Borel sets $A, B \in \mathcal{X}\otimes\mathcal{Y}$. Then, by definition,
\begin{equation*}
    cov\left[ G_j(A), G_{j'}(B) \right] = \mathbb{E}\left[ G_j(A) G_{j'}(B)\right] - \mathbb{E}\left[ G_j(A)\right] \mathbb{E}\left[ G_{j'}(B)\right]. 
\end{equation*}
From Proposition \ref{prop1}, $\mathbb{E}\left[ G_j(A)\right] = H^x(A)H^y(A)$ and similarly, $\mathbb{E}\left[ G_{j'}(B)\right] = H^x(B)H^y(B)$.
Now,
\begin{align}
    \nonumber \mathbb{E}\left[ G_j(A) G_{j'}(B)\right] & = \mathbb{E}\, \mathbb{E}\left[ G_j(A) G_{j'}(B)\mid Q\right]\\
    \nonumber & = \mathbb{E}\left[ \sum_{k=1}^{\infty} \pi_k^2 \delta_{\theta_k^x}(A)\delta_{\theta_k^x}(B) G_k^{y*}(A)G_k^{y*}(B) \right.\\
    \nonumber & \qquad \qquad \left. + \sum_{k_1\neq k_2} \pi_{k_1}\pi_{k_2} \delta_{\theta_{k_1}^x}(A)\delta_{\theta_{k_2}^x}(B) G_{k_1}^{y*}(A)G_{k_2}^{y*}(B)\right]\\
    \nonumber & = H^x(A\cap B) \sum_{k=1}^{\infty} \mathbb{E}\left[\pi_k^2\right]  \mathbb{E}\left[G_k^{y*}(A)G_k^{y*}(B)\right]  \\
    \label{eq:EGjGj'}& \qquad \qquad  + H^x(A)H^x(B)\mathbb{E}\left[\sum_{k_1\neq k_2} \pi_{k_1}\pi_{k_2}\right] \mathbb{E}\left[G_{k_1}^{y*}(A)G_{k_2}^{y*}(B)\right]
\end{align}
First, note that, 
\begin{align}
    \nonumber \mathbb{E}\left[G_k^{y*}(A)G_k^{y*}(B)\right]  & = H^y(A\cap B)\sum_{k=1}^{\infty}\mathbb{E}\left[\omega_{lk}^2\right] +  H^y(A)H^y(B)\:\mathbb{E}\left[ \sum_{l\neq l'}\omega_{lk}\omega_{l'k}\right]\\
    \nonumber & = \frac{1}{1+\beta} H^y(A\cap B) + H^y(A)H^y(B) \left[ 1- \sum_{k=1}^{\infty}\mathbb{E}\left[\omega_{lk}^2\right] \right]\\
    \label{eq:EGkGk}& = \frac{1}{1+\beta} H^y(A\cap B) + \frac{\beta}{1 + \beta}H^y(A)H^y(B). 
\end{align}
Next, 
\begin{align}
    \nonumber \mathbb{E}\left[G_{k_1}^{y*}(A)G_{k_2}^{y*}(B)\right]  & = \mathbb{E}\left[ \sum_{l=1}^{\infty}\omega_{lk_1}\delta_{\theta_l^y}(A) \sum_{l'=1}^{\infty}\omega_{l'k_2}\delta_{\theta_{l'}^y}(B)\right]\\
    \nonumber & = \mathbb{E}\left[ \sum_{l=1}^{\infty}\omega_{lk_1} \omega_{lk_2}\delta_{\theta_l^y}(A)\delta_{\theta_l^y}(B) +  \sum_{l\neq l'}\omega_{lk_1}\omega_{l'k_2}\delta_{\theta_{l}^y}(A)\delta_{\theta_{l'}^y}(B)\right]\\\
    \nonumber  & = H^y(A\cap B)\left[\sum_{l=1}^{\infty} \mathbb{E}[\omega_{lk_1}] \mathbb{E}[\omega_{lk_2}]\right] +  H^y(A) H^y(B)\left[\mathbb{E}\left\{\sum_{l\neq l'}\omega_{lk_1}\omega_{l'k_2}\right\}\right]\\
     \nonumber & \overset{iid}{=} H^y(A\cap B)\left[\sum_{l=1}^{\infty} \left(\mathbb{E}[\omega_{lk}]\right)^2\right] +  H^y(A) H^y(B)\left[1 - \sum_{l=1}^{\infty} \left(\mathbb{E}[\omega_{lk}]\right)^2 \right]\\
     \label{eq:EGk1Gk2}& = \frac{1}{1+ 2\beta} H^y(A \cap B) + \frac{2\beta}{1+ 2\beta} H^y(A)H^y(B),
\end{align}
where the last line follows again from straightforward calculations using the stick-breaking representation of $\bm{\omega}_k$. Thus, plugging in equations \eqref{eq:EGkGk} and \eqref{eq:EGk1Gk2} in \eqref{eq:EGjGj'}, we get,
\begin{align}
    \nonumber \mathbb{E}\left[ G_j(A) G_{j'}(B)\right] & = H^x(A\cap B)\left[ \frac{1}{1+\beta} H^y(A\cap B) + \frac{\beta}{1 + \beta}H^y(A)H^y(B)\right] \sum_{k=1}^{\infty} \mathbb{E}\left[\pi_k^2\right]  \\
    \nonumber & \  + H^x(A)H^x(B) \left[\frac{1}{1+ 2\beta} H^y(A \cap B) + \frac{2\beta}{1+ 2\beta} H^y(A)H^y(B) \right]\mathbb{E}\left[\sum_{k_1\neq k_2} \pi_{k_1}\pi_{k_2}\right]\\
    \nonumber & = H^x(A\cap B)\left[ q_2 H^y(A\cap B) + (1-q_2)H^y(A)H^y(B)\right]  q_1  \\
    & \  + H^x(A)H^x(B) \left[q_3 H^y(A \cap B) + (1-q_3) H^y(A)H^y(B) \right](1-q_1)
\end{align}
Finally, we get,
\begin{align*}
    cov\left[ G_j(A), G_{j'}(B) \right] & = q_1H^x(A\cap B)\left[ q_2 H^y(A\cap B) + (1-q_2)H^y(A)H^y(B)\right] \\
    & \  + (1-q_1)H^x(A)H^x(B) \left[q_3 H^y(A \cap B) + (1-q_3) H^y(A)H^y(B) \right] \\
    &  \ - H^x(A)H^x(B)H^y(A)H^y(B)
\end{align*}
and a bit of algebra gives,
\begin{align*}
    cov\left[ G_j(A), G_{j'}(B) \right] & = q_1\bigg[ H^x(A\cap B) \left\{q_2H^y(A\cap B) + (1-q_2) H^y(A)H^y(B)\right\} \\
    & \quad \qquad - H^x(A)H^x(B)\left\{ q_3 H^y(A\cap B) + (1-q_3) H^y(A)H^y(B)\right\}\bigg]\\
    & \quad + q_3 H^x(A)H^x(B) \left\{H^y(A\cap B) - H^y(A)H^y(B)\right\}.
\end{align*}
For the same Borel set $A$, we get,
\begin{align}
    \nonumber cov\left[ G_j(A), G_{j'}(A) \right] & = q_1\bigg[ H^x(A) \left\{q_2H^y(A) + (1-q_2) (H^y(A))^2\right\}  \\
    \nonumber & \quad\qquad - (H^x(A))^2\left\{ q_3 H^y(A) + (1-q_3) (H^y(A))^2\right\}\bigg]\\
    \nonumber & \quad + q_3 (H^x(A))^2) \left\{H^y(A) - (H^y(A))^2\right\}\\
    \nonumber & = H^x(A)H^y(A)\bigg[ q_1 \bigg(q_2 + H^y(A) - q_2 H^y(A) - H^x(A)H^y(A)\bigg) \\
    \label{eq:CovGjAGj'A}& \qquad \qquad \qquad \qquad + q_3(1-q_1)H^x(A)(1-H^y(A)) \bigg]
\end{align}
Finally,  as
\begin{align*}
    \rho_{jj'}^{NAM}(A) & = corr\left[ G_j(A), G_{j'}(A) \right]\\
    & = \frac{cov\left[ G_j(A), G_{j'}(A) \right]}{\sqrt{Var\left[G_j(A)\right]}\: \sqrt{Var\left[G_{j'}(A)\right]}}\\
    & = q_1 + \frac{q_3(1-q_1)H^x(A)(1- H^y(A))}{\left(q_2 + H^y(A) - q_2 H^y(A) - H^x(A)H^y(A)\right)},
\end{align*}
where the last line follows from \eqref{eq:CovGjAGj'A} and the expression of prior variance in Proposition \ref{prop1}. Finally, the proof follows from some algebra.
\end{proof}

\clearpage
\newpage
\section{Graphical Representation of the NAM Mixture Model}\label{supp::subsec:NAMM_Graphical}

In this section, we present the graphical model representation of the NAM mixture model (Figure~\ref{fig:NAM_Graphical_Rep}).

\begin{figure}[ht]
\centering
    \begin{tikzpicture}[
    scale=0.8,
    transform shape,
    >=Latex,
    every node/.style={font=\Large},
    circ/.style={draw,circle,minimum size=12mm,inner sep=0pt},
    plate/.style={draw,rectangle,inner sep=10pt},
    latent/.style={draw, rectangle, minimum width=1.1cm, minimum height=1.1cm, rounded corners}
    ]
    
    \node[plate, minimum width=15cm, minimum height=9.2cm] (plateJ) {};
    \node[anchor=west] at ($(plateJ.south east)+(-0.75cm, 0.55cm)$) {$J$};
    
    \node[circ] (Sj) at ($(plateJ.north west)+(2.6cm,-1.85cm)$) {$S_j$};
    
    \node[
    draw, rectangle,
    minimum width=1cm,
    minimum height=1cm,
    fill=gray!15
    ] (xj) at ($(Sj)+(0cm,-2.8cm)$) {$\bm{x}_j$};
    
    \node[plate, minimum width=6.5cm, minimum height=7.2cm, anchor=north west] (plateI)
    at ($(plateJ.north west)+(6.6cm,-1.15cm)$) {};
    \node[anchor=west] at ($(plateI.south east)+(-0.85cm,0.55cm)$) {$n_j$};
    
    \node[circ] (Mji) at ($(plateI.north)+(0cm,-2.45cm)$) {$M_{ji}$};
    
    \node[
    draw, rectangle,
    minimum width=1cm,
    minimum height=1cm,
    fill=gray!15
    ] (yji) at ($(Mji)+(0cm,-2.8cm)$) {$\bm{y}_{ji}$};
    
    \node[circ] (pi)      at ($(Sj)+(0cm, 3.45cm)$) {$\bm\pi$};
    \node[circ] (theta_x) at ($(pi)+(-2cm, 0cm)$) {$\theta_k^x$};
    \node[circ] (beta)    at ($(pi)+(0cm, 2cm)$) {$\beta$};
    
    \node[latent] (H_x)   at ($(theta_x)+(0cm, 2cm)$) {\Large $H^x$};
    
    \node[circ] (omega)   at ($(pi)+(7.25cm, 0cm)$) {$\bm{\omega}_k$};
    \node[circ] (theta_y) at ($(omega)+(-2cm, 0cm)$) {$\theta_l^y$};
    \node[circ] (alpha)   at ($(omega)+(0cm, 2cm)$) {$\alpha$};
    
    \node[latent] (H_y)   at ($(theta_y)+(0cm, 2cm)$) {\Large $H^y$};
    
    \draw[->, line width=1.5pt, -{Latex[length=2mm, width=2mm]}] (H_x) -- (theta_x);
    \draw[->, line width=1.5pt, -{Latex[length=2mm, width=2mm]}] (beta) -- (pi);
    \draw[->, line width=1.5pt, -{Latex[length=3mm, width=3mm]}] (pi) -- (Sj.north);
    \draw[->, dashed, line width=2pt, -{Latex[length=3mm, width=3mm]}] (theta_x)
    to[out=270, in=180] (xj);
    
    \draw[->, line width=1.5pt, -{Latex[length=2mm, width=2mm]}] (H_y) -- (theta_y);
    \draw[->, line width=1.5pt, -{Latex[length=2mm, width=2mm]}] (alpha) -- (omega);
    \draw[->, line width=1.5pt, -{Latex[length=3mm, width=3mm]}] (omega) -- (Mji.north);
    \draw[->, dashed, line width=2pt, -{Latex[length=3mm, width=3mm]}] (theta_y)
    to[out=270, in=180] (yji);
    
    \draw[->, line width=2pt, -{Latex[length=3mm, width=3mm]}, dashed] (Sj) -- (xj);
    \draw[->, line width=1.5pt, -{Latex[length=3mm, width=3mm]}] (Sj) to[out=360, in=140] (Mji);
    \draw[->, line width=2pt, -{Latex[length=3mm, width=3mm]}, dashed] (Mji) -- (yji);
    
    \node[
    draw,
    rectangle,
    fit=(theta_x),
    inner sep=8.5pt
    ] (hyperplate1) {};
    \node[anchor=south east, font=\small]
    at ($(hyperplate1.south east)+(0.1pt,2pt)$)
    {$\infty$};
    
    \node[
    draw,
    rectangle,
    fit=(theta_y),
    inner sep=8.5pt
    ] (hyperplate2) {};
    \node[anchor=south east, font=\small]
    at ($(hyperplate2.south east)+(0.1pt,2pt)$)
    {$\infty$};
    
    \node[plate = $\infty$
    draw,
    rectangle,
    fit=(omega),
    inner sep=8.5pt
    ] (hyperplate3) {};
    \node[anchor=south east, font=\small]
    at ($(hyperplate3.south east)+(0.1pt,2pt)$)
    {$\infty$};
    
    \end{tikzpicture}
\caption{Graphical representation of the NAM mixture model. Each node in the graph is associated with a random variable, where shaded rectangle denotes an observed variable. Rectangular plates denote replication of the model within the rectangle. The dashed lines represent the random variables that ultimately impact the observed variable. }
\label{fig:NAM_Graphical_Rep}
\end{figure}
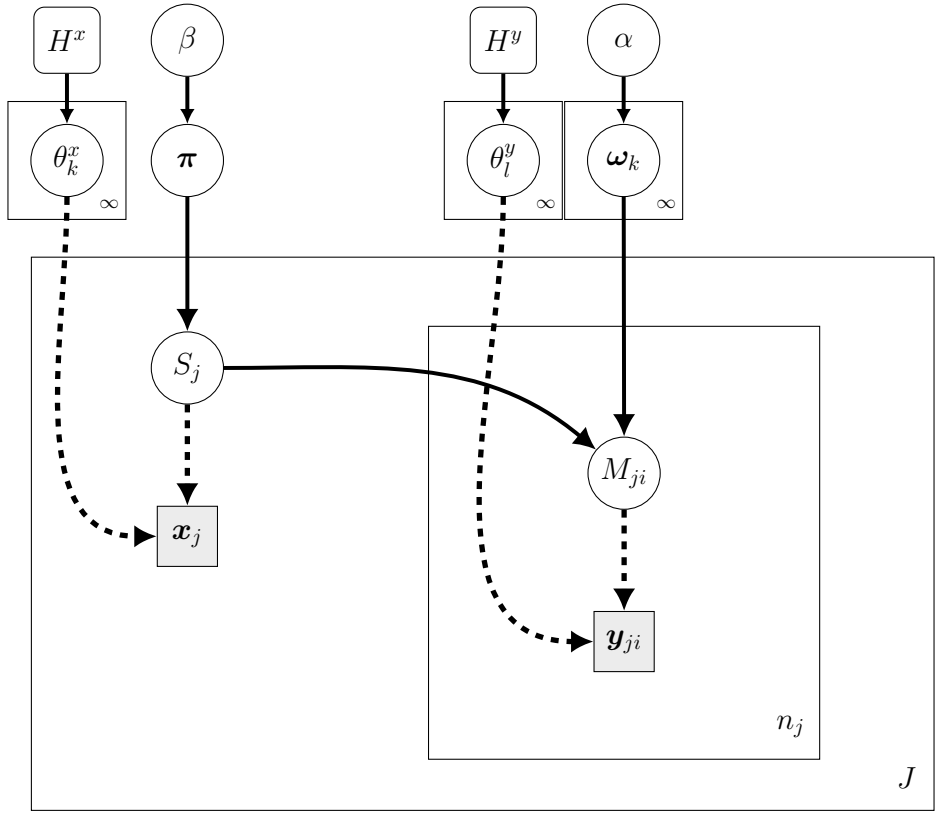

\clearpage\newpage
\section{Coordinate-Ascent Variational Inference Algorithm for Nested Atoms Model}\label{subsec:CAVI}
In this section, we detail the proposed coordinate-ascent variational inference (CAVI) algorithm for the NAM (given in Algorithm~\ref{alg:NAM_CAVI}). Particularly, we extend the grouped-data variational inference (VI) framework of \citealp{fiSAN} to our nested data setting. Additionally, note that since the CAM mixture model is a special case of the proposed NAM framework, Algorithm~\ref{alg:NAM_CAVI} can be readily adapted to facilitate posterior inference in the CAM without requiring substantial modifications. Specifically, in Step 1 of Algorithm~\ref{alg:NAM_CAVI}, removing the term $\frac{1}{2} \left(\ell^{(x,1)}_{k} +\ell^{(x,2)}_{jk}\right)$ and omitting the update of the variational parameter in Step 5 leads to a CAVI procedure tailored for the CAM.
\begin{breakablealgorithm}
 \caption{CAVI updates for the NAM}
 \label{alg:NAM_CAVI}
{\footnotesize
\hspace*{-6.5em}\textbf{Input:} $t \gets 0$. Randomly initialize $\boldsymbol{\Lambda}^{(0)}$. Define the threshold $\epsilon$ and randomly set $\Delta>\epsilon$.
\begin{algorithmic}
    \WHILE {$\Delta(t-1,t) > \varepsilon$}
    \STATE \hspace*{-1em}Set $t = t+1$; Let $\boldsymbol{\Lambda}^{(t-1)} = \boldsymbol{\Lambda}^{(t)};$ \\
    Update the variational parameters according to the following CAVI steps:
    \begin{enumerate}
    \item For $j=1,\dots,J$, $q^\star(S_j)$ is a $K$-dimensional multinomial, with $q^\star(S_j=k)=\rho_{jk}$ for $k=1,\dots,K$,
      \begin{align*}
       \log\rho_{jk} = 
            g(\bar{a}_k,\bar{b}_k) + \sum_{r=1}^{k-1}g(\bar{b}_r,\bar{a}_r) & +  \sum_{l=1}^{L} \left[ \left(\sum_{i=1}^{n_j} \xi_{jil}\right)\left\{ g(\bar{a}_{lk}, \bar{b}_{lk}) + \sum_{t=1}^{l-1} g(\bar{b}_{tk}, \bar{a}_{tk})\right\}\right] + \\
       &  +  \frac{1}{2} \left(\ell^{(x,1)}_{k} +\ell^{(x,2)}_{jk}\right) ,
    \end{align*}
    where 
    \begin{align*}
    g(x,y) & = \psi(x) - \psi(x+y), \ \ \text{with $\psi$ denoting the digamma function}, \\
    \ell^{(x,1)}_{k} & =  \sum_{i=1}^q\psi\left( (c^x_k-i+1)/2\right) +q\log2 + \log|\boldsymbol{D}^x_k|, \ \ \text{and}\\
    \ell^{(x,2)}_{jk} & = - q/t^x_k - c^x_k  (\bm{x}_j-\bm{m}^x_k)^T \boldsymbol{D}^x_k(\bm{x}_j-\bm{m}^x_k)
    \end{align*}
    \item For $j=1,\dots,J$ and $i=1,\dots,n_j$, $q^\star(M_{ji})$ is a $L$-dimensional multinomial, with $q^\star(M_{ji}=l)=\xi_{jil}$ for $l=1,\dots,L$,
    \begin{align*}
       \log\xi_{jil} = \frac{1}{2} \left(\ell^{(y,1)}_{l} +\ell^{(y,2)}_{jil}\right) + \sum_{k=1}^{K}\rho_{jk}\left\{g(\bar{a}_{lk}, \bar{b}_{lk}) + \sum_{t=1}^{l-1}g(\bar{b}_{tk}, \bar{a}_{tk})\right\},
    \end{align*}
    where 
    \begin{align*}
    \ell^{(y,1)}_{l} & =  \sum_{i=1}^p\psi\left( (c^y_l-i+1)/2\right) +p\log2 + \log|\boldsymbol{D}^y_l| \ \ \text{and}\\
    \ell^{(y,2)}_{jil} & = - p/t^y_l - c^y_l  (\bm{y}_{ji}-\bm{m}^y_l)^T \boldsymbol{D}^y_l(\bm{y}_{ji}-\bm{m}^y_l)
    \end{align*}
    
    \item For $k=1,\dots,K$ and $l=1,\ldots,L-1$, $q^\star(u_{lk})$ is a $Beta(\bar{a}_{lk},\bar{b}_{lk})$ distribution with $$
    \bar{a}_{lk} = 1 + \sum_{j=1}^{J}\sum_{i=1}^{n_j}\xi_{jit}, \quad \bar{b}_{lk} = r_1/r_2 + \sum_{j=1}^{J}\left\{\rho_{jk} \left( \sum_{i=1}^{n_j}\sum_{t=l+1}^{L}\xi_{jit} \right) \right\}.$$
    
    \item For $k=1,\dots,K-1$, $q^\star(v_{k})$ is a $Beta(\bar{a}_k,\bar{b}_k)$ distribution with $$
        \bar{a}_{k} = 1 + \sum_{j=1}^J\rho_{jk}, \quad \bar{b}_{k} = s_1/s_2 + \sum_{j=1}^J \sum_{l=k+1}^{K}\rho_{jl}.$$
 
    \item   For $k=1,\dots,K$, {$q^\star(\bm{\mu}^x_k, \bm{\Lambda}^x_k)$} is a $\mathrm{NW}(\bm{m}^x_k,t^x_k,c^x_k,\boldsymbol{D}^x_k)$ distribution with parameters
    \begin{align*}
        \bm{m}^x_k &= t_k^{x^{-1}}(\lambda^x_0\:\boldsymbol{\mu}^x_0+ N^x_{k}\bar{\bm{x}}_k),  \quad 
        t^x_k = \lambda^x_0 + N^x_{k}, \quad 
        c^x_k = \nu^x_0 + N^x_{k},\\
        \boldsymbol{D}_k^{x^{-1}} &= \boldsymbol{\Psi}_0^{x^{-1}} +\frac{\lambda^x_0 N^x_{k}}{\lambda^x_0 + N^x_{k}} \left\{\left(\bm{\bar{x}}_k-\bm{\mu}^x_0\right) \left(\bm{\bar{x}}_k-\bm{\mu}^x_0\right)^T\right\} + \bm{\mathcal{S}}_k^x,
    \end{align*}
    where 
    $$\begin{gathered}
        N^x_{k} = \sum_{j=1}^J\rho_{jk}, \qquad \bm{\bar{x}}_k = N_{k}^{x^{-1}} \left(\sum_{j=1}^J\rho_{jk}\:\bm{x}_j\right), \\
        \bm{\mathcal{S}}_{k}^x = \sum_{j=1}^J\rho_{jk}\: \left\{\left(\bm{x}_j - \bm{\bar{x}}_k\right)  \left(\bm{x}_j - \bm{\bar{x}}_k\right)^T\right\}.
    \end{gathered}$$

    \item   For $l=1,\dots,L$, {$q^\star(\bm{\mu}^y_l, \bm{\Lambda}^y_l)$} is a $\mathrm{NW}(\bm{m}^y_l,t^y_l,c^y_l,\boldsymbol{D}^y_l)$ distribution with parameters
    \begin{align*}
        \bm{m}^y_l &= t_l^{y^{-1}}(\lambda^y_0\:\boldsymbol{\mu}^y_0+ N^y_{l}\bar{\bm{y}}_l),  \quad 
        t^y_l = \lambda^y_0 + N^y_{l}, \quad 
        c^y_l = \nu^y_0 + N^y_{l},\\
        \boldsymbol{D}_l^{y^{-1}} &= \boldsymbol{\Psi}_0^{y^{-1}} +\frac{\lambda^y_0 N^y_{l}}{\lambda^y_0 + N^y_{l}} \left\{\left(\bm{\bar{y}}_l-\bm{\mu}^y_0\right)\left(\bm{\bar{y}}_l-\bm{\mu}^y_0\right)^T\right\} + \bm{\mathcal{S}}_l^y,
    \end{align*}
    where 
    $$\begin{gathered}
        N^y_l = \sum_{j=1}^{J}\sum_{i=1}^{n_j}\xi_{jil}, \qquad \bm{\bar{y}}_l = N_l^{y^{-1}} \sum_{j=1}^{J}\sum_{i=1}^{n_j}\xi_{jil}\:\bm{y}_{ji}, \\
        \bm{\mathcal{S}}^y_l = \sum_{j=1}^{J}\sum_{i=1}^{n_j}\xi_{jil}\:\left\{ \left(\bm{y}_{ji}-\bm{\bar{y}}_{l}\right)  \left(\bm{y}_{ji}-\bm{\bar{y}}_{l}\right)^T\right\}.
    \end{gathered}$$

    \item  $q^\star(\alpha)$ is a $\mathrm{Gamma}(s_1,s_2)$ distribution with parameters
    $$s_1 = a_\alpha + (K-1), \quad s_2 = b_\alpha - \sum_{k=1}^{K-1}g(\bar{b}_{k},\bar{a}_{k}).$$

     \item  $q^\star(\beta)$ is a $\mathrm{Gamma}(r_1,r_2)$ distribution with parameters
    $$r_1 = a_\beta + K(L-1), \quad r_2 = b_\beta - \sum_{l=1}^{L-1}\left\{\sum_{k=1}^{K} g(\bar{b}_{lk},\bar{a}_{lk})\right\}.$$
   \end{enumerate}
    Store the updated parameters in $\boldsymbol{\Lambda}$ and let $\boldsymbol{\Lambda}^{(t)}=\boldsymbol{\Lambda}$;
    \\
    Compute $\Delta(t-1,t) = \mathrm{ELBO}(\boldsymbol{\Lambda}^{(t)})- \mathrm{ELBO}(\boldsymbol{\Lambda}^{(t-1)})$.\\ 
 \ENDWHILE
 \STATE \textbf{Return} $\boldsymbol{\Lambda}^\star$, containing the optimized variational parameters.
\end{algorithmic}
}
\end{breakablealgorithm}

\clearpage\newpage
\section{Computation of the Evidence Lower Bound in the variational inference approach}\label{subsec::elbo}
Here we outline the  evidence lower bound (ELBO) evaluation for the NAM.
Recall that we use the notation $g(x,y) = \psi(x) - \psi(x+y)$. The minimization of the Kullback-Leibler divergence between the posterior and the variational distributions is equivalent to the maximization of the ELBO, expressed as
$$ELBO(q)=\mathbb{E}_q\left[\log p(\bm{y}, \bm{x},\boldsymbol{\Theta})\right] - \mathbb{E}_q\left[\log q_{\boldsymbol{\Lambda}}(\boldsymbol{\Theta})\right].$$

The first term, $\mathbb{E}_q\left[\log p(\bm{y}, \bm{x},\boldsymbol{\Theta})\right] $, can be decomposed into the following components:

\begin{enumerate}
    \item 
    $$\mathbb{E}_q[\log p(\bm{y}\mid \bm{x},\{\bm{\mu}^y_{l},\bm{\Lambda}^y_{l}\}_{l=1}^L)] = \frac{1}{2}\left\{\sum_{j=1}^{J}\sum_{i=1}^{n_j}\sum_{l=1}^{L}\xi_{jil}\left(\ell_{l}^{(y,1)} + \ell_{jil}^{(y,2)} - p \log2\pi \right)\right\}.$$
    \item $$\mathbb{E}_q[\log p(\bm{x}\mid \bm{S}, \{\bm{\mu}^x_{k},\bm{\Lambda}^x_{k}\}_{k=1}^K)] = \frac{1}{2}\left\{\sum_{j=1}^{J}\sum_{k=1}^{K}\rho_{jk}\left(\ell_{k}^{(x,1)} + \ell_{jk}^{(x,2)} - q \log2\pi\right) \right\}.$$
    \item $$
    \mathbb{E}_q\left[\log p(\bm{M} \mid \bm{S}, \boldsymbol{u})\right] = \sum_{j=1}^J \sum_{i=1}^{n_j} \sum_{k=1}^{K} \sum_{l=1}^{L} \rho_{jk}\:\xi_{jil}\left\{
    g(\bar{a}_{lk},\bar{b}_{lk}) + \sum_{r<l} g(\bar{b}_{rk},\bar{a}_{rk})
    \right\}.$$
    \item $$
    \mathbb{E}_q\left[\log p(\bm{S} \mid \boldsymbol{v})\right] = 
    \sum_{j=1}^J \sum_{k=1}^{K} \rho_{jk}\left\{
    g(\bar{a}_{k},\bar{b}_{k}) + \sum_{l<k} g(\bar{b}_{l},\bar{a}_{l}))
    \right\}.$$
    \item $$\mathbb{E}_q\left[\log p(\bm{v}\mid \alpha)\right] = 
    (K-1)\left(\psi(s_1)-\log(s_2)\right) + \left\{\left(\frac{s_1}{s_2}-1\right)
    \sum_{k=1}^{K-1} g(\bar{b}_{k},\bar{a}_{k})\right\}.$$ 
    \item $$\mathbb{E}_q\left[\log p(\bm{u}\mid \beta)\right] = 
    K(L-1)\left(\psi(r_1)-\log(r_2)\right) + \left\{\left(\frac{r_1}{r_2}-1\right)
    \sum_{k=1}^{K}\sum_{l=1}^{L-1} g(\bar{b}_{lk},\bar{a}_{lk})\right\}.$$ 
    \item $$\mathbb{E}_q\left[\log p(\alpha)\right] =  \log(\mathcal{C}_{\alpha}(a_\alpha,b_\alpha)) + (a_\alpha-1)(\psi(s_1)-\log(s_2)) - b_\alpha \frac{s_1}{s_2},$$ where $\mathcal{C}_{\alpha}(\cdot)$ is the normalizing constant of a Gamma distribution.
    \item $$\mathbb{E}_q\left[\log p(\beta)\right] =  \log(\mathcal{C}_{\beta}(a_\beta,b_\beta)) + (a_\beta-1)(\psi(r_1)-\log(r_2)) - b_\beta \frac{r_1}{r_2}.$$
    \item 
    \begin{align*}
\mathbb{E}_q\!\left[\log p\!\left(\prod_{l=1}^L\{\bm{\mu}^y_l,\bm{\Lambda}^y_l\}\right)\right] 
={}& L \log \bm{\mathcal{B}}\!\left(\boldsymbol{\Psi}^y_0,\nu^y_0\right) \nonumber\\
&\quad + \tfrac{1}{2}\Big\{(\nu^y_0 - p - 1)\sum_{l=1}^L \ell^{(y,1)}_{l}\Big\} \nonumber\\
&\quad - \tfrac{1}{2}\Big\{\sum_{l=1}^L c^y_l \,\bm{\mathcal{T}}\!\big(\boldsymbol{\Psi}_0^{y^{-1}}\boldsymbol{D}^y_l\big)\Big\} \nonumber\\
&\quad + \tfrac{1}{2}\Bigg[\sum_{l=1}^L \Big\{ 
     p\log\!\Big(\tfrac{\lambda^y_0}{2\pi}\Big)
     + \ell^{(y,1)}_{l}
     - \tfrac{p\lambda^y_0}{t^y_l} \nonumber\\
&\qquad\qquad\qquad 
     - \lambda^y_0 c^y_l (\bm{m}^y_l - \boldsymbol{\mu}^y_0)^T 
       \boldsymbol{D}^y_l (\bm{m}^y_l - \boldsymbol{\mu}^y_0)
   \Big\}\Bigg],
\end{align*}
    where $\bm{\mathcal{T}}(\cdot)$ is the trace operator and $\bm{\mathcal{B}}\left(\boldsymbol{\Psi}^y_0,\nu^y_0\right)$ is the inverse of the normalizing constant of a Wishart distribution~\citep[see, for more details, Appendix B of][]{Bishop2006}.
    %
    \item \begin{align*}
\mathbb{E}_q\!\left[\log p\!\left(\prod_{k=1}^K \{\bm{\mu}^x_k,\bm{\Lambda}^x_k\}\right)\right] 
={}& K \log \bm{\mathcal{B}}\!\left(\boldsymbol{\Psi}^x_0,\nu^x_0\right) \nonumber\\
&\quad + \tfrac{1}{2}\Big\{(\nu^x_0 - q - 1)\sum_{k=1}^K \ell^{(x,1)}_{k}\Big\} \nonumber\\
&\quad - \tfrac{1}{2}\Big\{\sum_{k=1}^K c^x_k \,\bm{\mathcal{T}}\!\big(\boldsymbol{\Psi}_0^{x^{-1}}\boldsymbol{D}^x_k\big)\Big\} \nonumber\\
&\quad + \tfrac{1}{2}\Bigg[\sum_{k=1}^K \Big\{
     q \log\!\Big(\tfrac{\lambda^x_0}{2\pi}\Big)
     + \ell^{(x,1)}_{k}
     - \tfrac{q \lambda^x_0}{t^x_k} \nonumber\\
&\qquad\qquad\qquad 
     - \lambda^x_0 c^x_k (\bm{m}^x_k - \boldsymbol{\mu}^x_0)^T
       \boldsymbol{D}^x_k (\bm{m}^x_k - \boldsymbol{\mu}^x_0)
   \Big\}\Bigg].
\end{align*}
\end{enumerate}

The second term is decomposed into the following six components:

\begin{enumerate}
    \item $$\mathbb{E}_q\left[\log q(\bm{S})\right] = \sum_{j=1}^J\sum_{k=1}^K  \rho_{jk}\log(\rho_{jk}).$$    
    \item $$\mathbb{E}_q\left[\log q(\bm{M})\right] = \sum_{j=1}^J \sum_{i=1}^{n_j} \sum_{l=1}^L \xi_{jil}\log(\xi_{jil}).$$
    \item $$\mathbb{E}_q\left[\log q(\bm{v})\right] =\sum_{k=1}^{K-1}  \{\log(\mathcal{C}_{\bm{v}}(\bar{a}_{k},\bar{b}_{k})) + (\bar{a}_{k}-1)g(\bar{a}_{k},\bar{b}_{k}) + (\bar{b}_{k}-1)g(\bar{b}_{k},\bar{a}_{k})\},$$ where $\mathcal{C}_{\bm{v}}(\cdot)$ is the normalizing constant of a Beta distribution.
    \item $$\mathbb{E}_q\left[\log q(\bm{u})\right] =\sum_{k=1}^K \sum_{l=1}^{L-1}  \{\log(\mathcal{C}_{\bm{u}}(\bar{a}_{lk},\bar{b}_{lk}))+
    (\bar{a}_{lk}-1)g(\bar{a}_{lk},\bar{b}_{lk})+ (\bar{b}_{lk}-1)g(\bar{b}_{lk},\bar{a}_{lk})\}.$$
    \item $$\mathbb{E}_q\left[\log q(\prod_{k=1}^K\{\bm{\mu}^x_k,\bm{\Lambda}^x_k\})\right] = \sum_{k=1}^K \left[\frac12\:\ell^{(x,1)}_{k} + \frac12\:q\left\{\log(t^x_k/2\pi)-1\right\} - \bm{\mathcal{H}}\left(q(\bm{\Lambda}^x_k)\right)\right],$$
    where $\bm{\mathcal{H}}\left(q(\bm{\Lambda}^x_k)\right) = -\log(\bm{\mathcal{B}}(D^x_k, c^x_k)) - 0.5\:(c^x_k - q - 1)\:\ell_k^{(x, 1)} + 0.5 \: c^x_k\:q$ is the entropy of a Wishart distribution.
    \item $$\mathbb{E}_q\left[\log q(\prod_{l=1}^L\{\bm{\mu}^y_l,\bm{\Lambda}^y_l\})\right] = \sum_{l=1}^L \left[\frac12\:\ell^{(y,1)}_{l} + \frac12\:p\left\{\log(t^y_l/2\pi)-1\right\} - \bm{\mathcal{H}}\left(q(\bm{\Lambda}^y_l)\right)\right],$$
    where $\bm{\mathcal{H}}\left(q(\bm{\Lambda}^y_l)\right) = -\log(\bm{\mathcal{B}}(D^y_l, c^y_l)) - 0.5\:(c^y_l - p - 1)\:\ell_l^{(y, 1)} + 0.5 \: c^y_l\:p.$
    \item $$\mathbb{E}_q\left[\log(q(\alpha))\right] = \log(\mathcal{C}_{\alpha}(s_1,s_2)) + (s_1-1)(\psi(s_1)-\log(s_2)) - s_1,$$ where $\mathcal{C}_{\alpha}(\cdot)$ is the normalizing constant of a Gamma distribution.
    \item $$\mathbb{E}_q\left[\log(q(\beta))\right] = \log(\mathcal{C}_{\beta}(r_1,r_2)) + (r_1-1)(\psi(r_1)-\log(r_2)) - r_1.$$
\end{enumerate}

\section{Simulations}\label{sec::suppl_simulations}
In this section, we provide the plots from the simulation studies, reported in the main manuscript. Particularly, Figures \ref{fig:NAM_CAM_fiSAN_OC_2D} - \ref{fig:NAM_CAM_fiSAN_OC_10D} display the boxplots of observational clustering (OC) accuracy for each group separately, computed across 50 independent replications for all the methods and dimensions of group- and observation-level variables.

\begin{figure}[!htp]
\centering
\includegraphics[width=1\columnwidth]{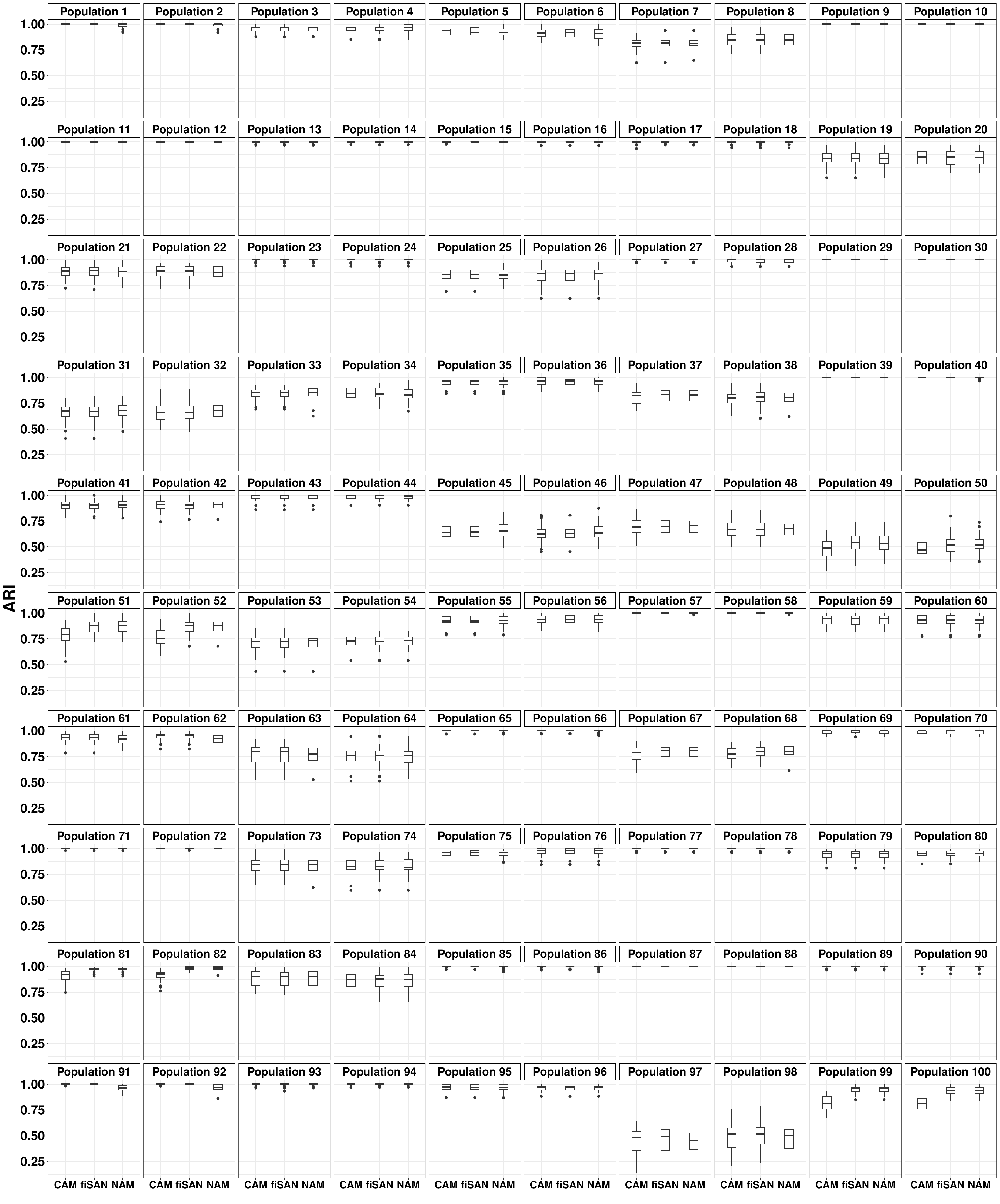}
\caption{Comparison of observational clustering accuracy of NAM with CAM and fiSAN when both the group- and observation-level variables are two-dimensional. Boxplots are reported over 50 independent replications.
}
\label{fig:NAM_CAM_fiSAN_OC_2D}
\end{figure}

\begin{figure}[!htp]
\centering
\includegraphics[width=1\columnwidth]{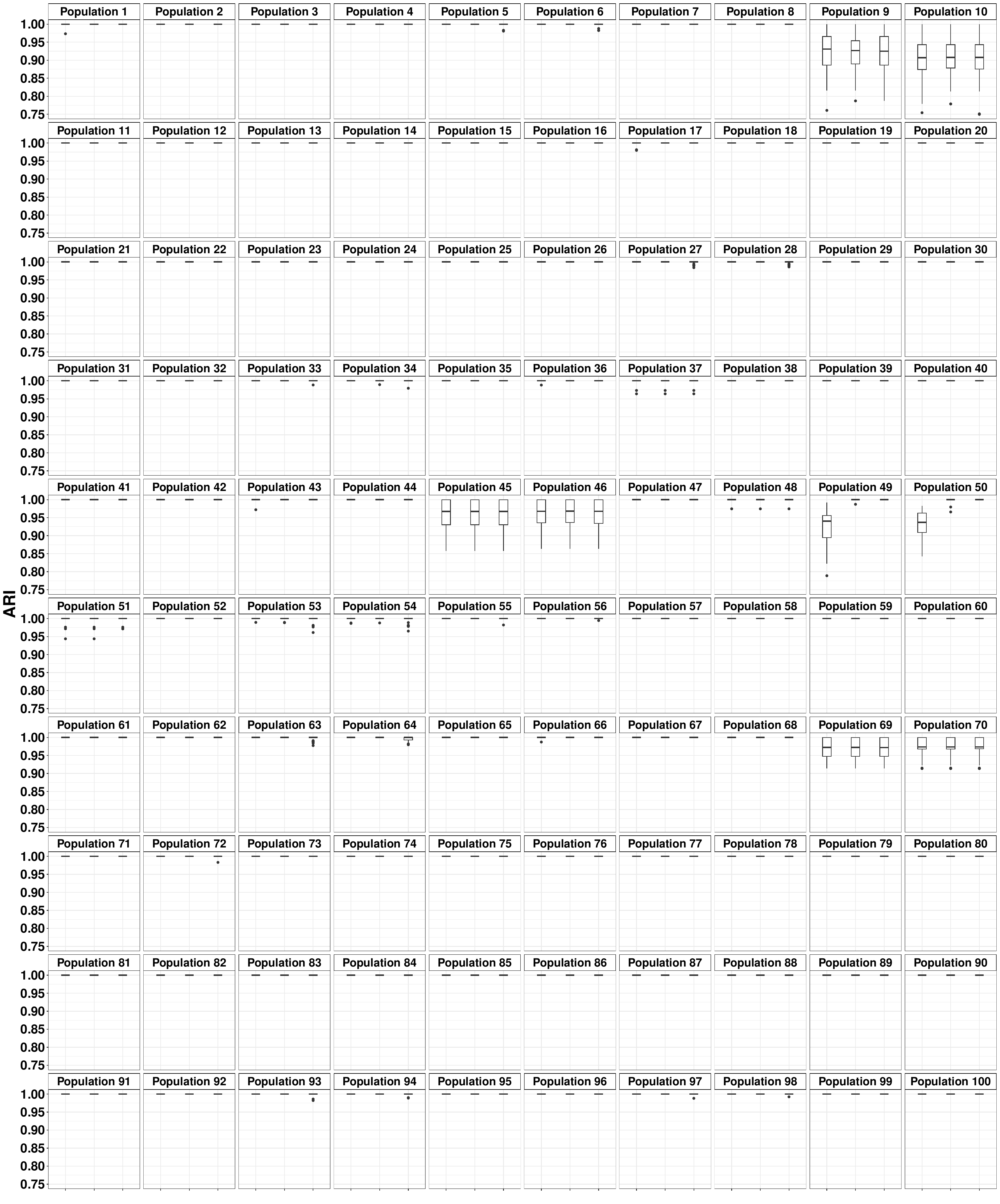}
\caption{Comparison of observational clustering accuracy of NAM with CAM and fiSAN when both the group- and observation-level variables are five-dimensional. Boxplots are reported over 50 independent replications.
}
\label{fig:NAM_CAM_fiSAN_OC_5D}
\end{figure}

\begin{figure}[!htp]
\centering
\includegraphics[width=1\columnwidth]{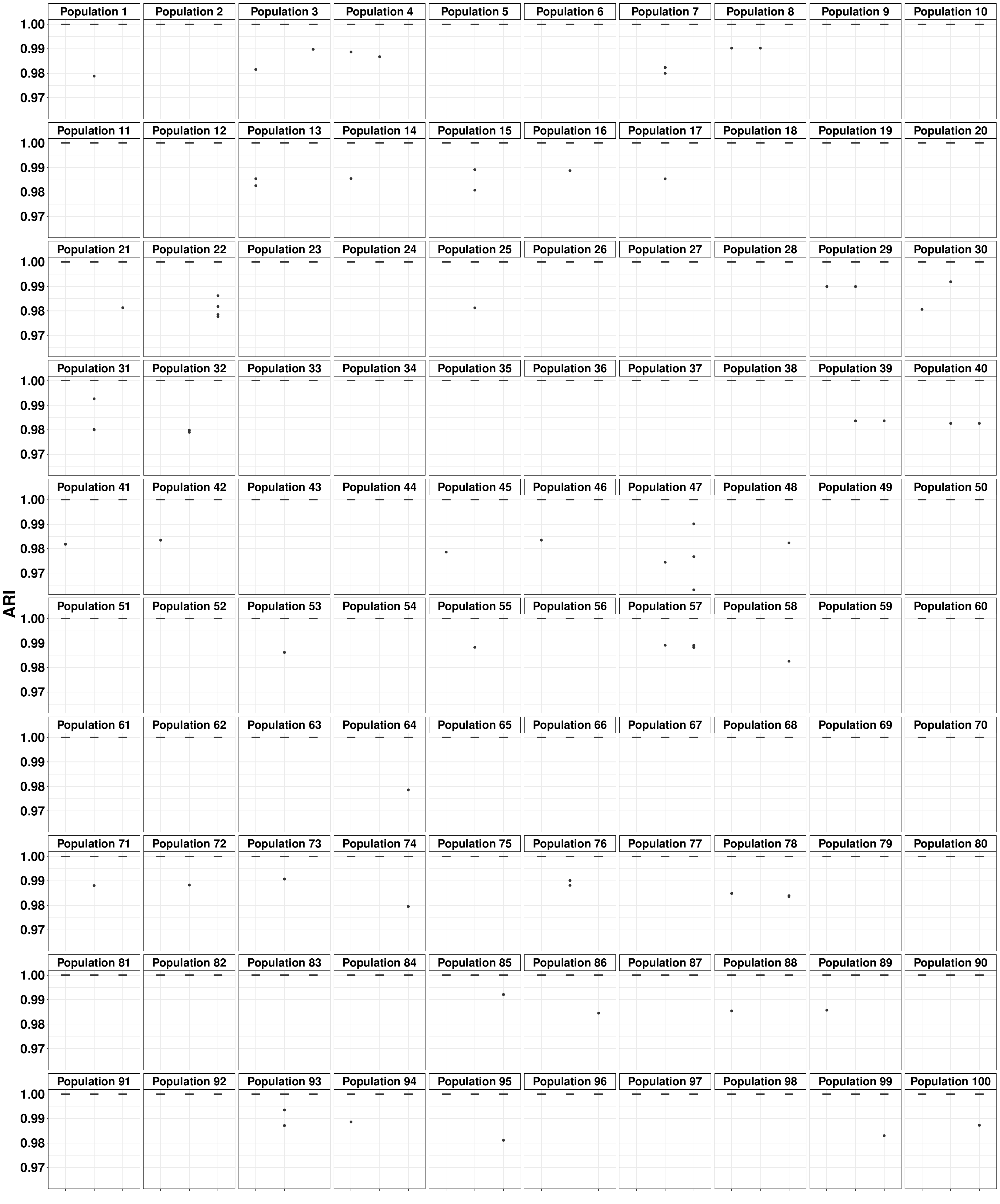}
\caption{Comparison of observational clustering accuracy of NAM with CAM and fiSAN when both the group- and observation-level variables are ten-dimensional. Boxplots are reported over 50 independent replications.
}
\label{fig:NAM_CAM_fiSAN_OC_10D}
\end{figure}
\clearpage
\newpage 
Figure~\ref{fig:All_GC_OC} provides a summary of the distribution of adjusted Rand index (ARI; \citealp{ARI}) values for estimating group clusters (GCs) and OCs across different combinations of mixture distributions used to generate the group-level and observation-level variables.
\begin{figure}[!htp]
        \centering
        \begin{subfigure}{0.99\textwidth}
          \centering
            \includegraphics[width= 1\linewidth]{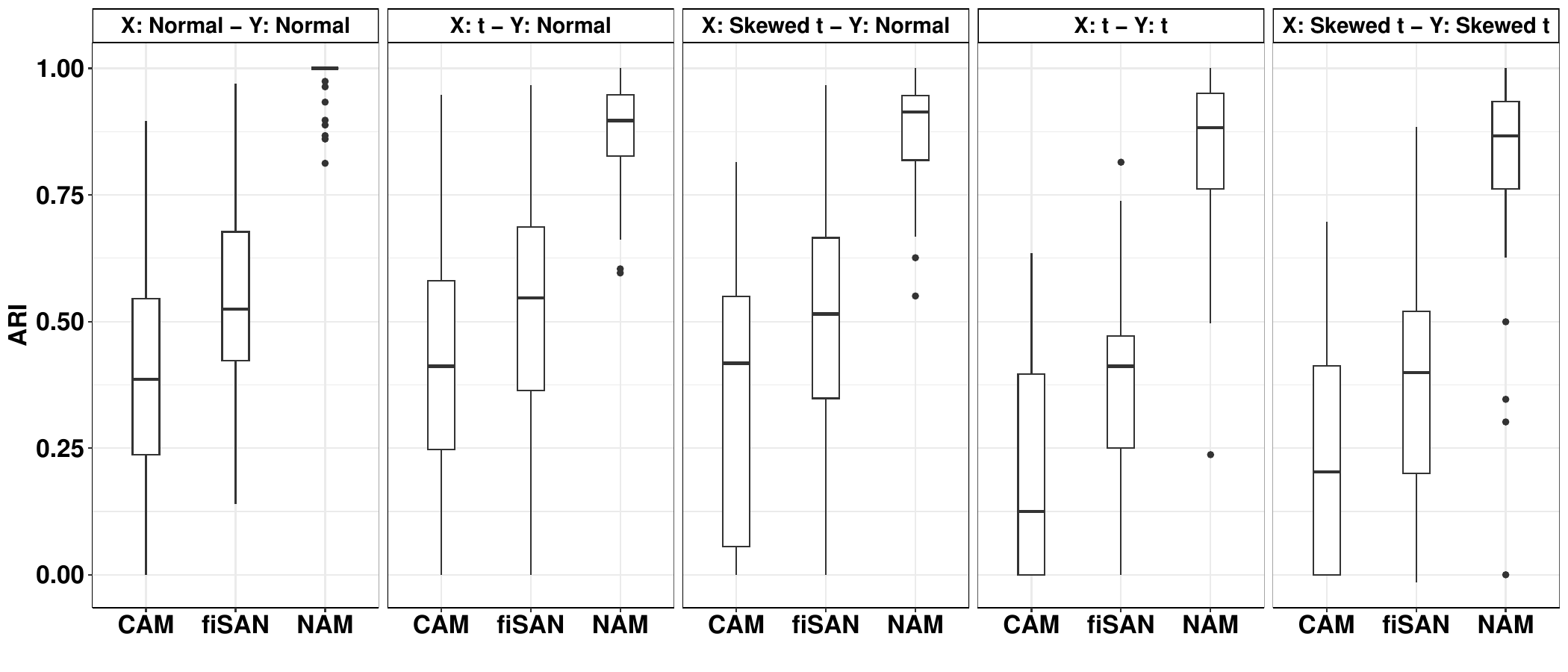}
            \caption{Group clustering}
            \label{fig:All_GC}
        \end{subfigure}
        \begin{subfigure}{0.99\textwidth}
            \centering
            \includegraphics[width= 1\linewidth]{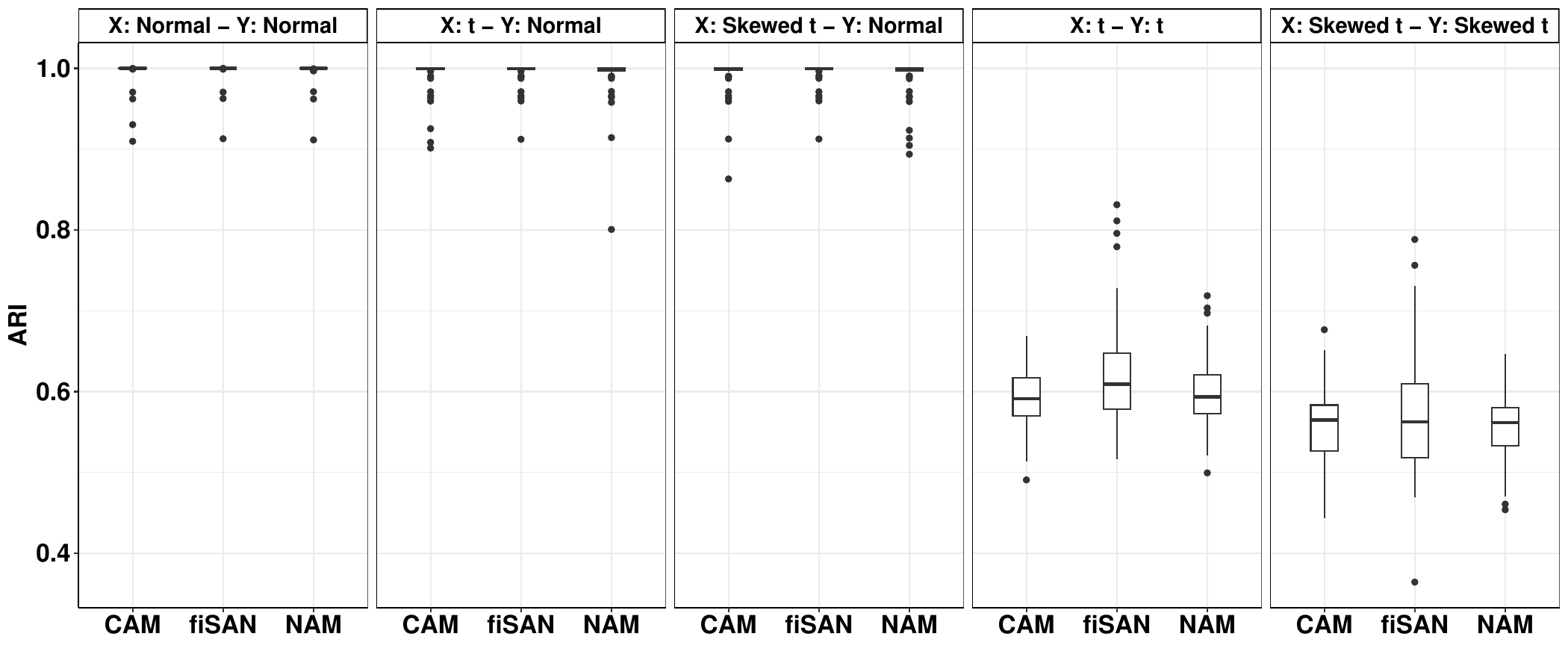}
            \caption{Observational clustering}
            \label{fig:All_OC}
        \end{subfigure}
        \caption{Comparison of clustering accuracy of NAM with CAM and fiSAN for (a) Group clusters and (b) Observational clusters. The top panel in the figures denote the distribution used to generate the data. For example, $X$: t - $Y$: Normal refers to the case when the group-level variables $\bm{x}_j$ was drawn from a mixture of multivariate t distribution while the observation-level variables $\bm{y}_{ji}$ were drawn from a mixture of multivariate Gaussian distribution, etc.}
        \label{fig:All_GC_OC}
        \end{figure}
        
In the main manuscript, we conducted simulation studies to assess the robustness of NAM to deviations from Gaussianity in the group-level variables ($\bm{x}_j$) and/or the observation-level variables ($\bm{y}_{ji}$). Specifically, we compared the performance of NAM with CAM and finite-infinite shared atoms nested model (fiSAN; \citealp{fiSAN}) by reporting the distributions of ARI values for estimating GCs and OCs across different combinations of mixture distributions used to generate the group- and observation-level variables. Here, we further examine the estimated numbers of GCs and OCs, as shown in Figure~\ref{fig:Num_All_GC_OC}. Figure~\ref{fig:Num_All_GC} indicates that NAM accurately recovers the true number of group clusters when both group- and observation-level variables are generated from Gaussian mixtures, while slightly overestimating the number of clusters when one or both levels follow non-Gaussian mixture distributions. In contrast, the competing methods more frequently fail to recover the true number of GCs. Figure~\ref{fig:Num_All_OC} shows that all three methods correctly estimate the true number of observational clusters when the observation-level variables are Gaussian, irrespective of the group-level distribution. When the observation-level variables are generated from non-Gaussian mixtures, all methods tend to overestimate the number of OCs, although NAM yields relatively fewer spurious clusters than the competing approaches.

\begin{figure}[!htp]
        \centering
        \begin{subfigure}{0.99\textwidth}
          \centering
            \includegraphics[width= 1\linewidth]{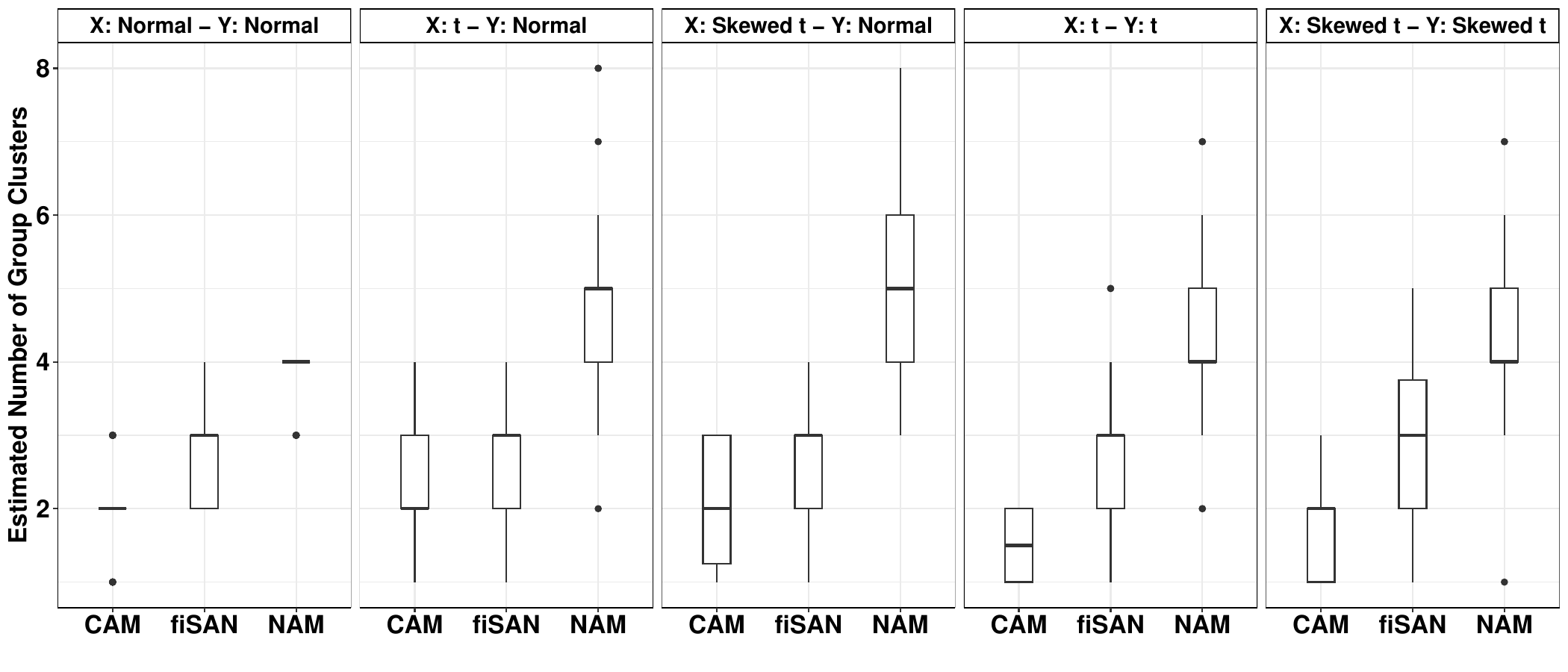}
            \caption{Group clustering}
            \label{fig:Num_All_GC}
        \end{subfigure}
        \begin{subfigure}{0.99\textwidth}
            \centering
            \includegraphics[width= 1\linewidth]{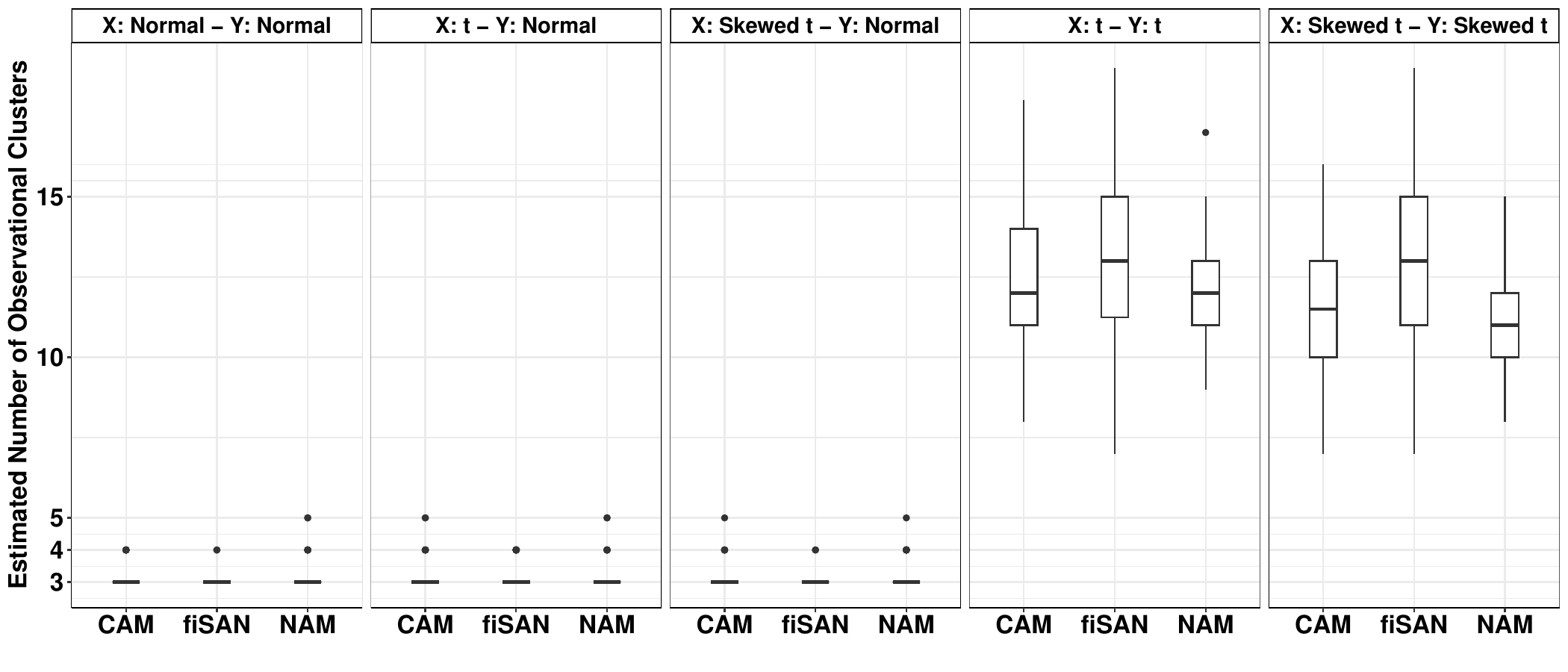}
            \caption{Observational clustering}
            \label{fig:Num_All_OC}
        \end{subfigure}
        \caption{Comparison of number of estimated clusters for NAM, CAM, and fiSAN for (a) Group clusters and (b) Observational clusters. The top panel in the figures denote the distribution used to generate the data. For example, $X$: t - $Y$: Normal refers to the case when the group-level variables $\bm{x}_j$ was drawn from a mixture of multivariate t distribution while the observation-level variables $\bm{y}_{ji}$ were drawn from a mixture of multivariate Gaussian distribution, etc.}
        \label{fig:Num_All_GC_OC}
        \end{figure}

\clearpage
\newpage
Figure~\ref{fig:GC_OC_ARI_Missing} summarizes the distributions of ARI values for recovering both GCs and OCs under varying levels of randomly omitted group-level variables ($r$).
        \begin{figure}[!htp]
        \centering
        \begin{subfigure}{0.99\textwidth}
          \centering
            \includegraphics[width= 1\linewidth]{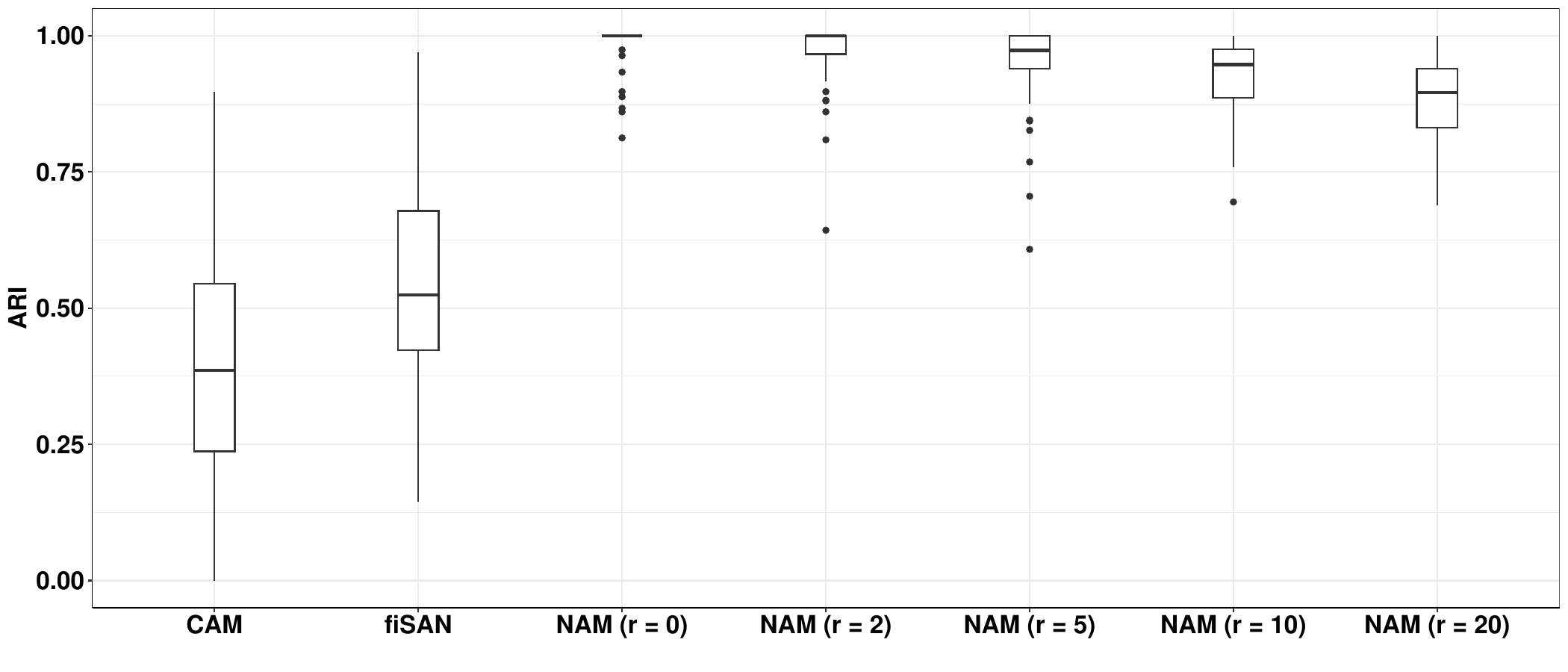}
            \caption{Group clustering}
            \label{fig:GC_ARI_Missing}
        \end{subfigure}
        \begin{subfigure}{0.99\textwidth}
            \centering
            \includegraphics[width= 1\linewidth]{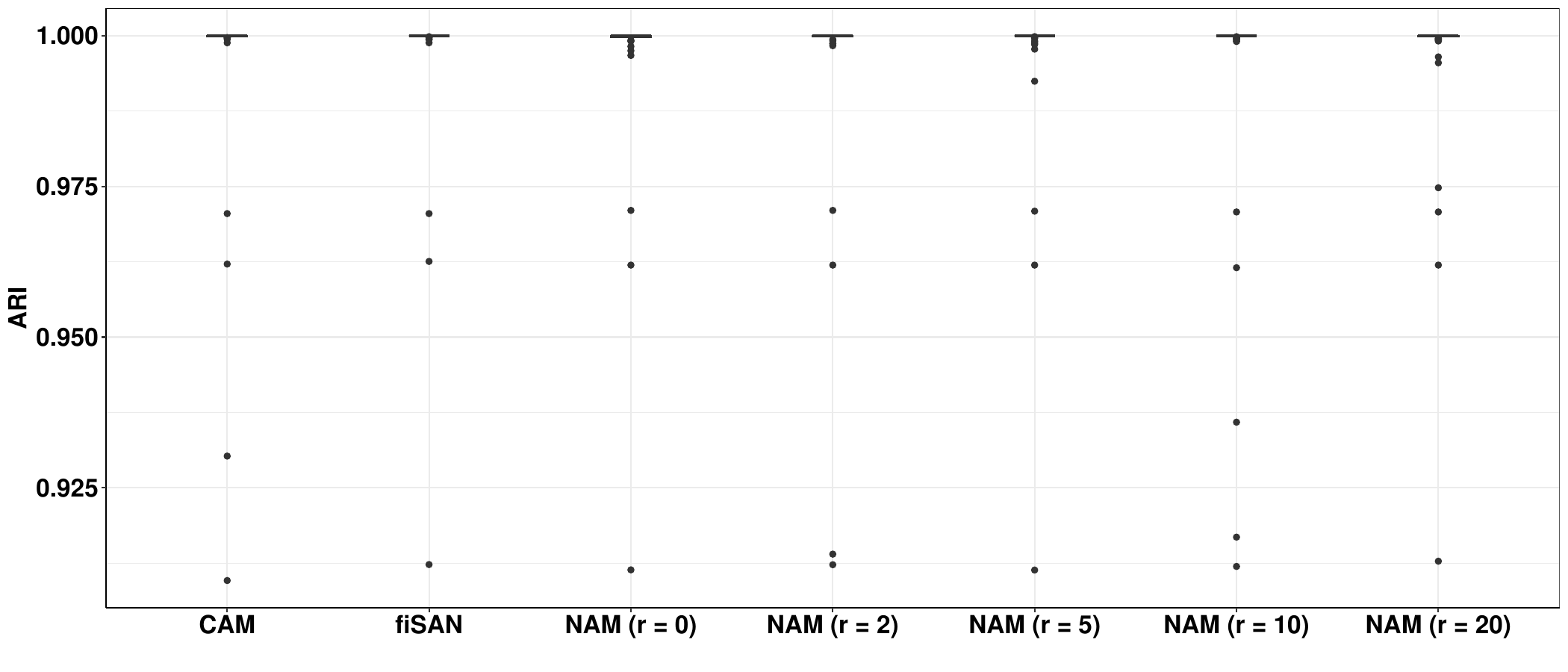}
            \caption{Observational clustering}
            \label{fig:OC_ARI_Missing}
        \end{subfigure}
        \caption{Comparison of clustering accuracy of NAM with CAM and fiSAN for (a) Group clusters and (b) Observational clusters under varying levels of randomly omitted group-level variables ($r$). The case $r = 0$ corresponds to the full-information setting.}
        \label{fig:GC_OC_ARI_Missing}
        \end{figure}        
\section{Real Data Analysis}\label{sec::suppl_rda}
\subsection{Diagnostics}\label{sec::suppl_rda:diagnostics}
In the main manuscript, for the real data analysis, we first applied principal component analysis (PCA) to the SNP data and retained the top five principal components (PCs). Similarly, we performed PCA on the gene expression data across all 961 individuals and 929,817 single cells, retaining the top five PCs. The resulting PCs were subsequently standardized by centering and scaling.

Prior to applying our proposed model, we conducted basic diagnostic checks on these standardized PCs. Figure~\ref{fig:X_Hist} presents the histograms for the SNP-derived PCs, which clearly suggest the presence of multiple distributional components. This observation supports the suitability of a mixture model with possible Gaussian components. Additionally, we fitted a Gaussian mixture model using the \texttt{Mclust} function from the R package \texttt{mclust} \citep{GMM_R} to the SNP-derived PC scores. Figure~\ref{fig:X_QQ} presents dimension-wise quantile–quantile plots comparing the empirical quantiles of the observed PC scores with the corresponding theoretical quantiles implied by the fitted mixture model. Overall, the empirical quantiles align reasonably well with the theoretical quantiles across dimensions, suggesting that a Gaussian mixture specification provides an adequate representation of the marginal distributional features of the SNP-derived PC data.

        \begin{figure}[!htp]
        \centering
           \begin{subfigure}{0.95\textwidth}
            \centering
            \includegraphics[width= 0.9\linewidth]{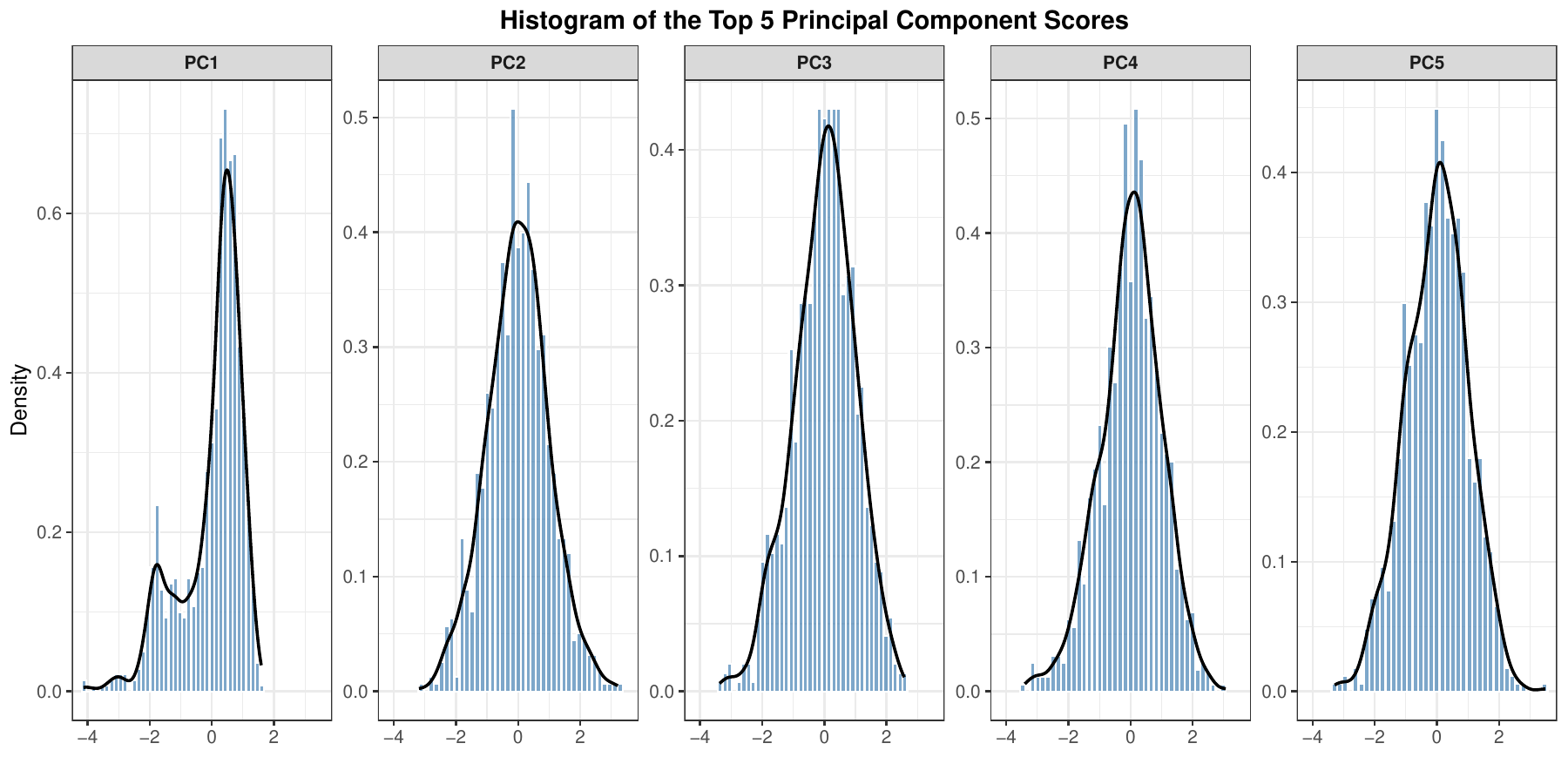}
            \caption{}
            \label{fig:X_Hist}
        \end{subfigure}
            \begin{subfigure}{0.95\textwidth}
              \centering
                \includegraphics[width= 0.9\linewidth]{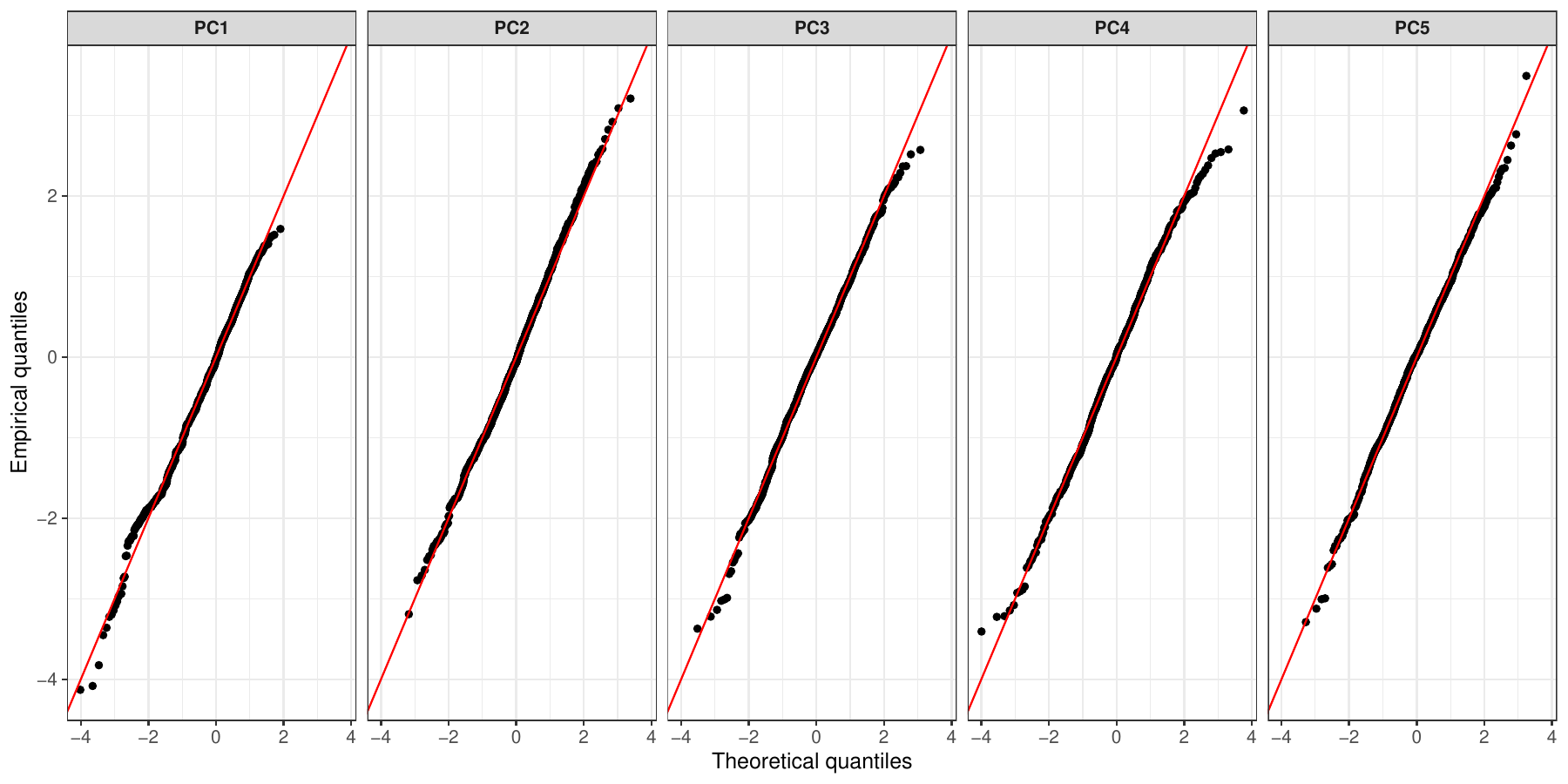}
                \caption{}
                \label{fig:X_QQ}
            \end{subfigure}
        \caption{Dimension-wise (a) histogram overlayed with density corresponding to the PCs obtained from the SNP data and (b) quantile-quantile plot obatained from fitted Gaussian mixture model. }
        \label{fig:X_QQ_Hist}
        \end{figure}

For the gene expression data, we randomly selected three individuals and performed analogous diagnostic checks. Figure~\ref{fig:Y_Hist} displays the histograms for the PCs obtained from each of these individuals. As before, Figure~\ref{fig:Y_Hist} reveals multimodality, again motivating the use of mixture models with possible Gaussian components. Similarly, Figure~\ref{fig:Y_QQ} displays the dimension-wise quantile–quantile plots for each of the three individuals, comparing the empirical quantiles with those implied by the corresponding fitted Gaussian mixture models. Across individuals and dimensions, the empirical quantiles exhibit overall agreement with the theoretical quantiles. While minor deviations are observable, mostly in PC dimension 5, the plots do not indicate substantial lack of fit, thereby supporting the suitability of the Gaussian mixture modeling assumption for the gene expression data as well.

        \begin{figure}[!htp]
        \centering
        \begin{subfigure}{0.95\textwidth}
            \centering
            \includegraphics[width= 0.9\linewidth]{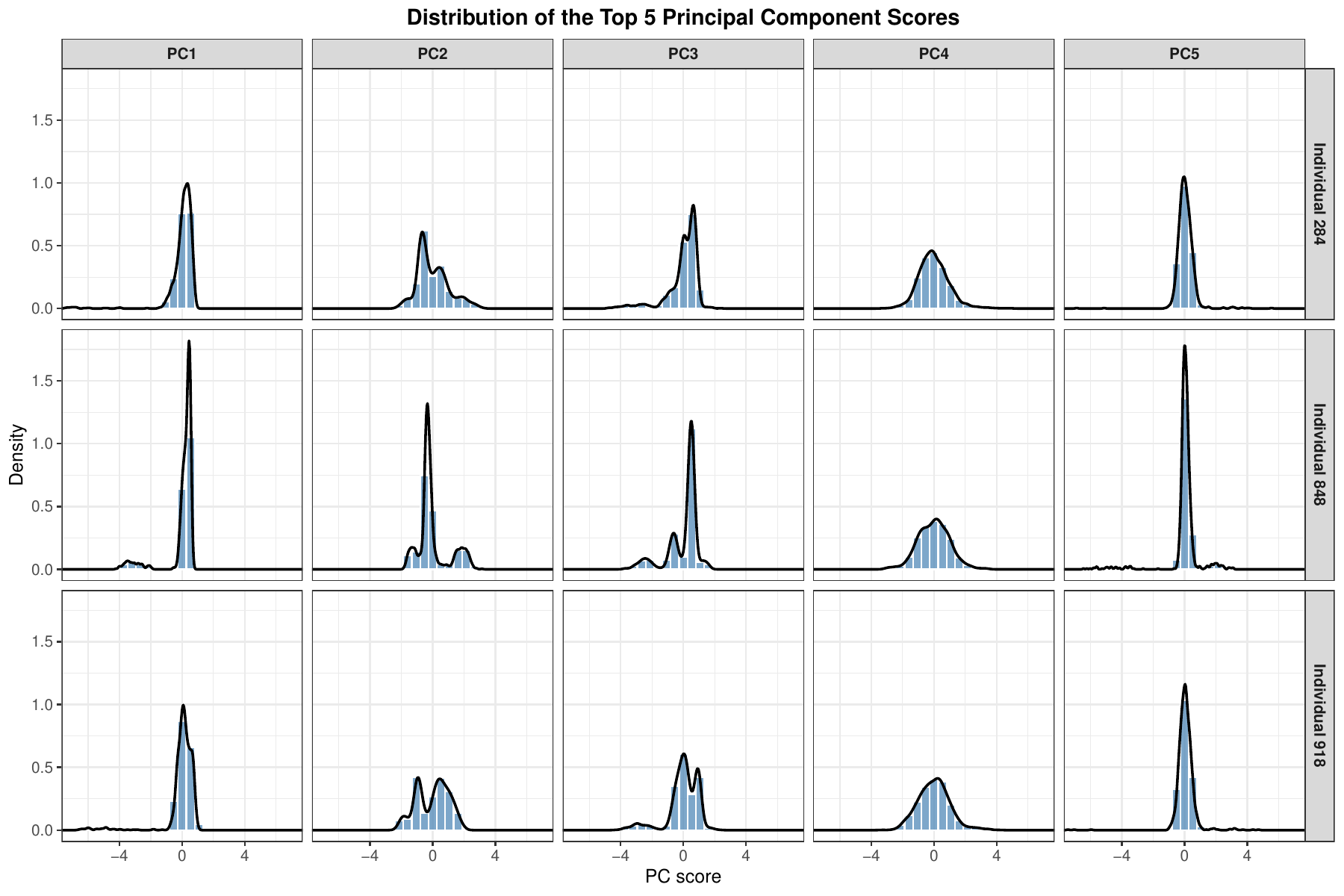}
            \caption{}
            \label{fig:Y_Hist}
        \end{subfigure}
        \begin{subfigure}{0.95\textwidth}
          \centering
            \includegraphics[width= 0.9\linewidth]{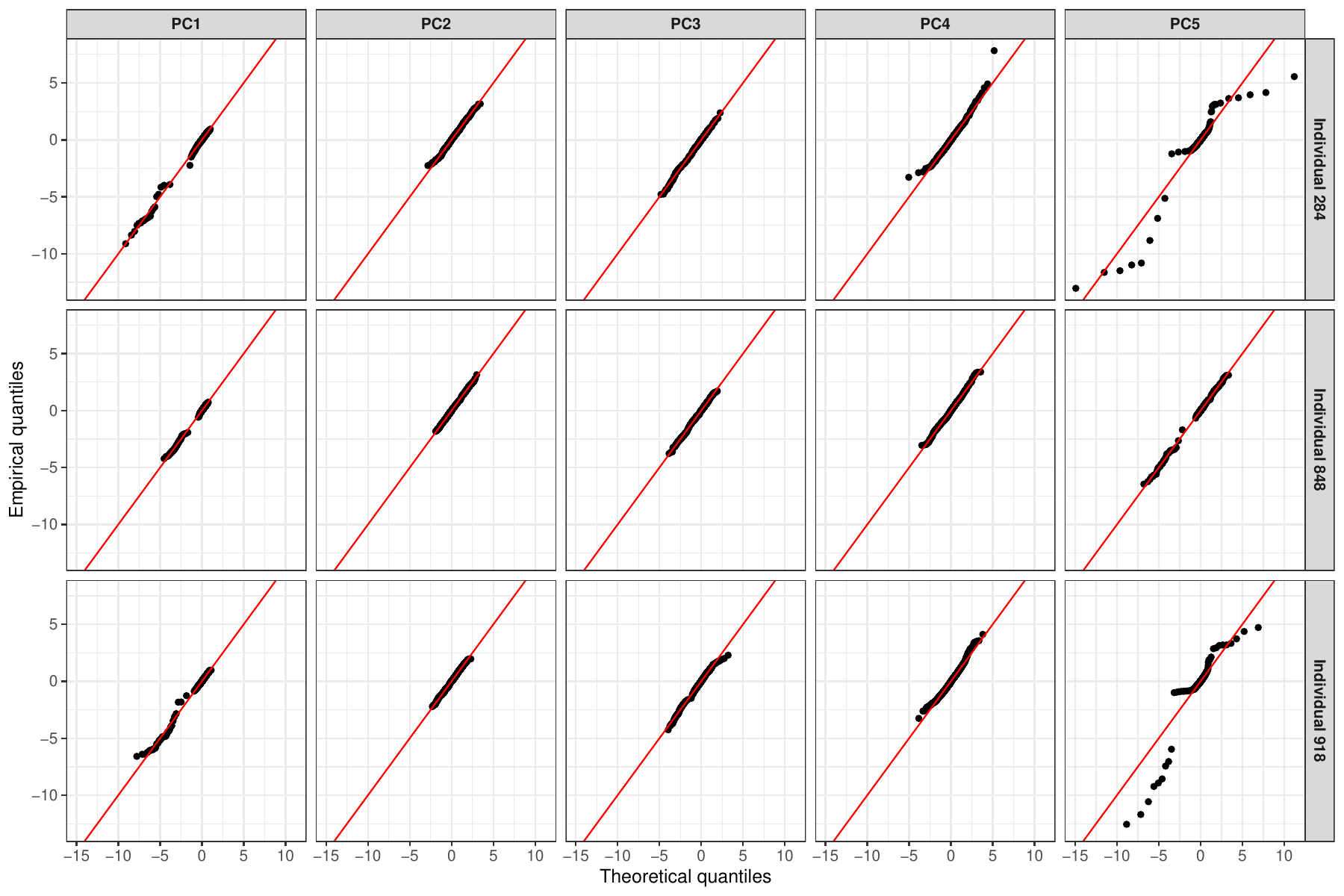}
            \caption{}
            \label{fig:Y_QQ}
        \end{subfigure}
        \caption{Dimension-wise (a) histogram overlayed with density corresponding to the PCs obtained from the gene expression data and (b) quantile-quantile plot obatained from fitted Gaussian mixture model..}
        \label{fig:Y_QQ_Hist}
        \end{figure}

\clearpage\newpage        
\subsection{Sensitivity}\label{sec::suppl_rda:sensitivity}
In the main manuscript, for the real data analysis, we set the truncation levels to $K = L = 50$ and performed 30 independent runs of the proposed VI algorithm using different initializations. Among these runs, the solution corresponding to the highest ELBO value at convergence was used for inference. In this section, we examined whether the truncation boundaries were active by tracking the numbers of GCs and OCs across iterations and across the multiple runs of the VI algorithm. Figure~\ref{fig:Num_GC_OC_Combined} presents boxplots of the estimated numbers of GCs and OCs aggregated over all iterations and across 30 distinct runs of the VI algorithm. Across all runs and iterations, the maximal numbers of GCs and OCs remain well below the truncation limits, suggesting that the chosen truncation levels are adequate for posterior inference.
    \begin{figure}[!htp]
    \centering
    \begin{subfigure}{0.45\textwidth}
      \centering
        \includegraphics[width= 1\linewidth]{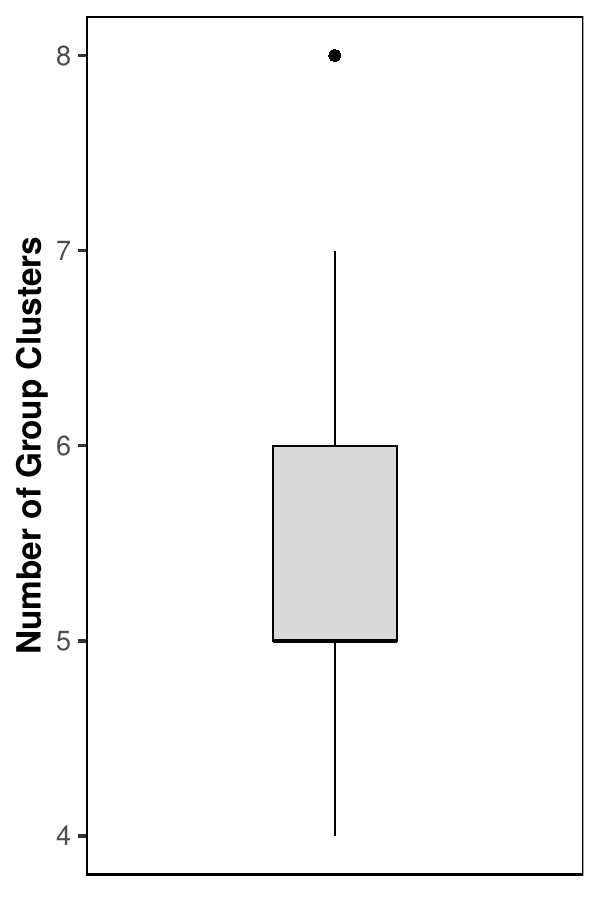}
        \caption{Group clustering}
        \label{fig:NUM_GC_Combined}
    \end{subfigure}
    \begin{subfigure}{0.45\textwidth}
        \centering
        \includegraphics[width= 1\linewidth]{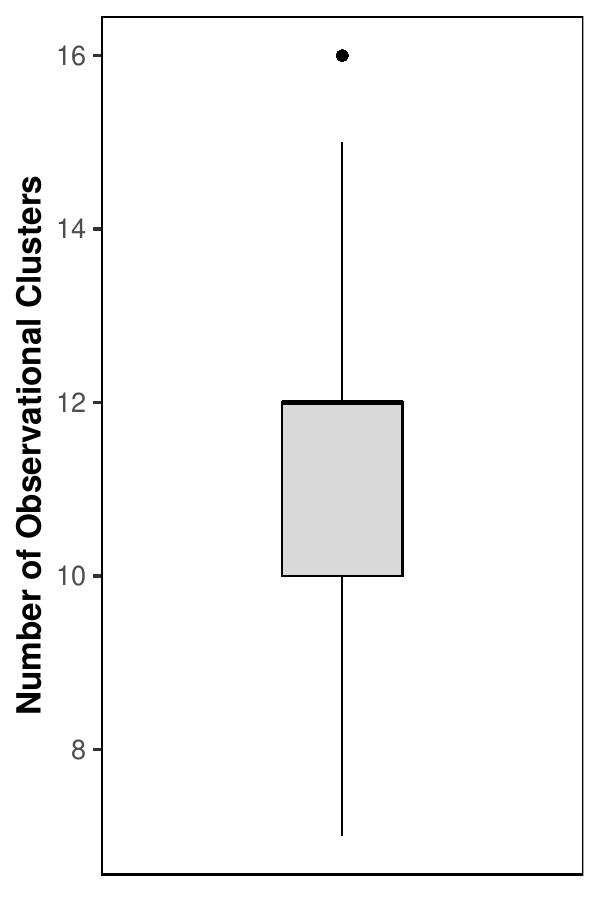}
        \caption{Observational clustering}
        \label{fig:NUM_OC_Combined}
    \end{subfigure}
    \caption{Boxplot of number of estimated (a) group clusters and (b) observational clusters across all iterations and across 30 distinct runs of our algorithm for the OneK1K data set.}
    \label{fig:Num_GC_OC_Combined}
    \end{figure}

\clearpage\newpage
\subsection{Additional Plots}\label{sec::suppl_rda:add_plots}
Figure~\ref{fig:Indiv131_True_Est_OCs} presents the uniform manifold approximation and projection (UMAP; \citealp{McInnes2018}) embeddings of the gene expression data for the individual numbered 131, with cells colored according to the estimated OCs alongside the annotated cell-type labels by OneK1K \citep{OneK1K_Yazar}.
    \begin{figure}[!htp]
        \centering
           \begin{subfigure}{0.65\textwidth}
            \centering
            \includegraphics[width= 1\linewidth]{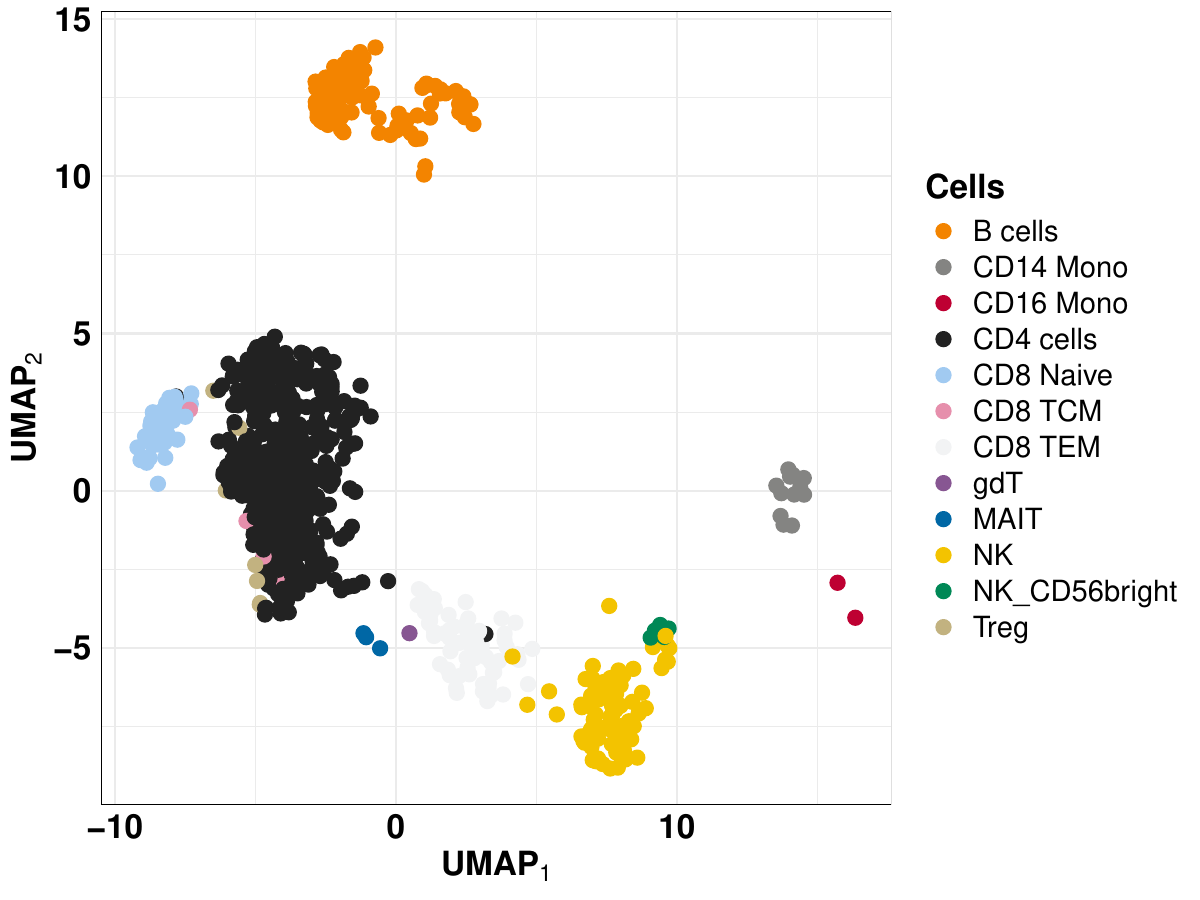}
            \caption{}
            \label{fig:Indiv131_TrueOCs}
        \end{subfigure}
            \begin{subfigure}{0.65\textwidth}
              \centering
                \includegraphics[width= 1\linewidth]{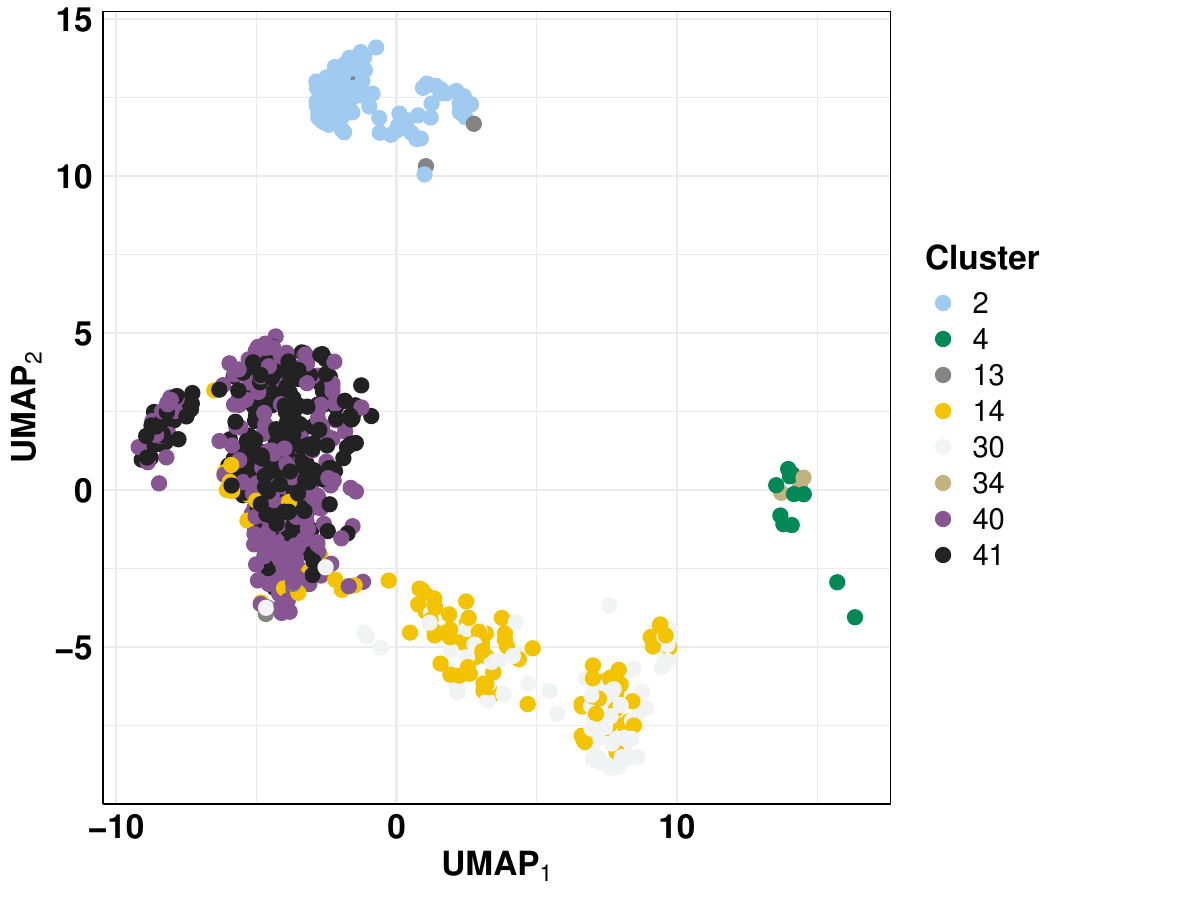}
                \caption{}
                \label{fig:Indiv131_EstimatedOCs}
            \end{subfigure}
        \caption{UMAP embeddings of the gene expression data for the individual labeled 131, colored by (a) manually annotated cell clusters and (b) estimated OCs obtained from the NAM.}
        \label{fig:Indiv131_True_Est_OCs}
        \end{figure}

\clearpage
\newpage

\section{Glossary}\label{supp::sec::glossary}
In this section, we provide a glossary of all abbreviations used throughout the paper.
\begin{table}[H]
                \centering
                \begin{tabular}{ll}
                \hline
                \textbf{Acronym} & \textbf{Description} \\
                \hline
                NAM & Nested Atoms Model \\
                DP & Dirichlet Process \\
                HDP & Hierarchical Dirichlet Process \\
                nDP & Nested Dirichlet Process \\
                CAM & Common Atoms Model \\
                fiSAN & Finite–Infinite Shared Atoms Nested Model \\
                HHDP & Hidden Hierarchical Dirichlet Process \\
                iid & Independent and Identically Distributed\\
                GEM & Griffiths–Engen–McCloskey distribution \\
                GC & Group Cluster \\
                OC & Observational Cluster \\
                VI & Variational Inference \\
                CAVI & Coordinate Ascent Variational Inference \\
                ELBO & Evidence Lower Bound \\
                ARI & Adjusted Rand Index \\
                PCA & Principal Component Analysis \\
                PC & Principal Component \\
                scRNA-seq & Single-cell RNA Sequencing \\
                SNP & Single Nucleotide Polymorphism \\
                NW & Normal–Wishart distribution \\
                UMAP & Uniform Manifold Approximation and Projection\\
                \hline
                \end{tabular}
                \caption{List of acronyms used throughout the paper.}
                \label{tab:acronyms}
            \end{table}
\clearpage

\end{document}